\documentclass{LMCS}

\def\dOi{10(3:10)2014}
\lmcsheading%
{\dOi}
{1--46}
{}
{}
{Jul.~\phantom05, 2013}
{Aug.~22, 2014}
{}

\usepackage{amsmath}
\usepackage{amssymb}
\usepackage{amsthm}
\usepackage{stmaryrd}
\usepackage{url}
\usepackage{code,hyperref}
\usepackage{mathrsfs}
\DeclareMathAlphabet{\mathpzc}{OT1}{pzc}{m}{it}
\let\mathpzc\mathscr
\let\mathpzc\mathcal

\def\bondi{{\bf bondi}}
\def\ltrans{[\![}
\def\rtrans{]\!]}

\usepackage{prooftree}
\usepackage{macros}

\title[A Concurrent Pattern Calculus]{A Concurrent Pattern Calculus\rsuper*}

\titlecomment{{\lsuper*}This paper significantly builds upon
``Concurrent Pattern Calculus'' that first introduced the titled work
\cite{GivenWilsonGorlaJay10}.
Moreover, the behavioural theory has been presented in
\cite{GivenWilsonPHD,GivenWilsonGorla13}.}

\author[T.~Given-Wilson]{Thomas Given-Wilson\rsuper a}	
\address{{\lsuper a}INRIA, Paris, France}	
\email{thomas.given-wilson@inria.fr}  
\thanks{{\lsuper a}The first author has been partially supported by the project ANR-12-IS02-001 PACE}

\author[D.~Gorla]{Daniele Gorla\rsuper b}	
\address{{\lsuper b}Dip.~di Informatica, ``Sapienza'' Universit\`a di Roma}	
\email{gorla@di.uniroma1.it}

\author[B.~Jay]{Barry Jay\rsuper c}	
\address{{\lsuper c}University of Technology, Sydney}	
\email{Barry.Jay@uts.edu.au}


\ACMCCS{[{\bf Theory of computation}]: Models of computation---Concurrency---Process calculi}

\renewcommand{\beq}{\simeq}
\renewcommand{\rew}{\redar}
\newcommand{\xangle}[1]{\langle #1 \rangle}

\newcommand{\withsetnot}[2]{#2} 

\makeatletter
\def \rightarrowfill{\m@th\mathord{\smash-}\mkern-6mu%
  \cleaders\hbox{$\mkern-2mu\mathord{\smash-}\mkern-2mu$}\hfill
  \mkern-6mu\mathord\to}
\makeatother
\makeatletter
\def \Rightarrowfill{\m@th\mathord{\smash=}\mkern-6mu%
  \cleaders\hbox{$\mkern-2mu\mathord{\smash=}\mkern-2mu$}\hfill
  \mkern-6mu\mathord\Rightarrow}
\makeatother
\def \overstackrel#1#2{\mathrel{\mathop{#1}\limits^{#2}}}

\begin{document}
\makeatactive


\maketitle  

\begin{abstract}
  Concurrent pattern calculus (CPC) drives interaction between processes by
  comparing data structures, just as sequential pattern calculus
  drives computation. By generalising from pattern matching to pattern
  unification, interaction becomes symmetrical, with information
  flowing in both directions.
  CPC provides a natural language to express trade where information
  exchange is pivotal to interaction.
  The unification allows some patterns to be more discriminating than
  others; hence, the behavioural theory must take this aspect into
  account, so that bisimulation becomes subject to compatibility of patterns.
  Many popular process calculi can be encoded in CPC; this allows for a
  gain in expressiveness, formalised through encodings.

\end{abstract} 

\section{Introduction}
\label{sec:intro}


The $\pi$-calculus \cite{milner.parrow.ea:calculus-mobile,sangiorgi.walker:theory-mobile} 
holds an honoured position amongst process calculi as a succinct calculus that can capture topological changes in a network, as well as encode computation as represented by $\l$-calculus \cite{Barendregt85}.
Interaction in $\pi$-calculus is done by matching upon a single name known by both the input and output primitives.
The polyadic $\pi$-calculus extends this by also matching on the length of the tuple of names to be communicated.
Linda \cite{Gel85} extends this further by allowing matching on any number of names known by both processes.
A more symmetric approach to communication is taken in Fusion calculus \cite{parrow.victor:fusion-calculus} 
that matches a channel name and tuple length, like polyadic $\pi$-calculus, but allows symmetric information exchange.
Other calculi consider structured information rather than simply names \cite{gordon1997ccp},
even matching arbitrary structures asymmetrically during communication \cite{BJPV11}.
Hence it is natural to explore how a concurrent pattern calculus can unify structured patterns with symmetric matching and communication mechanisms.

This paper develops pattern unification in a setting that supports
parallel composition, name restriction, and replication. This yields
{\em concurrent pattern calculus} (CPC), where prefixes for input and
output are generalised to patterns whose {\em unification} triggers a
symmetric flow of information, as represented by the sole interaction
rule
\[
(p\pre P \bnf q\pre Q) \quad \rew \quad \sigma P\bnf \rho Q
\]%
where $\sigma$ and $\rho$ are the substitutions on names resulting
from the unification of the patterns $p$ and $q$. 


The flexibility of the pattern unification and the symmetry of exchange 
in CPC align closely with the world of trade. Here the support for 
discovering a compatible process and exchanging information mirrors 
the behaviour of trading systems such as a stock market.
The main features of CPC are illustrated in the following sample trade interaction:
\[
\begin{array}{r}
\res{\mathit{\ sharesID}}\pro{\mathit{ABCShares}}\bullet\mathit{sharesID}\bullet\l x
\to \xangle{\mathit{charge\ x\ for\ sale}}
\\ \bnf \quad
\res{\mathit{\ bankAcc}}\pro{\mathit{ABCShares}}\bullet\l y\bullet\mathit{bankAcc}
\to \xangle{\mathit{save\ y\ as\ proof}}
\vspace*{.3cm}
\\
\redar
\res{\mathit{\ sharesID}}\res{\mathit{\ bankAcc}}
(\xangle{\mathit{charge\ bankAcc\ for\ sale}}\\
\bnf \xangle{\mathit{save\ sharesID\ as\ proof}})
\end{array}
\]
The first line models a seller that will synchronise with a buyer, using the protected
information $\mathit{ABCShares}$, and exchange its shares ($\mathit{sharesID}$) for bank
account information (bound to $x$).
The second line models a buyer. Notice that the information exchange is bidirectional
and simultaneous: {\it sharesID} replaces $y$ in the (continuation of the) buyer and
{\it bankAcc} replaces $x$ in the (continuation of the) seller.
Moreover, the two patterns $\pro{\mathit{ABCShares}}\bullet\mathit{sharesID}\bullet\l x$
and $\pro{\mathit{ABCShares}}\bullet\l y\bullet\mathit{bankAcc}$ also specify the
details of the stock being traded, that must be matched for equality in the pattern
matching, as indicated by the syntax $\pro {\,\cdot\,}$.

Pattern unification in CPC is even richer than indicated in this example, as unification
may bind a compound pattern to a single name; that is, patterns do not need to be
fully decomposed in unification.
For example, the bank account information could be specified, and matched upon, in
much more detail.
The buyer could provide the account name and number such as in the following pattern:
$\res{\mathit{\ accName}}\res{\mathit{\ accNum}}
\pro{\mathit{ABCShares}}\bullet\l y\bullet
(\mathit{name}\bullet\mathit{accName}\bullet\mathit{number}\bullet\mathit{accNum})$.
This more detailed buyer would still match against the seller, now yielding
$\xangle{\mathit{charge\ \mathit{name}\bullet\mathit{accName}\bullet\mathit{number}\bullet\mathit{accNum}\ for\ sale}}$.
Indeed, the seller could also specify a desire to only accept bank account information
whose structure includes some name and number (bound to $a$ and $b$ respectively) with the following pattern:
$\pro{\mathit{ABCShares}}\bullet\mathit{sharesID}\bullet
(\pro{\mathit name}\bullet\l a\bullet\pro{\mathit number}\bullet\l b)$
and continuation $\xangle{\mathit{charge\ a\ b\ for\ sale}}$.
This would also match with the detailed buyer information by
unifying $name$ with $\pro{name}$, $number$ with $\pro{number}$, and binding
$accName$ and $accNum$ to $a$ and $b$ respectively.
The second seller exploits the structural matching of CPC to only interact with a
buyer whose pattern is of the right structure (four sub-patterns) and contains
the right information (the protected names $name$ and $number$, and shared information
in the other two positions).

The structural patterns of CPC are inspired by those of {\em pattern calculus} \cite{JK09,pcb}
that supports even more computations than $\l$-calculus, since pattern-matching functions may
be {\em intensional} with respect to their arguments \cite{JayGW10}.
For example, the pattern $x~y$ can decompose any compound data structure
$u~v$ into its components $u$ and $v$.
This is different from {\em extensional} computation, like those of
$\l$- and $\pi$-calculus, where arguments are `atomic'.
This rich form of structural interaction, combined with concurrency, makes CPC very expressive, as
illustrated by the following diamond \cite{GivenWilsonPHD}:
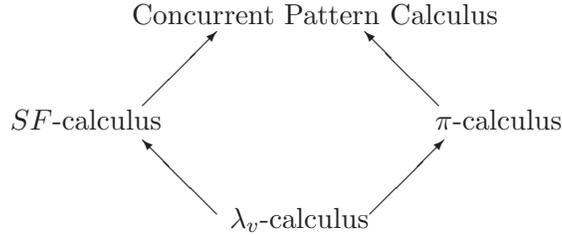
\begin{center}
\begin{picture}(240,85)(0,0)
\put(95,1){\mbox{$\l_v$-calculus}}
\put(12,38){\mbox{$SF$-calculus}}
\put(169,38){\mbox{~$\pi$-calculus}}
\put(58,78){\mbox{Concurrent Pattern Calculus}}
\put(90,6){\vector(-1,1){28}}
\put(148,6){\vector(1,1){28}}
\put(62,47){\vector(1,1){28}}
\put(174,47){\vector(-1,1){28}}
\end{picture}
\end{center}
The $\l$-calculus sits at the bottom and can be generalised either by 
$SF$-calculus \cite{JayGW10}, by considering intensionality in the sequential setting,
or by $\pi$-calculus, by considering concurrency in the extensional setting.
CPC completes the diamond by adding either concurrency to $SF$-calculus, or
intensionality to $\pi$-calculus. Thus, CPC is the most expressive of all
by supporting intensional concurrent computation.

The definition of CPC also includes a behavioural theory that defines
when two processes are behaviourally equivalent. This is done using a
standard approach in concurrency.  First, define an intuitive notion of
equivalence that equates processes with the same behaviour (i.e., with
the same interaction capabilities), in any context and along any reduction sequence, to yield a
notion of {\em barbed congruence}. Second, provide a more effective
characterisation of such equivalence by means of a {\em labelled
transition system} (LTS) and a {\em bisimulation-based} equivalence.
Although this path is familiar, some delicacy is required for each
definition.  For example, as unification of patterns may require
testing of names for equality, the {\em barbs} of CPC (i.e.  the
predicate describing the interactional behaviour of a CPC process)
must account for names that {\em might} be matched, not just those
that {\em must} be matched.  This is different from the standard barbs
of, say, the $\pi$-calculus.  Further, as some patterns are more
discriminating than others, the bisimulation defined here will rely on
a notion of {\em compatibility} of patterns, yielding a bisimulation
game in which a challenge can be replied to with a different, though
compatible, reply. This is reminiscent of the asynchronous
bisimulation for the asynchronous $\pi$-calculus
\cite{amadio.castellani.ea:bisimulations-asynchronous}
or the symbolic characterization of open bisimilarity in the $\pi$-calculus
\cite{San96}.

CPC's support for interaction that is both structured and symmetrical
allows CPC to simulate many approaches to interaction and reduction in the
literature \cite{G:IC08}.  For example, checking equality of channel names, as in
$\pi$-calculus \cite{milner.parrow.ea:calculus-mobile}, can be viewed
as a trivial form of pattern unification. It also supports 
unification of tuples of names, as in Linda \cite{Gel85}, or fusing names,
as in Fusion \cite{parrow.victor:fusion-calculus}.  Spi calculus
\cite{gordon1997ccp} adds patterns for numbers (zero and successors)
and encryptions. Also the Psi calculus \cite{BJPV11} introduces
support for structures, albeit with a limited symmetry.

More formally, $\pi$-calculus, Linda and Spi calculus can all be
encoded into CPC but CPC cannot be encoded into any of them.  By
contrast, the way in which name fusion is modeled in fusion calculus
is not encodable into CPC; conversely, the richness of CPC's pattern
unification is not encodable in fusion calculus. Similarly, the
implicit computation of name equivalence
in Psi calculus cannot be encoded within CPC;
the converse separation result is ensured by CPC's symmetry.

A natural objection to CPC is that its unification is too complex to
be an atomic operation. In particular, any limit to the size of
communicated messages could be violated by some match. Also, one
cannot, in practice, implement a simultaneous exchange of information,
so that pattern unification must be implemented in terms of simpler
primitives.
This objection applies to many other calculi.
For example, neither polyadic $\pi$-calculus' arbitrarily large tuple communication
nor Linda's pattern matching are atomic, but both underpin many existing programming environments
\cite{Pierce97pict:a,cpplinda,Klava,Lime}.
Indeed the arbitrary complexity of both Psi calculus and CPC patterns can
exceed the capability of any computer, yet both have implementations
\cite{Khorsandiaghai603139,cbondi}.
Simlar comments apply to other process calculi
\cite{20110201:jocaml,nomadic-pict}.
A further complexity is the secure synchronisation and information exchange between agents
in distinct locations \cite{DY83,Fournet07atype,bengtson2011refinement},
however since the focus here is on exploring structured, symmetric, pattern unification
and not implementations, we do not attempt to address these details.

The structure of the paper is as follows.  Section \ref{sec:cpc}
introduces symmetric matching through a concurrent pattern calculus
and an illustrative example. Section \ref{sec:bisim} defines the
behavioural theory of the language: its barbed congruence, LTS and the
alternative characterization via a bisimulation-based equivalence.  %
Section \ref{sec:compare} formalises the relation between CPC and
other process calculi.
Section \ref{sec:conclusions} concludes the paper.
Standard proofs have been moved to the Appendix.

\section{Concurrent Pattern Calculus} 
\label{sec:cpc} 
This section presents a {\em concurrent pattern calculus} (CPC) that
uses symmetric pattern unification as the basis of communication. Both
symmetry and pattern matching appear in existing models of
concurrency, but in more limited ways.  For example, $\pi$-calculus
requires a sender and receiver to share a channel, so that knowledge
of the channel is symmetric but information flows in one direction
only.  Fusion calculus achieves symmetry by fusing names together but
has no intensional patterns.  Linda's matching is more intensional as
it can test equality of an arbitrary number of names, and the number
of names to be communicated, in an atomic interaction.  Spi calculus
has even more intensional patterns, e.g.~for natural numbers, and can
check equality of terms (i.e.\ patterns), but does not perform
matching in general.  Neither Linda or Spi calculus support much
symmetry beyond that of the $\pi$-calculus.

The expressiveness of CPC comes from extending the class of
communicable objects from raw names to a class of {\em patterns} that
can be unified. This merges equality testing and bi-directional
communication in a single step.

\subsection{Patterns}

Suppose given a countable set of {\em names} ${\cal N}$ (meta-variables $n,m,x,y,z,\ldots$).
The {\em patterns} (meta-variables $p,p',p_1,q,q',q_1,\ldots$)
are built using names and have the following forms:
\[
\begin{array}{lll}
{\it Patterns}\quad p \ ::= \
&	\lambda x					& \mbox{binding name} \\
&	x									& \mbox{variable name} \\
&	\pro x						& \mbox{protected name} \\
&	p\bullet p \quad	& \mbox{compound}
\end{array}
\]%

A binding name $\lambda x$ denotes information sought, e.g.\ by a
trader; a variable name $x$ represents such information. Note that a
binding name binds the underlying process, defined in
Section~\ref{subsec:processes}.  Protected names $\pro x$ represent
information that can be checked but not traded.  A compound combines
two patterns $p$ and $q$, its {\em components}, into a pattern
$p\bullet q$.  Compounding is left associative, similar to application
in $\l$-calculus, and pure pattern calculus.
The {\em atoms} are patterns that are not compounds.  The
atoms $x$ and $\pro x$ {\em know} $x$.

Binding, variable and protected names are all well established
concepts in the literature. Indeed, there is a correspondence between
patterns and prefixes of more familiar process calculi, such as
$\pi$-calculus: binding names correspond to input arguments and
variable names to output arguments.  Moreover, a form of protected
names appear in Linda.
There is some subtlety in the relationship of protected names to
variable names.  As protected names specify a requirement, it is
natural that they unify with the variable form of the name.
Similarly, as protected names can be used to support
channel-based communication, it is also natural that protected names
unify with themselves.

Given a pattern $p$ the sets of: {\em variables names}, denoted ${\sf
  vn}(p)$; {\em protected names}, denoted ${\sf pn}(p)$; and {\em
  binding names}, denoted ${\sf bn}(p)$, are defined as expected with
the union being taken for compounds.  The {\em free names} of a
pattern $p$, written ${\sf fn}(p)$, is the union of the variable names
and protected names of $p$.  A pattern is {\em well formed} if its
binding names are pairwise distinct and different from the free ones.
All patterns appearing in the rest of this paper are assumed to be
well formed.

As protected names are limited to recognition and binding names are
being sought, neither should be communicable to another process.
This leads to the following definition.
\begin{defi}[Communicable Pattern]
\label{def:communicable}
A pattern is {\em communicable} if it contains no protected or binding names.
\end{defi}
Protection can be extended from names to communicable patterns by
defining
$$
\pro {p\bullet q} = \pro p \bullet \pro q
$$

A {\em substitution} $\sigma$ is defined as a partial function from names to communicable patterns.
The {\em domain} of $\sigma$ is denoted $\textsf{dom}(\sigma)$;
the free names of $\sigma$, written $\textsf{fn}(\sigma)$, is given by the union of the sets $\textsf{fn}(\sigma x)$ where $x \in \textsf{dom}(\sigma)$.
The {\em names} of $\sigma$, written $\textsf{names}(\sigma)$, are $\textsf{dom}(\sigma)\cup\textsf{fn}(\sigma)$.
A substitution $\sigma$ {\em avoids} a name $x$ (or a collection of names $\wt n$) if $x\notin\textsf{names}(\sigma)$ (respectively $\wt n\cap\textsf{names}(\sigma)=\{\}$).
Note that all substitutions considered in this paper have finite domain.
For later convenience, we denote by $\idsub_X$ the identity
substitution on a finite set of names $X$; it maps every name in $X$ to
itself, i.e.  $\idsub_X(x) = x$, for every $x\in X$.

Substitutions are applied to patterns as follows
\begin{eqnarray*}
\sigma x &=&
\left\{
\begin{array}{ll}
\sigma (x) & \mbox{if $x \in \textsf{dom}(\sigma)$}\\
x & \mbox{otherwise}
\end{array}
\right. \\
\sigma \pro x &=&
\left\{
\begin{array}{ll}
\pro{\sigma (x)} & \mbox{if $x \in \textsf{dom}(\sigma)$}\\
\pro x & \mbox{otherwise}
\end{array}
\right. \\
\sigma (\l x)&=& \l x\\
\sigma (p\bullet q)&=& (\sigma p)\bullet (\sigma q)
\end{eqnarray*}
The action of a substitution
$\sigma$ on patterns can be adapted to produce a function $\hat\sigma$
that acts on binding names rather than on free names.  In CPC, it is
defined by
\begin{eqnarray*}
\hat\sigma x &=& x\\
\hat\sigma \pro x &=& \pro x\\
\hat\sigma (\l x) &=&
\left\{
\begin{array}{ll}
\sigma(x) &\mbox{if $x \in \textsf{dom}(\sigma)$}\\
\l x &\mbox{otherwise}
\end{array}
\right. \\
\hat\sigma (p\bullet q)&=& (\hat\sigma p)\bullet (\hat\sigma q)
\end{eqnarray*}
When $\sigma$ is of the form $\{p_i/x_i\}_{i \in I}$,
then $\{p_i/\l x_i\}_{i \in I}$ may be used to denote $\hat\sigma$.

The {\em symmetric matching} or {\em unification} $\{p \pmatch q\}$ of
two patterns $p$ and $q$ attempts to unify $p$ and $q$ by generating
substitutions upon their binding names. When defined, the result is
a pair of substitutions whose domains are the binding names of $p$
and of $q$, respectively. The rules to generate the substitutions are:
\begin{eqnarray*}
\begin{array}{rcll}
\left.
\begin{array}{r}
\{x\pmatch x\}\\
\{x\pmatch \pro{x}\}\\
\{\pro{x}\pmatch x\}\\
\{\pro{x}\pmatch \pro{x}\}
\end{array}
\right\}\!\!\! &=\!\!\!&  (\{\}, \{\})\\
\{\lambda x\pmatch q\} &=\!\!\!&  (\{q/x\}, \{\}) & \mbox{if $q$ is communicable}\\
\{p\pmatch \lambda x\} &=\!\!\!&  (\{\}, \{p/x\}) & \mbox{if $p$ is communicable}\\
\{p_1\bullet p_2\pmatch q_1\bullet q_2\} &=\!\!\!&  ((\sigma_1\cup\sigma_2), (\rho_1\cup\rho_2))
			& \mbox{if } \{p_i\pmatch q_i\}= (\sigma_i, \rho_i) \mbox{ for } i \in \{1,2\}\\
\{p\pmatch q\} &=\!\!\!& \mbox{undefined} & \mbox{otherwise}
\end{array}
\end{eqnarray*}
Two atoms unify if they know the same name. A name that seeks
information (i.e., a binding name) unifies with any communicable pattern to
produce a binding for its underlying name.  Two compounds unify if
their corresponding components do; the resulting substitutions are
given by taking unions of those produced by unifying the components 
(necessarily disjoint as patterns are well-formed).  Otherwise the
patterns cannot be unified and the unification is undefined.

\begin{prop}
\label{prop:free_n_match}
If the unification of patterns $p$ and $q$ is defined then any protected
name of $p$ is a free name of $q$.
\end{prop}
\begin{proof}
By induction on the structure of $p$. 
\end{proof}

\subsection{Processes}
\label{subsec:processes}

The processes of CPC are given by:
\[
\begin{array}{lll}
{\it Processes}\quad P \ ::= \
&	\zero					& \mbox{null} \\
&	P|P						& \mbox{parallel composition} \\
&	!P						& \mbox{replication} \\
&	\res x P			& \mbox{restriction} \\
&	p\pre P	\quad	& \mbox{case}
\end{array}
\]%
The null process, parallel composition, replication and restriction
are the classical ones for process calculi: $\zero$ is the inactive
process; $P\bnf Q$ is the parallel composition of processes $P$ and
$Q$, allowing the two processes to evolve independently or to
interact; the replication $!P$ provides as many parallel copies of $P$
as desired; $\res x P$ binds $x$ in $P$ so that it is not visible from
the outside.  The
traditional input and output primitives are replaced by the {\em
  case}, viz.  $p\pre P$, that has a {\em  pattern} $p$ and a {\em body} $P$.  If
$P$ is $\zero$ then $p\to \zero$ may be denoted by $p$.

The free names of processes, denoted ${\sf fn}(P)$, are defined as
usual for all the traditional primitives and
\[
{\sf fn}(p\pre P) \quad =\quad {\sf fn}(p)\cup({\sf fn}(P)\backslash{\sf bn}(p))
\]%
for the case. As expected the binding names of the pattern bind their free occurrences in the body.

\subsection{Operational Semantics}
\label{s:reduction-new}

The application $\sigma P$ of a substitution $\sigma$ to a process $P$
is defined in the usual manner, provided that there is no name
capture.
Name capture can be avoided by $\alpha$-conversion (written $=_\alpha$) that
is the congruence relation generated by the following axioms:
\[
\begin{array}{rcll}
\res x P &=_\alpha& \res y (\{y/x\}P) & \quad y\notin{\sf fn}(P)\\
p\pre P &=_\alpha& (\{\l y/\l x\}p)\pre (\{y/x\}P) & 
			\quad x\in{\sf bn}(p),y\notin{\sf fn}(P)\cup{\sf bn}(p)
\end{array}
\]%

The {\em structural congruence relation} $\equiv$ is defined
just as in $\pi$-calculus \cite{milner:polyadic-tutorial}: it includes
$\alpha$-conversion and its defining axioms are:
\begin{equation*}
\begin{array}{c}
\vspace*{.3cm}
P\bnf \stoppr \equiv P
\qquad
P \bnf Q \equiv Q \bnf P
\qquad
P \bnf (Q \bnf R) \equiv (P \bnf Q) \bnf R
\\
\vspace*{.3cm}
\res n \stoppr \equiv \stoppr
\qquad
\res n \res m P \equiv \res m \res n P
\qquad
!P \equiv P\bnf !P
\\
P\bnf\res n Q \equiv \res n (P\bnf Q)
\quad \mbox{if $n \not\in{\sf fn}(P)$}
\end{array}
\end{equation*}
It states that: $|$ is a commutative, associative, monoidal operator,
with $\zero$ acting as the identity; that restriction has no effect
on the null process; that the order of restricted names is
immaterial; that replication can be freely unfolded; and that the
scope of a restricted name can be freely extended, provided that no
name capture arises.

For later convenience, $\wt x$ denotes a sequence of $x$'s, for example
$\wt n$ can denote the names $n_1,\ldots,n_i$.
Similarly, $\res {n_1} (\ldots (\res {n_i} P))$ will be written $\rest n P$;
however, due to structural congruence, these shall be considered as a set of names.
\withsetnot{}{For clarity set notation is used, for example
$\res {m_1} \res {m_2}\ldots \res {m_i} \res {n_1} \res {n_2} \ldots \res {n_j} P$
shall be denoted $\res {\wt m \cup \wt n} P$.}

The operational semantics of CPC is formulated via a {\em reduction
relation} $\redar$ between pairs of CPC processes.  Its defining
rules are:
\begin{eqnarray*}
\begin{array}{c}
\Rule{}
{}
{(p\pre P)\bnf(q\pre Q) \redar (\sigma P)\bnf(\rho Q)}{if $\{p\pmatch q\}=(\sigma,\rho)$}\\
\vspace*{.0cm}
\\
\Rule{}
{P \rew P'}{P|Q \rew P'|Q}{}
\ 
\Rule{}
{P \rew P'}{\res n P \rew \res n P'}{}
\ 
\Rule{}
{P \equiv Q \ \ Q \rew Q'\ \ Q' \equiv P'}{P \rew P'}{}
\vspace*{.1cm}
\end{array}
\end{eqnarray*}
CPC has one interaction axiom, stating that, if the unification of two
patterns $p$ and $q$ is defined and generates $(\sigma, \rho)$, then
the parallel composition of two cases $p\pre P$ and $q\pre Q$ reduces
to the parallel composition of $\sigma P$ and $\rho Q$.
Alternatively, if the unification of $p$ and $q$ is undefined, then no
interaction occurs.
Unlike the sequential setting, there is no need for a rule to capture failure
of unification since failure to interact with one process does not
prevent interactions with other processes.

The interaction rule is then closed under parallel composition,
restriction and structural congruence in the usual manner. Unlike
pure pattern calculus, but like pi-calculus, computation does not occur within the body of a
case. As usual, $\redar^k$ denotes $k$ interactions, and $\Redar$ denotes the reflexive and transitive closure
of $\redar$.

The section concludes with three simple properties of substitutions and the reduction relation.

\begin{prop}
\label{lem:fn-sub}
For every process $P$ and substitution $\sigma$, it holds that ${\sf
  fn}(\sigma P)\subseteq {\sf fn}(P)\cup{\sf fn}(\sigma)$.
\end{prop}
\begin{proof}
By definition of the application of $\sigma$.
\end{proof}

\begin{prop}
\label{prop:free_name_subset}
If $P\Redar P'$, then ${\sf fn}(P')\subseteq {\sf fn}(P)$.
\end{prop}
\begin{proof}
$P\Redar P'$ means that $P\redar\!\!\!^k\, P'$, for some $k \geq 0$. The proof is by induction on $k$
and follows by Proposition~\ref{lem:fn-sub}.
\end{proof}

\begin{prop}
\label{red:renaming}
If $P \redar P'$, then $\sigma P \redar \sigma P'$, for every $\sigma$.
\end{prop}
\begin{proof}
By induction on the derivation of  $P \redar P'$.
\end{proof}

\begin{prop}
\label{prop:free_n_match_proc}
Suppose a process $p\to P$ interacts with a process $Q$. If $x$ is a protected
name in $p$ then $x$ must be a free name in $Q$.
\end{prop}
\begin{proof}
  For $Q$ to interact with $p\to P$ it must be that $Q\Redar \rest n
  (q\to Q_1\bnf Q_2)$ such that $\wt n \cap {\sf fn}(p \to P) =
  \emptyset$ and $\{p\pmatch q\}$ is defined. Then, by
  Proposition~\ref{prop:free_n_match}, the free names of $q$ include
  $x$ and, consequently, $x$ must be free in $q\to Q_1\bnf Q_2$.  By
  Proposition~\ref{prop:free_name_subset}, $x$ is free in $Q$. Further,
  $x$ cannot belong to $\wt n$, since $x \in {\sf fn}(p \to P)$ and
  $\wt n \cap {\sf fn}(p \to P) = \emptyset$.
\end{proof}

\subsection{Trade in CPC}
\label{subsec:cpc-trade}

This section uses the example of share trading to explore the
potential of CPC. The scenario is that two  traders, a buyer
and a seller, wish to engage in trade. To complete a transaction, the
traders need to progress through two stages: {\em discovering} each
other and {\em exchanging} information. Both traders begin with a
pattern for their desired transaction.
The discovery phase can be characterised as a pattern-unification
problem, where traders' patterns are used to find a compatible
partner.
The exchange phase occurs when a buyer and seller have agreed upon a
transaction. Now each trader wishes to exchange information in a
single interaction, preventing any incomplete trade from occurring.

The rest of this section develops three solutions of increasing
sophistication that: demonstrate discovery; introduce a registrar to
validate the traders; and protects names to ensure privacy.
\enlargethispage{\baselineskip}

\subsubsection*{Solution 1}
\label{sec:example:sol1}
Consider two traders, a buyer and a seller. The buyer $\buy_1$ with
bank account $b$ and desired shares $s$ can be given by
\begin{eqnarray*}
\buy_1 &=& s\bullet\lambda m \pre m\bullet b\bullet \lambda x \pre B(x)
\end{eqnarray*}
The first pattern $s\bullet \lambda m$ is used to match with a
compatible seller using share information $s$, and to input a name $m$
to be used as a channel to exchange bank account information $b$ for
share certificates bound to $x$. The transaction successfully
concludes with $B(x)$.

The seller $\sel_1$ with share certificates $c$ and desired share sale
$s$ is given by
\begin{eqnarray*}
\sel_1 &=& \res n s\bullet n\pre n\bullet \lambda y\bullet c\pre S(y) 
\end{eqnarray*}
The seller creates a channel name $n$ and then tries to find a buyer
for the shares described in $s$, offering $n$ to the buyer to
continue the transaction. The channel is then used to exchange
billing information, bound to $y$, for the share certificates $c$.
The seller then concludes with the successfully completed transaction
as $S(y)$.

The discovery phase succeeds when the traders are placed in a parallel
composition and discover each other by unification on $s$
\begin{eqnarray*}
\buy_1|\sel_1
&\equiv&
	\res n (s\bullet \lambda m \pre m\bullet b \bullet \lambda x \pre B(x)
				\bnf	s\bullet n \pre n\bullet \lambda y \bullet c \pre S(y))\\
&\rew&
	\res n (n\bullet b \bullet \lambda x \pre B(x)
				\bnf	n\bullet \lambda y \bullet c \pre S(y)) 
\end{eqnarray*}
The next phase is to exchange billing information for share certificates, as in
$$
\res n (n\bullet b \bullet \lambda x \pre B(x)
				\bnf	n\bullet \lambda y \bullet c \pre S(y))
\ \ \rew\ \ \res n (B(c)\bnf S(b))
$$
The transaction concludes with the buyer having the share
certificates $c$ and the seller having the billing account $b$.

This solution allows the traders to discover each other and exchange
information atomically to complete a transaction. However, there is no way to
determine if a trader is  trustworthy.

\subsubsection*{Solution 2}
\label{sec:example:sol2}
Now add a registrar that keeps track of registered
traders. Traders offer their identity to potential partners
and the registrar confirms if the identity belongs to a valid trader.
The buyer is now
\begin{eqnarray*}
\buy_2 &=&	s\bullet i_B\bullet\lambda j \pre n_B\bullet j\bullet\lambda m \pre 
						m\bullet b\bullet \lambda x \pre B(x) 
\end{eqnarray*}
The first pattern now swaps the buyer's identity $i_B$ for the seller's, 
bound to $j$. The buyer then consults the registrar using the
identifier $n_B$ to validate $j$; if valid, the exchange continues as before.

Now define the seller symmetrically by
\begin{eqnarray*}
\sel_2 &=& s\bullet \lambda j\bullet i_S \pre n_S\bullet j\bullet\lambda m \pre
           m\bullet \lambda y\bullet c\pre S(y) 
\end{eqnarray*}
Also define the registrar $\reg_2$ with identifiers $n_B$ and $n_S$ to
communicate with the buyer and seller, respectively, by
\begin{eqnarray*}
\reg_2 &=& \res n (n_B\bullet{i_S}\bullet n\bnf n_S\bullet{i_B}\bullet n) 
\end{eqnarray*}
The registrar creates a new identifier $n$ and delivers it to traders who
have been validated; then it makes the identifier available to known
traders who attempt to validate another known trader.  Although rather
simple, the registrar can easily be extended to support a multitude of
traders.

Running these processes in parallel yields the following interaction
\begin{eqnarray*}
& & \buy_2 \bnf \sel_2 \bnf \reg_2\\
 &\equiv&	\res n (\ s\bullet i_B\bullet\lambda j \pre n_B\bullet j\bullet\lambda m \pre 
									m\bullet b\bullet \lambda x \pre B(x)
						\bnf  n_B\bullet\  {i_S}\bullet n\\
 && 	\qquad |\ 	s\bullet \lambda j\bullet i_S\ \pre n_S\bullet j\bullet\lambda m \pre
           				m\bullet \lambda y\bullet c\pre S(y)
 						\bnf  n_S\bullet  {i_B}\bullet n)\\
 &\rew&	  \res n (n_B\bullet i_S\bullet\lambda m \pre 
									m\bullet b\bullet \lambda x \pre B(x)
						\bnf  n_B\bullet  {i_S}\bullet n\\
 && 	\qquad |\   n_S\bullet i_B\bullet\lambda m \pre
           				m\bullet \lambda y\bullet c\pre S(y)
 						\bnf  n_S\bullet  {i_B}\bullet n) 
\end{eqnarray*}
The share information $s$ allows the buyer and seller to discover
each other and swap identities $i_B$ and $i_S$. The next two
interactions involve the buyer and seller validating each other's
identity and inputting the identifier to complete the transaction
\begin{eqnarray*}
 &&	  \res n (n_B\bullet i_S\bullet\lambda m \pre 
									m\bullet b\bullet \lambda x \pre B(x)
						\bnf  n_B\bullet  {i_S}\bullet n\\
 && 	\qquad |\   n_S\bullet i_B\bullet\lambda m \pre
           				m\bullet \lambda y\bullet c\pre S(y)
 						\bnf  n_S\bullet  {i_B}\bullet n)\\
&\rew&\res n (n\bullet b\bullet \lambda x \pre B(x)\\
 && 	\qquad |\   n_S\bullet i_B\bullet\lambda m \pre
           				m\bullet \lambda y\bullet c\pre S(y)
 						\bnf  n_S\bullet  {i_B}\bullet n)\\
&\rew&\res n (n\bullet b\bullet \lambda x \pre B(x)\bnf 
           				n\bullet \lambda y\bullet c\pre S(y))
\end{eqnarray*}
Now that the traders have validated each other, they can continue
with the exchange step from before
\begin{eqnarray*}
\res n (n\bullet b\bullet \lambda x \pre B(x)\bnf n\bullet \lambda y\bullet c\pre S(y))
& \rew & \res n (B(c)\bnf S(b))
\end{eqnarray*}
The traders exchange information and successfully complete with 
$B(c)$ and $S(b)$.

\subsubsection*{Solution 3}
\label{sec:example:sol3}
Although Solution 2 satisfies the desire to validate that
traders are legitimate, the freedom of unification allows for
malicious processes to interfere. Consider the promiscuous
process $\prom$ given by
\begin{eqnarray*}
\prom &=& \lambda z_1\bullet \lambda z_2\bullet a\pre P(z_1,z_2) 
\end{eqnarray*}
This process is willing to match any other process that will swap
two pieces of information for some arbitrary name $a$. Such a
process could interfere with the traders trying to complete the
exchange phase of a transaction. For example,
\begin{eqnarray*}
\res n (n\bullet b\bullet \lambda x \pre B(x)\bnf n\bullet \lambda y\bullet c\pre S(y))\bnf\prom\\
\qquad \qquad \rew \quad \res n (B(a)\bnf n\bullet \lambda y\bullet c\pre S(y) \bnf P(n,b))
\end{eqnarray*}
where the promiscuous process has stolen the identifier $n$ and the
bank account information $b$. The unfortunate buyer is left with some
useless information $a$ and the seller is waiting to complete the
transaction.

This vulnerability (emerging both in Solution 1 and 2) can be repaired by using protected
names.
For example, the buyer, seller and registrar of Solution 2 can become 
\begin{eqnarray*}
\buy_3 &=&	s\bullet i_B\bullet\lambda j \pre \pro {n_B}\bullet j\bullet\lambda m \pre 
						\pro m\bullet b\bullet \lambda x \pre B(x) \\
\sel_3 &=& s\bullet \lambda j\bullet i_S \pre \pro {n_S}\bullet j\bullet\lambda m \pre
           \pro m\bullet \lambda y\bullet c\pre S(y) \\
\reg_3 &=& \res n (\pro {n_B}\bullet\pro{i_S}\bullet n\bnf \pro {n_S}\bullet\pro{i_B}\bullet n)
\end{eqnarray*}
Now all communications between the buyer, seller and registrar
use protected identifiers: $\pro {i_B}, \pro {i_S}, \pro {n_B},\pro {n_S}$ and $\pro m$.
Thus, we just need to add the appropriate restrictions:
\begin{eqnarray*}
\res {i_B}\res {i_S}\res {n_B} \res {n_S} (\buy_3\bnf\sel_3\bnf\reg_3)
\end{eqnarray*}
Therefore, other processes can only interact with the traders during
the discovery phase, which will not lead to a successful
transaction. The registrar will only interact with the
traders as all the registrar's patterns have protected names
known only to the registrar and a trader (Proposition~\ref{prop:free_n_match_proc}).

\section{Behavioural Theory}
\label{sec:bisim}
\newcommand{\usedby}[1]{}

This section follows a standard approach in concurrency to defining
behavioural equivalences, beginning with a barbed congruence and
following with a labelled transition system (LTS) and a bisimulation
for CPC. We will prove that the two semantics do coincide. Then the
bisimulation technique will be used to prove some sample equational
laws for CPC.
\nobreak
\subsection{Barbed Congruence}
The first  step is to characterise the interactions a process can participate in via
{\em barbs}. 
Since a barb is an opportunity for interaction, a simplistic
definition could be the following:
\begin{equation}
\label{def-one}
P \barb{} \mbox{ iff }\ P \equiv p \pre P'\bnf P'' \mbox{, for some $p,P'$ and $P''$}
\end{equation}
However, this definition is too strong: for example, $\res n(n \pre P)$ does not
exhibit a barb according to \eqref{def-one}, but it can interact with an external
process, e.g. $\lambda x \pre \zero$.
Thus, an improvement to \eqref{def-one} is as follows:
\begin{equation}
\label{def-two}
P \barb{} \mbox{ iff }\ P \equiv \rest n(p \pre P'\bnf P'') \mbox{, for some $\wt n, p, P'$ and $P''$}
\end{equation}
However, this definition is too weak. Consider $\res n(\pro n \pre P)$: it exhibits
a barb according to \eqref{def-two}, but cannot interact with any external
process.
A further refinement on \eqref{def-two} could be:
\begin{equation}
\label{def-three}
P \barb{} \mbox{ iff }\ P \equiv \rest n(p \pre P'\bnf P'') \mbox{, for some $\wt n, p, P',P''$ s.t. }
{\sf pn}(p) \cap \wt n = \emptyset
\end{equation}
This definition is not yet the final one, as it is not sufficiently discriminating to have
only a single kind of barb.
Because of the rich form of interactions in CPC, there is no single identifier such as in CCS
and $\pi$-calculus $\pi$-calculus \cite{milner.sangiorgi:barbed-bisimulation},
thus CPC barbs include the set of names that {\em may} be tested for equality
in an interaction, not just those that {\em must} be equal. This leads to the following definition:

\begin{defi}[Barb]
\label{def:barb}
Let $P\barb{\wt m}$ mean that
$P\equiv \rest n (p\pre P'\bnf P'')$
for some $\wt n$ and $p$ and $P'$ and $P''$ such that 
${\sf pn}(p)\cap \wt n = \emptyset$ and $\wt m = {\sf fn}(p) \backslash \wt n$.
\end{defi}

For later convenience, define $P\suc_{\wt m}$ to mean that there exists some $P'$
such that $P\Redar P'$ and $P'\barb{\wt m}$.

Using this definition, a barbed congruence can be defined in the standard way
\cite{milner.sangiorgi:barbed-bisimulation,HY95} by requiring three properties.
Let $\Re$ denote a binary relation on CPC processes, and let a {\em context}
$\context C \cdot$ be a CPC process with the hole `$\,\cdot\,$'.

\begin{defi}[Barb preservation]
\label{def:barb-pres}
$\Re$ is barb preserving iff, for every $(P,Q) \in \Re$ and set of names  ${\wt m}$, it holds that $P \barb{\wt m}$
implies $Q \barb{\wt m}$.
\end{defi}

\begin{defi}[Reduction closure]
\label{def:barb-close}
$\Re$ is reduction closed iff, for every $(P,Q) \in \Re$, it holds that $P \redar P'$
implies $Q \redar Q'$, for some $Q'$ such that $(P',Q') \in \Re$.
\end{defi}

\begin{defi}[Context closure]
\label{def:cont-close}
$\Re$ is context closed iff, for every $(P,Q) \in \Re$ and CPC context
$\context C \cdot$, 
it holds that $(\context C P, \context C Q) \in \Re$.
\end{defi}

\begin{defi}[Barbed congruence]
\label{def:barb-con}
Barbed congruence, $\beq$, is the least binary relation on CPC processes that is symmetric, barb preserving, reduction closed and
context closed.
\end{defi}

Barbed congruence relates processes with the same
behaviour, as captured by barbs: two equivalent processes must exhibit the same behaviours, 
and this property should hold along every sequence of reductions and in every execution context.

The challenge in proving barbed congruence is to prove
context closure.  The typical way of solving the problem is by giving
a coinductive (bisimulation-based) characterization of barbed
congruence, that provides a manageable proof technique. In turn, this
requires an alternative operational semantics, by means of a labelled
transition system, on top of which the bisimulation equivalence can be
defined.

\subsection{Labelled Transition System}
\label{sec:LTS}

The following is an adaption of the late LTS for the
$\pi$-calculus \cite{milner.parrow.ea:calculus-mobile}.
{\em Labels} are defined as follows:
$$
\mu\ ::=\ \tau \quad |\quad \rest n p 
$$
where $\tau$ is used to label silent transitions.
\begin{figure}[t]
$$
\begin{array}{ll}
{\sf case:}&
(p \pre P) \ltsred p P\\
\\
{\sf resnon:}&
\prooftree P \ltsred\mu P'
\justifies \res n P\ltsred\mu \res n P'
\endprooftree\quad n \notin {\sf names}(\mu)
\vspace*{.5cm}
\\
{\sf open:}&
\prooftree P \ltsred{\rest n p} P'
\justifies \res m P\ltsred{\res {\withsetnot{\wt n,m}{\wt n\cup \{m\}}} p} P'
\endprooftree\quad m \in {\sf vn}(p) \setminus (\wt n \cup {\sf pn}(p) \cup {\sf bn}(p))
\vspace*{.5cm}
\\
{\sf unify:}&
\prooftree P \ltsred{\rest m p} P' \quad Q \ltsred{\rest n q} Q'
\justifies P\bnf Q \ltsred\tau \res{\withsetnot{\wt m,\wt n}{\wt m\cup \wt n}}(\sigma P'\bnf \rho Q')
\endprooftree\quad
\begin{array}{l}
\{p \pmatch q\} = (\sigma, \rho)\\
\wt m \cap {\sf fn}(Q) = \wt n \cap {\sf fn}(P) = \emptyset\\
\wt m \cap \wt n = \emptyset
\end{array}
\vspace*{.5cm}
\\
{\sf parint:}&
\prooftree P \ltsred\tau P'
\justifies P\bnf Q \ltsred\tau P'\bnf Q
\endprooftree
\vspace*{.5cm}
\\
{\sf parext:}&
\prooftree P \ltsred{\rest n p} P'
\justifies P\bnf Q \ltsred{\rest n p} P'\bnf Q
\endprooftree\quad (\wt n \cup {\sf bn}(p)) \cap {\sf fn}(Q) = \emptyset
\vspace*{.5cm}
\\
{\sf rep:}&
\prooftree !P|P \ltsred\mu P'
\justifies !P \ltsred\mu P'
\endprooftree
\end{array}
$$
\caption{Labelled Transition System for CPC (the symmetric versions of {\sf parint} and {\sf parext} have been omitted)}
\label{fig:lts}
\end{figure}

Labels are used in {\em transitions} $P\ltsred\mu P'$ between CPC
processes, whose defining rules are given in Figure~\ref{fig:lts}. If
$P\ltsred\mu P'$ then $P'$ is a {\em $\mu$-reduct} of $P$, alternatively the transition $P\ltsred\mu P'$
indicates that $P$ is able to {\em perform} $\mu$ and reduces to $P'$.
Rule {\sf case} states that a case's pattern can be used to interact with external processes.
Rule {\sf resnon} is used when a restricted name does not appear in the names of the label: it simply maintains the restriction on the process after the transition.
By contrast, rule {\sf open} is used when a restricted name occurs in the label: as the restricted name is going to be shared with other processes, the restriction is moved from the process to the label (this is called {\em extrusion}, in $\pi$-calculus terminology).
Rule {\sf unify} defines when two processes can interact to perform an internal action:
this can occur whenever the processes exhibit labels with unifiable patterns and with no possibility of clash or capture due to restricted names.
Rule {\sf parint} states that, if either process in a parallel composition can evolve with an internal action, then the whole process can evolve with an internal action.
Rule {\sf parext} is similar, but is used when the label is visible: when one of the processes in parallel exhibits an external action, then the whole composition exhibits the same external action, as long as the restricted or binding names of the label do not appear free in the parallel component that does not generate the label.
Finally, rule {\sf rep} unfolds the replicated process to infer the action. 

Note that $\alpha$-conversion is always assumed, so that the  side conditions can always be satisfied in practice. 

The presentation of the LTS is concluded with the following two
results.  First, for every $P$ and $\mu$, there are finitely many
$\equiv$-equivalence classes of $\mu$-reducts of $P$
(Proposition~\ref{prop:imfin}).  Second, the LTS induces the same
operational semantics as the reductions of CPC
(Proposition~\ref{prop:tau-red}).  As CPC reductions only involve
interaction between processes and not external actions, it is
sufficient to show that any internal action of the LTS is mimicked by
a reduction in CPC, and vice versa. All proofs are in Appendix A, because they are quite standard.

\begin{defi}
An LTS is {\em structurally image finite} if, for every $P$ and $\mu$, it holds that
$\{P' :\, P \ltsred\mu P'\}/_\equiv$ contains finitely many elements.
\end{defi}

\begin{prop}
\label{prop:imfin}
The LTS defined in Figure~\ref{fig:lts} is structurally image finite.
\end{prop}
\usedby{Lemma \ref{lem:bisim-rep}}

\begin{lem}
\label{lem:lts-exhibit-p}
If $P\ltsred{\rest m p}P'$ then 
there exist $\wt n$ and $Q_1$ and $Q_2$ such that 
$P \equiv \rest m \rest n (p\to Q_1\bnf Q_2)$ and
$P' \equiv \rest n (Q_1\bnf Q_2)$ and
$\wt n\cap{\sf names}(\rest m p)=\emptyset$ and ${\sf bn}(p)\cap{\sf fn}(Q_2)=\emptyset$.
\end{lem}
\usedby{Prop \ref{prop:tau-red} and Lemma \ref{lem:sound-one}}

\begin{prop}
\label{prop:tau-red}
If $P \ltsred\tau P'$ then $P \redar P'$.
Conversely, if $P \redar P'$ then there exists $P''$ such that $P \ltsred\tau P'' \equiv P'$.
\end{prop}
\usedby{Lemma \ref{lem:sound-two}, Lemma \ref{lem:succ-beq} and Thm \ref{thm: complete}}

To conclude, it is known that $\alpha$-conversion must be handled with care \cite{UBN07}.
A way in which we can leave it out from our presentation is to follow \cite{BP09}, that
also has the advantage of being implementable in Isabelle/HOL \cite{NPW02}.
However, we prefer to follow a more traditional approach in our presentation.

\subsection{Bisimulation}
\label{sec:bis}

We now develop a {\em bisimulation} relation for CPC that equates processes with the same interactional behaviour; this is captured by the labels of the LTS.
The complexity for CPC is that the labels for external actions contain patterns, and some patterns are more general than others, in that they unify with more things. 
For example, a transition $P\ltsred{\pro n}P'$ performs the action $\pro n$; however a similar external action of another process could be the variable name $n$ and the transition $Q\ltsred n Q'$.
Both transitions have the same barb, that is $P\barb{n}$ and $Q\barb{n}$; however their labels are not identical
and, indeed, the latter can interact with a process performing a transition labeled with $\l x$ whereas the former cannot.
Thus, a {\em compatibility} relation is defined on patterns that can be used to develop the bisimulation.
The rest of this section discusses the development of compatibility and concludes with the definition of bisimulation for CPC.

Bisimilarity of two processes $P$ and $Q$ can be captured by a challenge-reply game 
based upon the actions the processes can take.
One process, say $P$, issues a {\em challenge} and evolves to a new state $P'$.
Now $Q$ must perform an action that is a {\em proper reply} and evolve to a state $Q'$.
If $Q$ cannot perform a proper reply then the challenge issued by $P$ can distinguish $P$ and $Q$, 
and shows they are not equivalent.
If $Q$ can properly reply then the game continues with the processes $P'$ and $Q'$.
Two processes are bisimilar (or equivalent) if any challenge by one
can be answered by a proper reply from the other. 

The main complexity in defining a bisimulation to capture this challenge-reply game is the choice of actions, i.e.\ challenges and replies.
In most process calculi, a challenge is replied to with an identical action \cite{Mil89,milner.parrow.ea:calculus-mobile}.
However, there are situations in which an exact reply would make the bisimulation equivalence too fine for characterising barbed congruence \cite{amadio.castellani.ea:bisimulations-asynchronous,DGP:IC07}.
This is due to the impossibility for the language contexts to force barbed congruent processes to execute the same action; in such calculi more liberal replies must be allowed. 
That CPC lies in this second group of calculi is demonstrated by the following two examples.

\begin{exa}{Example 1}
\label{ex:input}
Consider the processes
\[
P = \l x \bullet \l y \to x\bullet y
\qquad \mbox{and} \qquad
Q = \l z \to z
\]
together with the challenge $P \ltsred{\l x \bullet \l y} x \bullet y$.
One may think that a possible context
$\context {C_{{\it \l x \bullet \l y}}} \cdot$
to enforce a proper reply could be $\ \cdot\bnf w \bullet w \pre \pro w$, for $w$ fresh.
Indeed, $\context {C_{{\it \l x\bullet\l y}}} P\redar w\bullet w \bnf \pro w$
and the latter process exhibits a barb over $w$.
However, the exhibition of action $\l x \bullet \l y$ is {\em not} necessary
for the production of such a barb: indeed,
$\context {C_{{\it \l x\bullet\l y}}} Q\redar w\bullet w \bnf \pro w$,
but in doing so $Q$ performs $\l z$ instead of $\l x \bullet \l y$.
\end{exa}

\begin{exa}{Example 2}
\label{ex:pro}
Consider the processes
\[
P = \pro n \pre \zero
\qquad \mbox{and} \qquad
Q = n\pre \zero
\]
together with the context $\context {C_{{\it \pro n}}} \cdot = n\pre\pro w$, for $w$ fresh.
Although $\context {C_{{\it \pro n}}} P\redar\pro w$ and the latter process exhibits a barb over $w$,
the exhibition of action $\pro n$ is {\em not} necessary for the production of such a barb:
$\context {C_{{\it \pro n}}} Q\redar\pro w$ also exhibits a barb on $w$,
but in doing so $Q$ performs $n$ instead of $\pro n$.
\end{exa}

\medskip

Example~1 shows that CPC pattern-unification allows binding names to be contractive: 
it is not necessary to fully decompose a pattern to bind it.
Thus a compound pattern may be bound to a single name or to more than one name in unification.
Example~2 illustrates that CPC pattern-unification on protected names only requires the other pattern know the name, but such a name is not necessarily protected in the reply.

These two observations make it clear that some patterns are more discerning than others,
i.e.\ unify with fewer patterns than others. This leads to the following definitions.

\begin{defi}
\label{def:match}
Define a {\em match} $(p,\sigma)$ to be a pattern $p$ and substitution $\sigma$ such that
${\sf bn}(p)={\sf dom}(\sigma)$.
\end{defi}

\begin{defi}
\label{def:compat}
Let $(p,\sigma)$ and $(q,\rho)$ be matches.
Define inductively that $p$ is {\em compatible} with $q$ by $\sigma$ and $\rho$, 
denoted $p,\sigma\compat q,\rho$ as follows:
$$
\begin{array}{rcll}
p,\sigma &\compat& \l y,\{\hat\sigma p/y\} & \mbox{if } {\sf fn}(p)=\emptyset\\
n,\{\} &\compat& n,\{\}\\
\pro n ,\{\} &\compat& \pro n ,\{\}\\
\pro n ,\{\}&\compat& n,\{\}\\
p_1\bullet p_2,\sigma_1\cup\sigma_2 &\compat& q_1\bullet q_2,\rho_1\cup\rho_2
                \quad&  \mbox{if } p_i,\sigma_i\compat q_i,\rho_i \mbox{, for } i \in \{1,2\}
\end{array}
$$
\end{defi}

The idea behind this definition is that a pattern $p$ is compatible
with another pattern $q$ with substitutions $(\sigma,\rho)$ if and only if every other pattern $r$ that
unifies $p$ by some substitutions $(\theta,\sigma)$ also unifies with
$q$ with substitutions $(\theta,\rho)$.
That is, the patterns that unify with $p$ are a subset of
the patterns that unify with $q$.  This will be proved later in
Proposition~\ref{lem:pat-lessthan}.

The compatibility relation on patterns provides the concept of proper reply in the challenge-reply game.

\begin{defi}[Bisimulation]
\label{def:bisim}
A symmetric binary relation on CPC processes $\Re$ is a bisimulation if,
for every $(P,Q) \in \Re$ and $P \ltsred\mu P'$, it holds that:
\begin{itemize}
  \item if $\mu = \tau$, then $Q \ltsred\tau Q'$, for some $Q'$ such that $(P',Q') \in \Re$;
  \item if $\mu = \rest n p$ and $({\sf bn}(p)\cup \wt n)\cap {\sf fn}(Q)=\emptyset$,
		then, for all matches $(p,\sigma)$ with 
        ${\sf fn}(\sigma)\cap \wt n = \emptyset$,
        there exist a match $(q,\rho)$ and $Q'$ 
		such that
        $p,\sigma \compat q,\rho$ and
        $Q\ltsred {\rest n q}Q'$ 
        and $(\sigma P',\rho Q')\in \Re$.
\end{itemize}
Denote by $\bisim$ the largest bisimulation closed under any substitution.
\end{defi}

The definition is inspired by the early bisimulation congruence for the
$\pi$-calculus \cite{milner.parrow.ea:calculus-mobile} (actually, it is what
in \cite{sangiorgi.walker:theory-mobile} is called {\em strong full bisimilarity} -- see Definition 2.2.2 therein):
	for every possible instantiation $\sigma$ of the binding names,
	there exists a proper reply from $Q$.
Of course, $\sigma$ cannot be chosen arbitrarily:
	it cannot use in its range names that were opened by $P$.
Also the action $\mu$ cannot be arbitrary, as in the $\pi$-calculus:
	its restricted and binding names cannot occur free in $Q$.

Unlike the $\pi$-calculus, however, the reply from $Q$ can be
different from the challenge from $P$:
	this is due to the fact that contexts in CPC are not powerful enough to
	enforce an identical reply (as highlighted in Examples~1 and~2).
Indeed, this notion of bisimulation allows a challenge $p$ to be replied to by
any compatible $q$, provided that $\sigma$ is properly adapted (yielding $\rho$, as
described by the compatibility relation) before being applied to $Q'$.
This feature somehow resembles the symbolic characterization of 
open bisimilarity given in \cite{San96,BM08}. There, labels are pairs made up
of an action and a set of equality constraints. A challenge can be replied to by a 
smaller (i.e.~less constraining) set. However, the action in the reply must be 
the same (in \cite{San96}) or becomes the same once we apply the name identifications
induced by the equality constraints (in \cite{BM08}).

An alternative approach may consider a standard bisimulation defined on top of an LTS that directly captures
the difficulties of Examples 1 and 2.
That is, allow $\l z\to z$ to reduce with label $\l x\bullet \l y$ to $x\bullet y$; similarly, allow
$n\to 0$ to reduce with label $\pro n$ to $0$. For example, this would allow the transition 
$\l w\bullet \l x\bullet n \to P \ltsred{\l w\bullet (\l y\bullet \l z)\bullet \pro n} \{y\bullet z/x\}P$.
The difficulty with this approach is that, for every binding name in a pattern, there would be an
infinite collection of transitions. For example, $\l x \to P$ would have transitions $\l x\to P\ltsred\mu \sigma P$, 
where $\mu$ can be $\l x_1\bullet \ldots\bullet \l x_i$ (for every $i \geq 1$), but also
$(\l y_1\bullet \ldots\bullet \l y_j)\bullet \l z$ and $\l y\bullet(\l z_1\bullet\ldots\bullet \l z_k)$,
and so forth (with $\sigma$ adapted appropriately).
This would make working with the LTS very heavy (the LTS would not be finitely branching anymore);
so the simplicity of relating patterns by compatibility is used here.

\subsection{Properties of Compatibility}

This section considers some properties of the compatibility relation on patterns introduced in Section~\ref{sec:bis}; 
they are formalised for later exploitation, even though some of them also illustrate some general features of patterns.
In particular, we show that compatibility preserves information used for barbs, is stable under substitution, is reflexive and transitive.

\begin{lem}
\label{prop:compat-fn}
If $p,\sigma\compat q,\rho$ then ${\sf fn}(p)={\sf fn}(q)$ and ${\sf vn}(p) \subseteq {\sf vn}(q)$ 
and ${\sf pn}(q) \subseteq {\sf pn}(p)$.
\end{lem}
\usedby{Lemma \ref{lem:sound-one} and \ref{lem:bisim-nuR}}
\begin{proof}
By definition of compatibility and induction on the structure of $q$.
\end{proof}

Given two substitutions $\sigma$ and $\theta$, denote with $\theta[\sigma]$ 
the composition of $\sigma$ and $\theta$,
with domain limited to the domain of $\sigma$, 
i.e.\ the substitution mapping every $x\in{\sf dom}(\sigma)$ to $\theta(\sigma(x))$.

\begin{lem}
\label{lem:compat-compose-subs}
If $p,\sigma \compat q,\rho$ then $p,\theta[\sigma] \compat q,\theta[\rho]$, for every $\theta$.
\end{lem}
\usedby{Lemma \ref{lem:succ-beq} and Thms \ref{thm: complete} and \ref{eq-grrr}}
\begin{proof}
By induction on the structure of $q$. The only interesting base case is when $q = \l y$.
Since ${\sf dom}(\theta[\sigma]) = {\sf dom}(\sigma)$, we have that
$p,\theta[\sigma] \compat q,\vartheta$, for $\vartheta = \{ \widehat{\theta[\sigma]}(p)/y\}
= \{\theta(\rho(y))/y\} = \theta[\rho]$.
\end{proof}

\begin{prop}[Compatibility is reflexive]
\label{prop:compat-reflexive}
For all matches $(p,\sigma)$, it holds that $p,\sigma\compat p,\sigma$.
\end{prop}
\usedby{Lemma \ref{lem:bisim-case}, Lemma \ref{lem:bisim-nuR} and Thms \ref{thm: complete} and \ref{eq-grrr}}
\begin{proof}
By definition of compatibility.
\end{proof}

\begin{prop}[Compatibility is closed under substitution]
\label{prop:compat-sub-closed}
If $p,\sigma\compat q,\rho$ then for all substitutions $\theta$ there exists
$\sigma '$ and $\rho '$ such that
$\theta p, \sigma '\compat \theta q,\rho '$.
\end{prop}
\usedby{Lemma~\ref{lem:bisim-case}}
\begin{proof}
By induction on the structure of $q$.
\end{proof}

\begin{prop}[Compatibility is transitive]
\label{prop:trans-pattern}
If $p,\sigma \compat q,\rho$ and $q,\rho \compat r,\theta$ then $p,\sigma \compat r,\theta$.
\end{prop}
\usedby{Lemma \ref{lem:trans-bisim}}
\begin{proof}
By induction on $r$. We have three possible base cases:
\begin{itemize}
\item $r = \l z$: in this case, $q = \l y_1 \bullet \ldots \bullet \l y_n$, for some $n \geq 1$, and
	$$\theta = \{\hat\rho q/z\} = \{\rho(y_1) \bullet \ldots \bullet \rho(y_n) / z\}.$$ 
	Again by definition of compatibility,
	$p = \l x_1^1 \bullet \ldots \bullet \l x_1^{k_1} \bullet \ldots \bullet \l x_n^1 \bullet \ldots \bullet \l x_n^{k_n}$,
	for some $k_1,.., k_n \geq 1$, and 
	$$\rho = \{\hat\sigma (x_i^1 \bullet \ldots \bullet x_i^{k_i})/y_i\}_{i = 1,..,n} = 
	\{\sigma(x_i^1) \bullet \ldots \bullet \sigma(x_i^{k_i})/y_i\}_{i = 1,..,n}.$$
	Thus, $\theta = \{\sigma(x_1^1) \bullet \ldots \bullet \sigma(x_1^{k_1}) \bullet \ldots \bullet
	\sigma(x_n^1) \bullet \ldots \bullet \sigma(x_n^{k_n})/z\} = \{\hat\sigma p/z\}$
	and $p,\sigma \compat r,\theta$, as desired.
\item $r = \pro n$: in this case $q = \pro n$ and $\theta = \rho = \{\}$. Again by compatibility,
	$p = \pro n$ and $\sigma = \{\}$; thus $p,\sigma \compat r,\theta$.
\item $r = n$: in this case $q$ can either be $\pro n$ or $n$, and $\theta = \rho = \{\}$. Again by compatibility,
	$p = \pro n$ or $p = n$ (this is possible only when $q = n$), and $\sigma = \{\}$; in all cases, 
	$p,\sigma \compat r,\theta$.
\end{itemize}
For the inductive step, let $r = r_1 \bullet r_2$. By compatibility, $q = q_1 \bullet q_2$ and
$\theta = \theta_1 \cup \theta_2$ and $\rho = \rho_1 \cup \rho_2$, with
$q_i,\rho_i \compat r_i,\theta_i$, for $i = 1,2$. Similarly, $p = p_1 \bullet p_2$ and 
$\sigma = \sigma_1 \cup \sigma_2$, with $p_i,\sigma_i \compat q_i,\rho_i$, for $i = 1,2$. 
By two applications of the inductive hypothesis, we obtain $p_i,\sigma_i \compat r_i,\theta_i$, for $i = 1,2$,
and by definition of compatibility we can conclude.
\end{proof}

The next result captures the idea behind the definition of
compatibility: the patterns that unify with $p$ are a subset of the
patterns that unify with $q$.

\begin{prop}
\label{lem:pat-lessthan}
If $p,\sigma\compat q,\rho$ then, for every $r$ and $\theta$ such that $\{r\pmatch p\}=(\theta,\sigma)$,
we have that $\{r\pmatch q\}=(\theta,\rho)$.
\end{prop}
\proof
The proof is by induction on $q$. There are three possible base cases:
\begin{itemize}
	\item If $q = \l y$ then ${\sf fn}(p)=\emptyset$ and $\rho = \{\hat\sigma p/y\}$; for the unification of
				$r$ and $p$ to be defined, it must be $\theta = \{\}$ and
				$\sigma=\{r_i/x_i\}_{x_i\in{\sf bn}(p)}$, each $r_i$ is communicable and $\hat\sigma p=r$.
				It follows that $\{r\pmatch q\}=(\{\},\{r/y\})=(\{\},\{\hat\sigma p/y\}) = (\theta, \rho)$.
	\item If $q = \pro n$ then $p = \pro n$ and $\sigma = \rho = \{\}$.
				For $r$ to unify with $p$, it must be that $r$ is $n$ or $\pro n$; in both cases $\theta = \{\}$.
				Hence, $\{r\pmatch q\}=(\{\},\{\})=(\theta,\rho)$.
	\item If $q = n$ then $p$ is either $n$ or $\pro n$, and $\sigma = \rho = \{\}$. 
				In both cases, $r$ can as well be either $n$ or $\pro n$.
				The proof is similar to the previous case.
\end{itemize}
For the inductive step, $q = q_1\bullet q_2$; by comparability, $p = p_1\bullet p_2$.
There are two possible cases for $r$ to unify with $p$:
	\begin{itemize}
	\item If $r = \l z$, then $p$ must be communicable and $\theta = \{p/z\}$; 
		thus, by definition of comparability, $q=p$ and $\sigma = \rho = \{\}$.
		Hence, $\{r\pmatch q\}=(\{q/z\},\{\})=(\{p/z\},\{\})=(\theta,\rho)$.
	\item Otherwise, for $r$ to unify with $p$, it must be $r = r_1\bullet r_2$ 
		with $\{r_i \pmatch p_i\} = (\theta_1,\sigma_i)$, for $i \in \{1,2\}$,
		and $\sigma = \sigma_1 \uplus \sigma_2$ and $\theta = \theta_1 \uplus \theta_2$. 
		Conclude by two applications of induction and by definition of compatibility.\qed
	\end{itemize}

\noindent Notice that the converse does not hold. Take $p = n$ and $q = \pro n$; we have that, for every
$r$ such that $\{r\pmatch p\}=(\theta,\sigma)$, we have that $\{r\pmatch q\}=(\theta,\rho)$ (the only
such $r$'s are $n$ and $\pro n$, for which $\sigma = \theta = \rho = \{\}$); however, $n , \{\} \not\compat
\pro n , \{\}$.

The following result is a variation of the previous lemma, that fixes $\sigma$ to ${\sf id}_{{\sf bn}(p)}$ but
allows an arbitrary substitution in the unification with $r$.

\begin{prop}
\label{lem:compat-match}
If $p,{\sf id}_{{\sf bn}(p)} \compat q,\rho$ and $\{p\pmatch r\}=(\vartheta,\theta)$,
then $\{q\pmatch r\}=(\vartheta[\rho],\theta)$.
\end{prop}
\usedby{Thm \ref{eq-grrr}}
\proof
By induction on $q$. There are three possible base cases:
\begin{itemize}
\item $q = \l x$: by Definition~\ref{def:compat}, it must be that ${\sf fn}(p) = \emptyset$, i.e.
$p = \l x_1 \bullet \ldots \bullet \l x_k$, for some $k$. Thus, $\rho = \{x_1 \bullet \ldots \bullet x_k / x\}$,
$r = r_1 \bullet \ldots \bullet r_k$ communicable, $\vartheta = \{r_1/x_1 \ldots r_k/x_k\}$ and $\theta = \{\}$.
By definition of unification, $\{q\pmatch r\}=(\{r/x\},\{\})$ and conclude, since
$\{r/x\} = \{r_1 \bullet \ldots \bullet r_k / x\} = \vartheta[\rho]$.

\item $q = n$: in this case, it must be either $p = n$ or $p = \pro n$; in both cases, $\rho = \{\}$.
If $p = n$, then conclude, since $q = p$ and $\vartheta[\rho] = \vartheta$.
If $p = \pro n$, obtain that $r$ can be either $n$ or $\pro n$; in both cases $\vartheta = \theta = \{\}$
and conclude.

\item $q = \pro n$: in this case $p = \pro n$, $\rho = \{\}$ and work like in the previous case.
\end{itemize}
For the inductive case, $q = q_1 \bullet q_2$; thus, by Definition~\ref{def:compat},
$p = p_1 \bullet p_2$ and $p_i,{\sf id}_{{\sf bn}(p_i)} \compat q_i,\rho_i$, where
$\rho_i = \rho|_{{\sf bn}(q_i)}$, for $i \in \{1,2\}$. We have two possibilities for $r$:
\begin{itemize}
\item $r = r_1 \bullet r_2$, where $\{p_i\pmatch r_i\}=(\vartheta_i,\theta_i)$, for $i \in \{1,2\}$; 
moreover, $\vartheta = \vartheta_1 \uplus \vartheta_2$
and $\theta = \theta_1 \uplus \theta_2$. Apply induction two times and to conclude.

\item $r = \l x$ and $p$ is communicable; thus, $\vartheta = \{\}$ and $\theta = \{p/x\}$.
  By definition of compatibility, $q = p$ and $\rho = \{\}$. Conclude
  with $\{\} = \{\}[\{\}]$.\qed
\end{itemize}

\noindent As compatibility is an ordering on matches,
it is interesting to observe that, for every pattern $p$, there is a unique (up to $\alpha$-conversion) 
maximal pattern w.r.t. $\compat$. Note that, as in Proposition~\ref{lem:compat-match} the substitution can be 
fixed (or indeed entirely elided).

\begin{prop}
\label{prop:maximal}
For every pattern $p$ there exists a maximal pattern $q$ with respect to $\compat$;
this pattern is unique up-to $\alpha$-conversion of binding names.
\end{prop}
\usedby{Lemma \ref{lem:bisim-rep}}
\begin{proof}
The proof is by induction on the structure of $p$:
\begin{itemize}
	\item If ${\sf fn}(p)=\emptyset$, then the largest $q$ w.r.t. $\compat$ is $\l y$ for some fresh $y$.
	\item If $p$ is $n$ or $\pro n$, then the largest $q$ w.r.t. $\compat$ is $n$.
	\item Otherwise, if $p=p_1\bullet p_2$, then proceed by induction on $p_1$ and $p_2$.
\end{itemize}
The only arbitrary choice is the $y$ used in the first item, that can be $\alpha$-converted to
any other fresh name.
\end{proof}

To conclude the properties of the compatibility, it is worth remarking that it does not yield a lattice: there is no supremum for the two patterns $\l x$ and $n$.

\subsection{Soundness of the Bisimulation}
\label{sec:sound}

This section proves soundness by showing that the bisimulation relation is included in barbed congruence.
This is done by showing that the bisimilarity relation is an equivalence, it is
barb preserving, reduction closed and context closed.

\begin{lem}
\label{lem:trans-bisim}
If $P\bisim Q$ and $Q\bisim R$ then $P\bisim R$.
\end{lem}
\usedby{Lemma \ref{lem:bisim-rep}}
\begin{proof}
Standard, by exploiting Proposition~\ref{prop:trans-pattern}.
\end{proof}

\begin{lem}
\label{lem:sound-one}
$\bisim$ is barb preserving.
\end{lem}
\usedby{Thm \ref{thm: sound}}
\begin{proof}
Straightforward by Lemmata~\ref{prop:compat-fn} and~\ref{lem:lts-exhibit-p}, and by definition of the LTS.
\end{proof}

Closure under any context is less easy to prove.
The approach here is as follows:
prove bisimilarity is closed under case prefixing,
name restriction and parallel composition; finally, prove closure under replication.
Proofs of these lemmata are in Appendix B, because they are adaptions
of the corresponding results for the $\pi$-calculus.
These three results will easily entail closure under arbitrary contexts (Lemma~\ref{lem:sound-three}).

\begin{lem}
\label{lem:bisim-case}
If $P\bisim Q$ then $p \to P \bisim p \to Q$.
\end{lem}
\usedby{Lemma \ref{lem:sound-three}}

\begin{lem}
\label{lem:bisim-nu}
If $P\bisim Q$ then $\rest n P \bisim \rest n Q$.
\end{lem}
\usedby{Lemma \ref{lem:bisim-rep}, Lemma \ref{lem:sound-three} and Thm \ref{eq-grrr}}

\begin{lem}
\label{lem:bisim-par}
If $P\bisim Q$ then $P\bnf R \bisim Q\bnf R$.
\end{lem}
\usedby{Lemma \ref{lem:bisim-rep}, Lemma \ref{lem:sound-three} and Thm \ref{eq-grrr}}

\begin{lem}
\label{lem:bisim-rep}
If $P\bisim Q$ then $!P\bisim\,\, !Q$.
\end{lem}
\usedby{Lemma \ref{lem:sound-three}}

\begin{lem}
\label{lem:sound-three}
$\bisim$ is contextual.
\end{lem}
\usedby{Thm \ref{thm: sound}}
\begin{proof}
Given two bisimilar processes $P$ and $Q$, it is necessary to show that for any context 
$\context C \cdot$ it holds that $\context C P \bisim \context C Q$.
The proof is by induction on the structure of the context.
If $\context C \cdot\define \cdot$ then the result is immediate.
For the inductive step, we reason by case analysis on the outer operator of $\context C \cdot$:
then, we simply use the inductive hypothesis and Lemmata~\ref{lem:bisim-case}, \ref{lem:bisim-nu},
\ref{lem:bisim-par} and~\ref{lem:bisim-rep}, respectively.
\end{proof}

Finally, we have to prove that bisimilarity is reduction closed; to this aim, we first need
to prove that structural congruence is contained in bisimilarity.

\begin{lem}
\label{lem:struct-is-a-bisim}
$\equiv$ is a bisimulation closed under substitutions.
\end{lem}
\usedby{Thm \ref{lem:sound-two}}
\proof
For every structural axiom $LHS \equiv RHS$ it suffices to show that $\{(LHS,RHS)\} \cup \bisim$
is a bisimulation. Closure under contexts follows from Lemma~\ref{lem:sound-three}. Closure under
substitutions follows from the fact that the axioms only involve bound names. The only exception 
is the rule for scope extension; however, the fact that substitution application is capture-avoiding
allows us to conclude.
\qed

\begin{lem}
\label{lem:sound-two}
$\bisim$ is reduction closed.
\end{lem}
\usedby{Thm \ref{thm: sound}}
\begin{proof}
By Proposition~\ref{prop:tau-red} and Lemma~\ref{lem:struct-is-a-bisim}.
\end{proof}

The soundness of bisimilarity w.r.t. barbed congruence now follows.

\begin{thm}[Soundness of bisimilarity]
\label{thm: sound}
$\bisim\ \subseteq\ \beq$.
\end{thm}
\begin{proof}
Lemma~\ref{lem:sound-one}, Lemma~\ref{lem:sound-three} and Lemma~\ref{lem:sound-two}
entail that $\bisim$ satisfies the conditions of Definition~\ref{def:barb-con}.
\end{proof}

\subsection{Completeness of the Bisimulation}

Completeness is proved by showing that barbed congruence is a bisimulation.
There are two results required:
showing that barbed congruence is closed under substitutions,
and showing that, for any challenge, a proper reply can be obtained via closure under an appropriate context.
To this aim, we define notions of success and failure that can be reported.
A fresh name $w$ is used for reporting success, with a barb $\barb{w}$ indicating success, and $\suc_w$ indicating a reduction sequence that can eventually report success. Failure is handled similarly using the fresh name $f$.
A process $P$ {\em succeeds} 
if $P\suc_w$ and $P\not\suc_f$;
$P$ is {\em successful} if $P \equiv \rest n(\pro w\bullet p\ |\ P')$, for some $\wt n$ and $p$ and $P'$ 
such that $w \not\in \wt n$ and $P'\not\suc_f$.

The next lemma shows that barbed congruence is closed under any substitution.

\begin{lem}
\label{lem:bcon-sub}
If $P\beq Q$ then $\sigma P\beq\sigma Q$, for every $\sigma$.
\end{lem}
\usedby{Lemma \ref{lem:succ-beq} and Thm \ref{thm: complete}}
\begin{proof}
Given a substitution $\sigma$, choose patterns $p$ and $q$ such that $\{p\pmatch q\}=(\sigma,\{\})$;
to be explicit, $p = \l x_1 \bullet \ldots\bullet \l x_k$ and $q = \sigma(x_1) \bullet \ldots \bullet \sigma(x_k)$,
for $\{x_1,\ldots,x_k\} = {\sf dom}(\sigma)$.
Define $\context C \cdot \define p\to \cdot\ \bnf q$;
by context closure, $\context C P\beq\context C Q$.
By reduction closure, the reduction $\context C P\redar \sigma P$ can be replied to 
only by $\context C Q\redar \sigma Q$; hence, $\sigma P\beq \sigma Q$, as desired.
\end{proof}

The other result to be proved is that challenges can be tested for a proper reply by a context.
When the challenge is an internal action, the reply is also an internal action; thus, the empty 
context suffices, as barbed congruence is reduction closed.
The complex scenario is when the challenge is a pattern together with a set of restricted names, 
i.e., a label of the form $\rest n p$.
Observe that in the bisimulation such challenges also fix a substitution $\sigma$ whose domain is the binding names of $p$.
Most of this section develops a reply for a challenge of the form $(\rest n p,\idsub_{{\sf bn}(p)})$;
the general setting (with an arbitrary $\sigma$) will be recovered in Theorem~\ref{thm: complete}
by relying on Lemma~\ref{lem:compat-compose-subs}.

The context for forcing a proper reply for a challenge of the form $(\rest n p,\idsub_{{\sf bn}(p)})$ is developed in three steps.
The outcome will be a process that interacts by some pattern $p'$ and reduces to a collection of tests $T$ such that $\theta T$ succeeds
if and only if $\{p\pmatch p'\}=(\sigma,\theta)$.
The first step presents the {\em specification} of a pattern and a set of names $N$ (to be thought of as the free names of the processes being compared for bisimilarity); this is the information required to build a reply context.
The second step develops auxiliary processes to test specific components of a pattern, based on information from the specification.
The third step combines these into a reply context that succeeds if and only if it interacts 
with a process that exhibits a proper reply to the challenge.

For later convenience, we define the {\em first projection} ${\sf fst}(-)$ and
{\em second projection} ${\sf snd}(-)$ of a set of pairs: e.g.,
${\sf fst}(\{(x,m),(y,n)\})=\{x,y\}$ and ${\sf snd}(\{(x,m),(y,n)\})=\{m,n\}$, respectively.

\begin{defi}
\label{def:bchar}
The {\em specification} ${\sf spec}^{N}(p)$ of a pattern $p$ with respect to a finite set of names $N$ is defined follows:
$$
\begin{array}{rcl}
\vspace{.2cm}
{\sf spec}^{N}(\l x) &=& x,\{\},\{\} \\
\vspace{.2cm}
{\sf spec}^{N}(n) &=& 
\left\{\!\!\!
\begin{array}{ll}
\vspace{.2cm}
\l x,\{(x,n)\},\{\} &\mbox{if $n\in N$ and $x$ is fresh for $N$ and $p$}\\
\l x,\{\},\{(x,n)\} &\mbox{if $n\notin N$ and $x$ is fresh for $N$ and $p$}
\end{array}
\right.
\\
{\sf spec}^{N}(\pro n) &=& \pro n,\{\},\{\}\\
{\sf spec}^{N}(p\bullet q) &=& p'\bullet q',F_p\uplus F_q,R_p\uplus R_q \ \mbox{ if }\left\{\!\!
\begin{array}{l}
{\sf spec}^{N}(p) = p',F_p,R_p\\
{\sf spec}^{N}(q) = q',F_q,R_q
\end{array}\right.
\end{array}
$$
where $F_p\uplus F_q$ denotes $F_p\cup F_q$, provided that ${\sf fst}(F_p) \cap {\sf fst}(F_q) = \emptyset$
(a similar meaning holds for $R_p\uplus R_q$).
\end{defi}
Observe that, since we only consider well formed patterns, the disjoint unions $F_p\uplus F_q$ and $R_p\uplus R_q$ are defined and the choices of the binding names for $p'$ are pairwise distinct.

Given a pattern $p$, the specification ${\sf spec}^{N}(p)=p',F,R$ of $p$ with respect to a set of names $N$ has three components:
\begin{enumerate}
\item $p'$, called the {\em complementary pattern}, is a pattern used to ensure that the context interacts with a 
process that exhibits a pattern $q$ such that $p$ is compatible with $q$ (via some substitutions);
\item $F$ is a collection of pairs $(x,n)$ made up by a binding name in $p'$ and the expected 
name (free in the process being tested) it will be bound to;
\item finally, $R$ is a collection of pairs $(x,n)$ made up by a binding name in $p'$ and the expected
name (restricted in the process being tested) it will be bound to.
\end{enumerate}

\noindent The specification is straightforward for binding names, protected names and compounds.
When $p$ is a variable name, $p'$ is a fresh binding name $\l x$ and the intended binding of $x$ to $n$ 
is recorded in $F$ or $R$, according to whether $n$ is free or restricted, respectively.

\begin{lem}
\label{prop:spec-pp}
Given a pattern $p$ and a finite set of names $N$, let ${\sf spec}^N(p)=p',F,R$.
Then, $\{p\pmatch p'\}=(\idsub_{{\sf bn}(p)},\{n/x\}_{(x,n) \in F \cup R})$.
\end{lem}
\usedby{Thm \ref{thm:lts-2-reply-succeed}}
\begin{proof}
By straightforward induction on the structure of $p$.
\end{proof}
\enlargethispage{\baselineskip}

To simplify the definitions, let $\prod_{x\in S} {\mathcal P}(x)$ be the parallel composition of 
processes ${\mathcal P}(x)$, for each $x$ in $S$. 
The tests also exploit a check ${\sf check}(x,m,y,n,w)$ to ensure equality or 
inequality of name substitutions:
\[
{\sf check}(x,m,y,n,w) = 
\left\{
\begin{array}{ll}
\res z (\pro z\bullet \pro x\bnf \pro z\bullet \pro y\to \pro w) & \mbox{if }  m=n
\vspace*{.2cm}\\
\pro w\bnf \res z (\pro z\bullet \pro x\bnf \pro z\bullet \pro y\to \pro f\bullet\l z) &\mbox{otherwise}
\end{array}
\right.
\]
Observe that failure here is indicated by pattern $\pro f \bullet \l z$; in this way,
two failure barbs cannot unify and so they cannot disappear during computations.

\begin{defi}[Tests]
Let $w$ and $f$ be fresh names, i.e.\ different from all the other names around. Then define:
\begin{equation*}
\begin{array}{rcl}
{\sf free}(x,n,w) &=& \res m (\pro m\bullet \pro n\to \pro {w}\bnf \pro m\bullet\pro x)
\\
\
\\
{\sf rest}^{N}(x,w) &=& \pro {w}\bnf\res m \res z (\ \pro m\bullet x\bullet z\\
 & & \qquad \qquad \bnf \pro m\bullet (\l y_1\bullet \l y_2)\bullet \l z\to \pro f\bullet \l z 
\\
 & & \qquad \qquad \bnf \prod_{n\in N} \pro m\bullet \pro n\bullet \l z\to \pro f\bullet \l z\ )
\\
\
\\
{\sf equality}^R(x,m,w) &=& \rest {w_y}(\ \pro {w_{y_1}}\to \ldots \to \pro {w_{y_i}} \to \pro w\\
\vspace*{.2cm}& & \qquad\quad\bnf \prod_{(y,n)\in R} {\sf check}(x,m,y,n,w_y)\ )\\
& \multicolumn{2}{l}{\mbox{where $\wt y = \{y_1,\ldots,y_i\}={\sf fst}(R)$}}
\end{array}
\end{equation*}
\end{defi}

The behaviour of the tests just defined is formalized by the following three results.

\begin{lem}
\label{lem:free-succ}
Let $\theta$ be such that $\{n,w\} \cap {\sf dom}(\theta) = \emptyset$; then,
$\theta({\sf free}(x,n,w))$ succeeds if and only if $\theta(x)=n$.
\end{lem}
\usedby{Lemma \ref{lem:theta-charP} and \ref{lem:reply-rename}}
\begin{proof}
Straightforward.
\end{proof}

\begin{lem}
\label{lem:rest-succ}
Let $\theta$ be such that $(N \cup \{w,f\}) \cap {\sf dom}(\theta) = \emptyset$; then,
$\theta({\sf rest}^N(x,w))$ succeeds if and only if $\theta(x) \in {\cal N} \setminus N$.
\end{lem}
\usedby{Lemma \ref{lem:theta-charP} and \ref{lem:reply-rename}}
\begin{proof}
Straightforward.
\end{proof}

\begin{lem}
\label{lem:equi-succ}
Let $\theta$ be such that $({\sf snd}(R) \cup \{w,f,m\}) \cap {\sf dom}(\theta) = \emptyset$; then,
$\theta({\sf equality}^R(x,m,w))$ succeeds if and only if, for every $(y,n)\in R$, 
$m=n$ if and only if $\theta(x)=\theta(y)$.
\end{lem}
\usedby{Lemma \ref{lem:theta-charP} and \ref{lem:reply-rename}}
\begin{proof}
In order for $\theta({\sf equality}^R(x,m,w))$ to succeed by exhibiting a barb $\pro w$, 
each check $\theta({\sf check}(x,m,y,n,w_y))$ must succeed by producing $\pro {w_y}$.
The rest of the proof is straightforward.
\end{proof}

\begin{lem}
\label{lem:exact-reduct-tests}
Let $T$ be a test and $\theta$ be a substitution such that $\theta(T)$ succeeds; 
there are exactly $k$ reductions of $\theta(T)$ to a successful process, where $k$ depends only on the structure of $T$.
\end{lem}
\begin{proof}
Straightforward for free and restricted tests, for which $k = 1$ and $k = 0$, respectively. 
For an equality test ${\sf equality}^R(x,m,w)$ it suffices to observe that each successful check has an exact 
number of reductions to succeed (1, if $m=n$, 0 otherwise) and then there is a reduction to consume the success 
barb of each check.
Thus, $k = |R|+h$, where $h$ is the number of pairs in $R$ whose second component equals $m$.
\end{proof}

From now on, we adopt the following notation: if $\wt n = n_1,\ldots,n_i$, then $\pro w \bullet \wt n$
denotes $\pro w \bullet n_1 \bullet \ldots \bullet n_i$. Moreover, $\theta(\wt n)$ denotes
$\theta(n_1),\ldots,\theta(n_i)$; hence, $\pro w \bullet \theta(\wt n)$ denotes
$\pro w \bullet \theta(n_1) \bullet \ldots \bullet \theta(n_i)$.

\begin{defi}
\label{def:char}
The {\em characteristic process} ${\sf char}^{N}(p)$ of a pattern $p$ with respect to a finite set of names $N$ is ${\sf char}^N(p)=p'\to {\sf tests}^N_{F,R}$ where ${\sf spec}^{N}(p)=p',F,R$ and 
\begin{equation*}
\begin{array}{rcll}
{\sf tests}^N_{F,R} &\define& \rest {w_x} \rest {w_y}(\\
& & \qquad\pro {w_{x_1}}\to\ldots\to\pro{w_{x_i}}\to\pro {w_{y_1}}\to\ldots\to\pro {w_{y_j}}\to\pro w\bullet \wt x\\
& & \qquad\bnf \prod_{(x,n)\in R} {\sf equality}^R(x,n,w_x)\\
& & \qquad\bnf \prod_{(y,n)\in F} {\sf free}(y,n,w_y)\\
& & \qquad\bnf \prod_{(y,n)\in R} {\sf rest}^N(y,w_y)\ )
\end{array}
\end{equation*}
where $\wt x = \{x_1,\ldots,x_i\} = {\sf fst}(R)$ and
$\wt y = \{y_1,\ldots,y_j\} = {\sf fst}(F)\cup {\sf fst}(R)$.
\end{defi}

\begin{lem}
\label{lem:theta-charP}
Let $\theta$ be such that ${\sf dom}(\theta) = {\sf fst}(F) \cup {\sf fst}(R)$; then,
$\theta({\sf tests}^N_{F,R})$ succeeds if and only if
\begin{enumerate}
\item for every $(x,n) \in F$ it holds that $\theta(x) = n$;
\item for every $(x,n) \in R$ it holds that $\theta(x) \in {\cal N}\setminus N$;
\item for every $(x,n)$ and $(y,m) \in R$ it holds that $n=m$ if and only if $\theta(x)=\theta(y)$.
\end{enumerate}
\end{lem}
\usedby{Thm \ref{thm:lts-2-reply-succeed}}
\begin{proof}
By induction on $|F \cup R|$ and Lemmata~\ref{lem:free-succ},
\ref{lem:rest-succ} and \ref{lem:equi-succ}. Indeed, by Definition~\ref{def:bchar},
${\sf fst}(F \cup R) \cap ({\sf snd}(F\cup R) \cup N) = \emptyset$; moreover, freshness of
$w$ and $f$ implies that $\{w,f\} \cap {\sf dom}(\theta) = \emptyset$.
\end{proof}

Note that the following results will consider the number of reductions required to succeed.
These are significant to proving the results in the strong setting, but unimportant in
the weak setting, i.e.\ with $\redar$ replaced by $\Redar$.

\begin{lem}
\label{lem:min-reduct-charP}
Given ${\sf char}^N(p)$ and any substitution $\theta$
such that ${\sf dom}(\theta) = {\sf fst}(F)\cup{\sf fst}(R)$ and $\theta({\sf tests}^N_{F,R})$ succeeds,
then there are exactly $k$ reduction steps $\theta({\sf tests}^N_{F,R})\redar^k \pro w\bullet \theta(\wt x)\ |\ Z$,
where $\wt x = {\sf fst}(R)$ and $Z \beq \zero$ and $k$ depends only on $F$ and $R$ and $N$; moreover,
no sequence of reductions shorter than $k$ can yield a successful process.
\end{lem}
\usedby{Prop \ref{prop:reply-context-minimum} and Thm \ref{thm:lts-2-reply-succeed}}
\begin{proof}
By induction on $|F \cup R|$ and Lemma~\ref{lem:exact-reduct-tests}.
\end{proof}

Notice that $k$ does not depend on $\theta$; thus, we shall refer to $k$ as the number of 
reductions for ${\sf tests}^N_{F,R}$ to become successful.
The crucial result we have to show is that the characterisation of a pattern $p$ with respect to a set of names $N$ can yield a reduction via a proper reply (according to Definition~\ref{def:bisim}) to the challenge $\rest n p$ when $\wt n$ does not intersect $N$.
A reply context for a challenge $(\rest n p,\idsub_{{\sf bn}(p)})$ with a finite set of names $N$ can be defined by exploiting the characteristic process.

\begin{defi}
\label{def:reply-context}
A {\em reply context} $\CopNnoarg C p N (\cdot)$ for the challenge $(\rest n p,\idsub_{{\sf bn}(p)})$ with a finite set of names $N$ such that $\wt n$ is disjoint from $N$ is defined as follows:
\begin{eqnarray*}
\CopNnoarg C p N (\cdot) &\define& {\sf char}^{N}(p)\bnf \cdot
\end{eqnarray*}
\end{defi}

\begin{prop}
\label{prop:reply-context-minimum}
Given a reply context $\CopNnoarg C p N (\cdot)$, the minimum number of reductions required for $\CopNnoarg C p N (Q)$ 
to become successful (for any $Q$) is the number of reduction steps for ${\sf tests}^N_{F,R}$ to become successful plus 1.
\end{prop}
\usedby{Thms \ref{thm:lts-2-succeed-reply} and \ref{thm: complete}}
\begin{proof}
By Definition~\ref{def:reply-context}, success can be generated only after removing the case $p'$ from ${\sf char}^{N}(p)$;
this can only be done via a reduction together with $Q$, i.e.\ $Q$ must eventually yield a pattern $q$
that unifies with $p'$. The minimum number of reductions is obtained when $Q$ already yields such a $q$, i.e.\ when $Q$ is a 
process of the form $\rest m (q\to Q_1\bnf Q_2)$, for some $\wt m$ and $q$ and $Q_1$ and $Q_2$ such that $\{p'\pmatch q\}=(\theta,\rho)$ and $\theta({\sf tests}^N_{F,R})$ succeeds. In this case, 
${\sf dom}(\theta) = {\sf bn}(p') = {\sf fst}(F\cup R)$; 
by Lemma~\ref{lem:min-reduct-charP}, $\theta({\sf tests}^N_{F,R})$ is successful after
$k$ reductions; thus, $\CopNnoarg C p N (Q)$ is successful after $k+1$ reductions, and this is
the minimum over all possible $Q$'s.
\end{proof}

Denote the number of reductions put forward by Proposition~\ref{prop:reply-context-minimum} as $\lb N p$.
The main feature of $\CopNnoarg C p N (\cdot)$ is that, when the hole is filled with a process $Q$, it holds that
$\CopNnoarg C p N (Q)$ is successful after $\lb N p$ reductions if and only if
there exist $(q,\rho)$ and $Q'$ such that $Q\ltsred{\rest n q}Q'$ and $p,\idsub_{{\sf bn}(p)}\compat q,\rho$.
This fact is proved by Propositions~\ref{thm:lts-2-reply-succeed} and~\ref{thm:lts-2-succeed-reply}.

\begin{prop}
\label{thm:lts-2-reply-succeed}
Suppose given a challenge $(\rest n p,\idsub_{{\sf bn}(p)})$, a finite set of names $N$, a process $Q$ and fresh names $w$ and $f$ such that 
$(\wt n \cup \{w,f\}) \cap N = \emptyset$ and $({\sf fn}(\rest n p) \cup {\sf fn}(Q)) \subseteq N$.
If $Q$ has a transition of the form $Q\ltsred{\rest n q}Q'$ and there is a substitution $\rho$ such that
$p,\idsub_{{\sf bn}(p)}\compat q,\rho$ then
$\CopNnoarg C p N (Q)$ succeeds and has a reduction sequence
$\CopNnoarg C p N (Q)\redar^k \rest n(\rho Q'\bnf\pro w\bullet\wt n\ |\ Z)$, where $k=\lb N p$ and $Z \beq \zero$.
\end{prop}
\usedby{Lemma \ref{lem:succ-beq} and Thm \ref{thm: complete}}
\begin{proof}
We assume, by $\alpha$-conversion, that binding names of $p$ are fresh, in particular do not appear in $Q$.
By Lemma~\ref{prop:spec-pp} $\{p\pmatch p'\}=(\sigma,\theta)$ where $\sigma=\idsub_{{\sf bn}(p)}$ and $\theta= \{n/x\}_{(x,n) \in F \cup R}$.
By Proposition~\ref{lem:pat-lessthan} $\{q\pmatch p'\}=(\rho,\theta)$; thus $\CopNnoarg C p N (Q)\redar \rest n (\rho Q'\bnf \theta ({\sf tests}^N_{F,R}))$. Since $w$ and $f$ do not appear in $Q$, the only possibility of producing a successful
process is when $\theta ({\sf tests}^N_{F,R})$ succeeds; this is ensured by Lemma~\ref{lem:theta-charP}. The thesis follows by Lemma~\ref{lem:min-reduct-charP}.
\end{proof}

The main difficulty in proving the converse result is the possibility of renaming restricted names. 
Thus, we first need a technical lemma that ensures us the possibility of having the same set of restricted 
names both in the challenge and in the reply, as required by the definition of bisimulation.

\begin{lem}
\label{lem:reply-rename}
Let $p$ and $N$ be such that ${\sf pn}(p) \subseteq N$, ${\sf bn}(p) \cap N = \emptyset$ and
${\sf spec}^{N}(p)=p',F,R$. If $q$ is such that ${\sf bn}(p) \cap {\sf fn}(q) = \emptyset$ and
$\{p'\pmatch q\}=(\theta,\rho)$ such that $\theta({\sf tests}^N_{F,R})$ succeeds, then:
\begin{itemize}
\item $|{\sf vn}(p) \setminus N| = |{\sf vn}(q) \setminus N|$; 
\item there exists a bijective renaming $\zeta$ of ${\sf vn}(q) \setminus N$ into ${\sf vn}(p) \setminus N$
such that $p,\idsub_{{\sf bn}(p)}\compat \zeta q,\rho$;
\item $\theta = \{n/x\}_{(x,n) \in F}\ \cup\ \{\zeta^{-1}(n)/x\}_{(x,n)\in R}$.
\end{itemize}
\end{lem}
\usedby{Thm \ref{thm:lts-2-succeed-reply}}
\proof
The proof is by induction on the structure of $p$. We have three possible base cases:
\begin{enumerate}
	\item If $p=\l x$, then $p' = x$ and $F = R = \emptyset$. By definition of pattern unification,
		$q \in \{x,\pro x, \l y\}$, for any $y$. Since $x \in {\sf bn}(p)$ and ${\sf bn}(p) \cap {\sf fn}(q) = \emptyset$,
		it can only be $q = \l y$; then, $\theta = \{\}$ and $\rho = \{x/y\}$. This suffices to conclude,
		since ${\sf vn}(p) \setminus N = {\sf vn}(q) \setminus N = \emptyset$ and
		$\l x,\idsub_{\{x\}} \compat \l y,\rho$.
	\item If $p = n$, then $p' = \l x$, for $x$ fresh. Let us distinguish two subcases:
		\begin{enumerate}
		\item If $n \in N$, then $F = \{(x,n)\}$ and $R =
                  \emptyset$. By definition of pattern unification,
                  $q$ must be communicable, $\rho = \{\}$ and $\theta
                  = \{q/x\}$.  Since ${\sf tests}^N_{F,R}$ only
                  contains ${\sf free}(x,n)$, by
                  Lemma~\ref{lem:free-succ} it holds that $q =
                  \theta(x) = n$. This suffices to conclude, since
                  ${\sf vn}(p) \setminus N = {\sf vn}(q) \setminus N =
                  \emptyset$ and $n, \{\} \compat n,\{\}$.

		\item If $n \not\in N$, then $F = \emptyset$ and $R = \{(x,n)\}$. Like before,
		$q$ must be communicable, $\rho = \{\}$ and $\theta = \{q/x\}$. 
		Since ${\sf tests}^N_{F,R}$ contains ${\sf rest}^N(x)$,
		by Lemma~\ref{lem:rest-succ} it holds that $q = \theta(x) = m \in {\cal N} \setminus N$;
		thus, $|{\sf vn}(p) \setminus N| = |{\sf vn}(q) \setminus N| = 1$.
		This suffices to conclude, by taking $\zeta = \{n/m\}$, since $n, \{\} \compat n,\{\}$. 
		\end{enumerate}
	\item If $p = \pro n$, then $p' = \pro n$ and $F = R = \emptyset$. By definition of pattern unification,
		$q \in \{n,\pro n\}$ and $\rho = \theta = \{\}$. In any case, 
		${\sf vn}(p) \setminus N = {\sf vn}(q) \setminus N = \emptyset$ and
		$p,\{\} \compat q,\{\}$.
\end{enumerate}
For the inductive case, let $p = p_1\bullet p_2$. By definition of specification,
$p' = p_1'\bullet p_2'$, $F = F_1 \uplus F_2$ and $R = R_1 \uplus R_2$, where
${\sf spec}^{N}(p_i)=p_i',F_i,R_i$, for $i \in \{1,2\}$. 
By definition of pattern unification, there are two possibilities for $q$:
\begin{enumerate}
\item If $q=\l z$, for some $z$, then $p'$ must be communicable and
	$\theta = \{\}$ and $\rho = \{p'/z\}$.
    If $p'$ is communicable then by definition of specification
    ${\sf vn}(p)=\emptyset={\sf vn}(q)$ and ${\sf vn}(p) \setminus N = {\sf vn}(q) \setminus N = \emptyset$
    and conclude with $p,\idsub_{{\sf bn}(p)}\compat\l z,\rho$.

\item Otherwise, it must be that $q = q_1\bullet q_2$, with
	$\{p'_i\pmatch q_i\} = (\theta_i,\rho_i)$, for $i \in \{1,2\}$; moreover,
	$\theta = \theta_1 \cup \theta_2$ and $\rho = \rho_1 \cup \rho_2$.
	Since the first components of $F_1$ and $F_2$ are disjoint (and similarly for $R_1$ and $R_2$),
	$\theta({\sf tests}^N_{F,R})$ succeeds implies that both $\theta({\sf tests}^N_{F_1,R_1})$ and
	$\theta({\sf tests}^N_{F_2,R_2})$ succeed, since every test of $\theta({\sf tests}^N_{F_i,R_i})$
	is a test of $\theta({\sf tests}^N_{F,R})$. Now, by two applications of the induction hypothesis,
	we obtain that, for $i \in \{1,2\}$:
	\begin{itemize}
	\item $|V_i| = |W_i|$, where $V_i = {\sf vn}(p_i) \setminus N$ and $W_i = {\sf vn}(q_i) \setminus N$; 
	\item there exists a bijective renaming $\zeta_i$ of $W_i$ into $V_i$
	such that $p_i,\idsub_{{\sf bn}(p_i)}\compat \zeta_i q_i,\rho_i$;
	\item $\theta_i = \{n/x\}_{(x,n) \in F_i}\ \cup\ \{\zeta_i^{-1}(n)/x\}_{(x,n)\in R_i}$.
	\end{itemize}
    We now show that every name $m\in W_i$ is in the domain of $\zeta_i$ and that $\zeta_i(m)\in V_i$.
    Further, that if $m$ is in only one of $W_i$ then $\zeta_i(m)$ only appears in the
    corresponding $V_i$, alternatively if $m$ is in both $W_1$ and $W_2$ then $\zeta_1(m) = \zeta_2(m)=n$
    for some $n\in V_1\cap V_2$.
	\begin{enumerate}
	\item {\em if $m \in W_1 \setminus W_2$, then $\zeta_1(m) \in V_1 \setminus V_2$:}\\
		by contradiction, assume that $\zeta_1(m) = n \in V_1 \cap V_2$ (indeed, 
		$\zeta_1(m) \in V_1$, by construction of $\zeta_1$). By construction of the specification,
		there exists $(x,n) \in R_1$. Moreover, since $n \in V_2$, there exists $m' \in W_2$
		such that $\zeta_2(m') = n$ but $m' \neq m$. Again by construction of the specification,
		there exists $(y,n) \in R_2$. By inductive hypothesis, $\theta_1(x) = \zeta_1^{-1}(n) = m$ 
		and $\theta_2(x) = \zeta_2^{-1}(n) = m'$. But then $\theta({\sf check}(x,n,y,n))$, that
		is part of $\theta({\sf tests}^N_{F,R})$, cannot succeed, since $\theta_1(x) \neq \theta_2(y)$
		(see Lemma~\ref{lem:equi-succ}). Contradiction.
	\item {\em if $m \in W_2 \setminus W_1$, then $\zeta_2(m) \in V_2 \setminus V_1$:}\\
		similar to the previous case.
	\item {\em if $m \in W_1 \cap W_2$, then $\zeta_1(m) = \zeta_2(m) \in V_1 \cap V_2$:}\\
		let $n_i = \zeta_i(m) \in V_i$; by construction of the specification,
		there exists $(x_i,n_i) \in R_i$. By contradiction, assume that $n_1 \neq n_2$.
		Then, $\theta({\sf check}(x_1,n_1,x_2,n_2))$, that is part of $\theta({\sf tests}^N_{F,R})$, 
		reports failure, since by induction $\theta_1(x_1) = \zeta_1^{-1}(x_1) = m = \zeta_2^{-1}(x_2) = \theta_2(x_2)$
		(see Lemma~\ref{lem:equi-succ}). Contradiction.
	\end{enumerate}
	Thus, $V_1 \cup V_2$ and $W_1 \cup W_2$ have the same cardinality; moreover, $\zeta = \zeta_1 \cup \zeta_2$ is a
	bijection between them and it is well-defined (in the sense that $\zeta_1$ and $\zeta_2$ coincide on all elements
	of ${\sf dom}(\zeta_1) \cap {\sf dom}(\zeta_2)$ -- see point (c) above). 
	Thus, $p_1,\idsub_{{\sf bn}(p_1)}\compat \zeta q_1,\rho_1$ and $p_2,\idsub_{{\sf bn}(p_2)}\compat \zeta q_2,\rho_2$;
	so, $p,\idsub_{{\sf bn}(p)}\compat \zeta q,\rho$. Moreover, 
\[\eqalign{
  \theta &= \theta_1 \cup \theta_2 \cr
&= \{n/x\}_{(x,n) \in F_1}\ \cup\ \{\zeta^{-1}(n)/x\}_{(x,n)\in R_1}
	\cup \{n/x\}_{(x,n) \in
          F_2}\ \cup\ \{\zeta^{-1}(n)/x\}_{(x,n)\in R_2}\cr
 &= \{n/x\}_{(x,n) \in F}\ \cup\ 
	\{\zeta^{-1}(n)/x\}_{(x,n)\in R},
  }
\] as desired.\qed
\end{enumerate}
\enlargethispage{\baselineskip}

\begin{prop}
\label{thm:lts-2-succeed-reply}
Suppose given a challenge $(\rest n p,\idsub_{{\sf bn}(p)})$, a finite set of names $N$, a process $Q$ and fresh names $w$ and $f$ such that ${\sf bn}(p) \cap N = (\wt n \cup \{w,f\}) \cap N = \emptyset$ and $({\sf fn}(\rest n p) \cup {\sf fn}(Q)) \subseteq N$.
If $\CopNnoarg C p N (Q)$ is successful after $\lb N p$ reduction steps,
then there exist $(q,\rho)$ and $Q'$ 
such that $Q\ltsred{\rest n q}Q'$
and $p,\idsub_{{\sf bn}(p)}\compat q,\rho$.
\end{prop}
\usedby{Lemma \ref{lem:succ-beq} and Thm \ref{thm: complete}}
\begin{proof}
By Proposition~\ref{prop:reply-context-minimum},
there must be a reduction $\CopNnoarg C p N (Q) \redar \rest m (\theta ({\sf tests}^N_{F,R})\bnf\rho Q'')$
obtained because $Q\ltsred{\rest m q'}Q''$ and $\{p'\pmatch q'\}=(\theta,\rho)$.
Since $w\notin {\sf fn}(Q,p')$ and $\CopNnoarg C p N (Q)\suc_w$, it must be that
$\theta ({\sf tests}^N_{F,R})$ succeeds; by 
Proposition~\ref{prop:reply-context-minimum}, this happens in $\lb N p -1$ reduction steps.

By hypothesis, ${\sf fn}(\rest n p) \subseteq N$; thus, ${\sf vn}(p) \setminus N = \wt n$.
Moreover, by $\alpha$-conversion, $\wt m \cap {\sf fn}(Q) = \emptyset$;
thus, by ${\sf fn}(\rest m q') \subseteq {\sf fn}(Q) \subseteq N$, we have
that ${\sf vn}(q') \setminus N = \wt m$. Since ${\sf bn}(p) \cap N = \emptyset$,
we also have that ${\sf bn}(p) \cap {\sf fn}(q') = \emptyset$; thus, we can
use Lemma~\ref{lem:reply-rename} and obtain a bijection $\zeta = \{\wt n / \wt m\}$ 
such that $p,\idsub_{{\sf bn}(p)}\compat \zeta q',\rho$; moreover, by $\alpha$-conversion,
$Q\ltsred{\rest n \zeta q'} \zeta Q''$. We can conclude by taking $q = \zeta q'$ and
$Q' = \zeta Q''$.
\end{proof}

We are almost ready to give the completeness result, we just need an auxiliary lemma that allows us to remove
success and dead processes from both sides of a barbed congruence, while also opening the scope of the names
exported by the success barb.

\begin{lem}
\label{lem:succ-beq}
Let $\rest m(P\bnf\pro w\bullet\wt m \bnf Z) \beq \rest m(Q\bnf\pro w\bullet\wt m \bnf Z)$, 
for $w\notin{\sf fn}(P,Q,\wt m)$ and $Z \beq \zero$; then $P\beq Q$.
\end{lem}
\usedby{Thm \ref{thm: complete}}
\begin{proof}
By Theorem~\ref{thm: sound}, it suffices to prove that
$$
\begin{array}{ll}
\Re = \{(P,Q)\ : & \rest m(P\bnf\pro w\bullet\wt m \bnf Z) \beq \rest m(Q\bnf\pro w\bullet\wt m \bnf Z)
\\
& \wedge\ w\notin{\sf fn}(P,Q,\wt m)\ \wedge\ Z \beq \zero\}
\end{array}
$$
is a bisimulation. Consider the challenge $P \ltsred\mu P'$ and reason by case analysis on $\mu$.
\begin{itemize}
\item If $\mu = \tau$, then $\rest m(P\bnf\pro w\bullet\wt m \bnf Z) \ltsred\tau 
\rest m(P'\bnf\pro w\bullet\wt m \bnf Z) = \hat P$.
By Proposition~\ref{prop:tau-red} and reduction closure, 
$\rest m(Q\bnf\pro w\bullet\wt m \bnf Z) \ltsred\tau \hat Q$ such that $\hat P \beq \hat Q$.
By Proposition~\ref{prop:free_n_match_proc} (since $w \notin {\sf fn}(Q)$) and $Z \beq \zero$, it can only be that
$\hat Q = \rest m(Q'\bnf\pro w\bullet\wt m \bnf Z)$, where $Q \ltsred\tau Q'$.
By definition of $\Re$, we conclude that $(P',Q') \in \Re$.

\item If $\mu = \rest n p$, for $({\sf bn}(p) \cup \wt n) \cap {\sf fn}(Q) = \emptyset$.
By $\alpha$-conversion, we can also assume that ${\sf bn}(p) \cap (\wt n \cup \wt m \cup {\sf fn}(P)) = \emptyset$.
Let us now fix a substitution $\sigma$ such that ${\sf dom}(\sigma) = {\sf bn}(p)$ and
${\sf fn}(\sigma) \cap \wt n=\emptyset$. Consider the context
$$
\context C \cdot =\ \cdot\ \bnf\ \pro w\bullet \wt{\l m} \to(\sigma({\sf char}^{N}(p)) \bnf 
\pro w \bullet \wt{\l n} \to \pro {w'\,} \bullet \wt n \bullet \wt m)
$$
for $w'$ fresh (in particular, different from $w$). Consider now the following sequence of reductions:
$$
\begin{array}{ll}
\multicolumn{2}{l}{\context C {\rest m(P|\pro w\bullet\wt m|Z)}\vspace*{.2cm}}
\\
\vspace*{.2cm}
\redar & 
\rest m(\sigma(\CopNnoarg C p N (P)) \bnf Z \bnf
\pro w \bullet \wt{\l n} \to \pro {w'\,} \bullet \wt n \bullet \wt m)
\\
\vspace*{.2cm}
\redar^{\!_{\lb N p}} \!\!\!\!\! & 
\rest m( \rest n(\sigma P'|\pro w\bullet \wt n|\sigma Z') \bnf Z \bnf
\pro w \bullet \wt{\l n} \to \pro {w'\,} \bullet \wt n \bullet \wt m)
\\
\redar & 
\res{\withsetnot{\wt n, \wt m}{\wt n\cup\wt m}}(\sigma P'\bnf \pro {w'\,} \bullet \wt n \bullet \wt m \bnf Z|\sigma Z')
\ \ = \ \ \hat P
\end{array}
$$
The first reduction is obtained by unifying $\pro w\bullet \wt m$ with the first
case of $\context C \cdot$; this replaces the binding names $\wt m$ in the context with 
the variable names $\wt m$ and the scope of the restriction is extended consequently.
Moreover, $\sigma({\sf char}^{N}(p)) \bnf P = \sigma({\sf char}^{N}(p) \bnf P) = 
\sigma(\CopNnoarg C p N (P))$: the first equality holds because
${\sf dom}(\sigma) = {\sf bn}(p)$ and ${\sf bn}(p) \cap {\sf fn}(P) = \emptyset$;
the second equality holds by definition of reply context.
The second sequence of reductions follows by Proposition~\ref{thm:lts-2-reply-succeed}
(ensuring that $\CopNnoarg C p N (P) \redar^{\!_{\lb N p}} \rest n(P'|\pro w\bullet \wt n|Z')$, for $Z' \beq \zero$), 
Proposition~\ref{red:renaming}
and by the fact that $\sigma(\rest n(P'|\pro w\bullet \wt n|Z')) = \rest n(\sigma P'|\pro w\bullet \wt n|\sigma Z')$
(indeed, $w$ is fresh and ${\sf names}(\sigma) \cap \wt n = \emptyset$).
Moreover, notice that $\sigma Z' \beq \sigma \zero \beq \zero$, because of 
Lemma~\ref{lem:bcon-sub}.
The last reduction is obtained by unifying $\pro w\bullet \wt n$ with the 
case $\pro w\bullet \wt {\l n}$ of the context; this replaces the binding names $\wt n$ 
in the context with the variable names $\wt n$ and the scope of the restriction 
is extended consequently.

Consider now $\context C {\rest m(Q|\pro w\bullet\wt m|Z)}$; then
reduction closure yields
$\context C {\rest m(Q|\pro w\bullet\wt m|Z)} \redar^{{\lb N p+2}} \hat Q$ such
that $\hat P \beq \hat Q$. As $\hat P$ has a barb containing $w'$, so must  $\hat Q$;
by definition of $\context C \cdot$, this can happen only if $\CopNnoarg C p N (Q)$
is successful after $\lb N p$ steps. By Proposition~\ref{thm:lts-2-succeed-reply},
this entails that there exist $(q,\rho)$ and $Q'$ 
such that $Q\ltsred{\rest n q}Q'$
and $p,\idsub_{{\sf bn}(p)}\compat q,\rho$. Moreover, with a reasoning similar to that
for the reductions of $\context C {\rest m(P|\pro w\bullet\wt m|Z)}$, we can conclude
that $\hat Q = \res{\withsetnot{\wt n, \wt m}{\wt n\cup \wt m}}
(\sigma[\rho](Q') \bnf \pro {w'\,} \bullet \wt n \bullet \wt m \bnf Z|\sigma Z')$;
indeed, in this case applying  Proposition~\ref{thm:lts-2-reply-succeed} yields
$\CopNnoarg C p N (Q) \redar^{\!_{\lb N p}} \rest n(\rho Q'|\pro w\bullet \wt n|Z')$.

To state that $(\sigma P', \sigma[\rho](Q')) \in \Re$, it suffices to notice
that $Z|\sigma Z' \beq \zero$; this holds because of contextuality of barbed congruence.
Finally, Lemma~\ref{lem:compat-compose-subs} entails that $p,\sigma\compat q,\sigma[\rho]$: 
indeed, $\sigma[\idsub_{{\sf bn}(p)}] = \sigma$ because ${\sf dom}(\sigma) = {\sf bn}(p)$.
This shows that $(q,\sigma[\rho])$ and $Q'$ is a proper reply to the challenge 
$P \ltsred{\rest n p} P'$ together with $\sigma$.
\end{itemize}
Closure under substitution holds by definition of $\Re$.
\end{proof}

\begin{thm}[Completeness of the bisimulation] 
\label{thm: complete}
$\beq\ \subseteq\ \bisim$.
\end{thm}
\begin{proof}
It is sufficient to prove that, for every pair of processes $P$ and $Q$ such that $P\beq Q$ and
for every transition $P\ltsred\mu P'$, there exists a proper reply (according to the definition of 
the bisimulation) of $Q$ and the reducts are still barbed congruent.
This is straightforward when $\mu = \tau$, due to reduction closure and Proposition~\ref{prop:tau-red}.
The difficult case if when $\mu=\rest n p$, for $({\sf bn}(p)\cup\wt n) \cap {\sf fn}(Q)=\emptyset$.
In this case fix a substitution $\sigma$ such that ${\sf dom}(\sigma)={\sf bn}(p)$ and ${\sf fn}(\sigma)\cap\wt n=\emptyset$.

By Propositions~\ref{thm:lts-2-reply-succeed} and~\ref{prop:compat-reflexive},
$\CopNnoarg C p N (P)$ is successful after $k$ reduction steps, where $k=\lb N p$.
It follows by barbed congruence that $\CopNnoarg C p N (Q)$ is successful after $k$ reduction steps too;
Proposition~\ref{thm:lts-2-succeed-reply} then implies that $Q\ltsred{\rest n q}Q'$
for some $(q,\rho')$ and $Q'$ 
such that $p,\idsub_{{\sf bn}(p)}\compat q,\rho'$.

By two applications of Proposition~\ref{thm:lts-2-reply-succeed} it follows that
$\CopNnoarg C p N (P)\redar^k \rest n(P'\bnf \pro w\bullet \wt n \bnf Z)$, for $Z \beq \zero$, and
$\CopNnoarg C p N (Q)\redar^k \rest n(\rho'Q'\bnf \pro w\bullet \wt n \bnf Z)$.
Notice that, by Proposition~\ref{prop:reply-context-minimum} and definition of the reply context, 
these are the only possibilities that yield a success barb in $k$ reductions.
Furthermore, reduction closure of $\beq$ and Lemma~\ref{lem:succ-beq} imply that $P'\beq \rho ' Q'$.
By Lemma~\ref{lem:bcon-sub}, we obtain $\sigma P'\beq \sigma (\rho' Q') = \sigma[\rho'](Q')$.
By Lemma~\ref{lem:compat-compose-subs}, $p,\idsub_{{\sf bn}(p)}\compat q,\rho'$ 
implies $p,\sigma\compat q,\sigma[\rho']$. This suffices to conclude.
\end{proof}

To conclude, we want to stress that we have developed our semantic theories
in the {\em strong} setting just for the sake of simplicity. Their {\em weak} counterparts,
consisting of allowing multiple $\tau$s/reductions in every predicate, can be obtained
in the usual manner \cite{milner.parrow.ea:calculus-mobile,milner.sangiorgi:barbed-bisimulation}.
As usual, the main difficulty is in the completeness proof just shown. Indeed, to force a proper
reply via contexts, we need to use ``fresh" barbs (viz, the $w$s and $f$ in our tests). Such barbs 
have to be removed after the forced action; however, this can be done only if the left and right
hand side of the equated processes have the same shape (see Lemma~\ref{lem:succ-beq}).
This is straightforward in the strong case, where the number of reductions (viz, $\lb N p$ --
see Propositions~\ref{thm:lts-2-reply-succeed} and~\ref{thm:lts-2-succeed-reply}) ensures this property. 
By contrast, in the weak case we can stop along this sequence of $\lb N p$ reductions, since the
weak barbs will cover the missing steps. This requires a different proof of Lemma~\ref{lem:succ-beq}.
To achieve this, we can follow the traditional path: instead of having a single success barb $w$,
our tests have two (say, $w$ and $w'$) in mutual exclusion (i.e., internal choice); we then 
consider an additional reduction that excludes $w'$. This ensures that the other process must also reach
the reduction that excludes $w'$ and thus complete the proof.

\subsection{Example Equivalences}
\label{subsec:examples}

This section considers some examples where bisimulation can be used to show the equivalence of processes.
The first example exploits the unification of protected names with both variable and protected names:
\begin{eqnarray*}
\pro n\to P\bnf !n\to P &\bisim& !n\to P
\end{eqnarray*}
It states that the processes $\pro n\to P\bnf !n\to P$ can be subsumed by the more compact process $!n\to P$;
indeed, any interaction of the left hand processes can be properly responded to by the right hand process and vice versa.

The second example considers the contractive nature of binding names in CPC:
a case with the pattern $\l x\bullet\l y$ can be subsumed by a case with the pattern $\l z$ as long as some conditions are met.
For example:
\begin{equation*}
\begin{array}{rcll}
\l x\bullet \l y\to P \bnf !\l z \to Q &\bisim& !\l z\to Q &
\quad \mbox{if $P\bisim \{x\bullet y/z\}Q$}
\end{array}
\end{equation*}
The side condition requires that the bodies of the cases must be bisimilar under a substitution that preserves the structure of any pattern bound by $\l x\bullet\l y$ in the process $Q$.

These examples both arise from pattern unification and also appear in the compatibility relation.
Indeed, the examples above are instances of a general result:

\begin{prop}
\label{eq-grrr}
Let $P=p\to P'\bnf !q\to Q'$ and $Q=\ !q\to Q'$.
If there exists $\rho$ such that $p,\idsub_{{\sf bn}(p)}\compat q,\rho$ and $P'\bisim\rho Q'$, 
then $P\bisim Q$.
\end{prop}
\begin{proof}
It suffices to prove that
$$
\Re = \{(p \to P'|Q|R,Q|R)\, :\, Q =\,\, !q \to Q'\ \wedge\ 
\exists \rho\,.\, p,{\sf id}_{{\sf bn}(p)}\compat q,\rho\ \wedge\ P'\bisim\rho Q'\}\ \cup\ \bisim
$$
is a bisimulation. To this aim, consider every challenge from $p \to P'|Q|R$ and show that
there exists a transition from $Q|R$ that is a proper reply (according to the bisimulation). 
The converse (when the challenge comes from $Q \bnf R$) is easier.

Let $p \to P'|Q|R \ltsred\mu \hat P$; there are two possibilities for $\mu$:
\begin{enumerate}
\item $\mu = \rest n p'$: in this case, we also have to fix a substitution $\sigma$ such that
${\sf dom}(\sigma) = {\sf bn}(p')$ and ${\sf fn}(\sigma) \cap \wt n = \emptyset$. 
There are three possible ways for producing $\mu$:
	\begin{enumerate}[label=\({\alph*}]
	\item $\mu = p$ and $\hat P = P'|Q|R$: in this case, since the action comes from $p \to P'$,
	by the side condition of rule {\sf parext}, it must be that ${\sf bn}(p) \cap {\sf fn}(Q|R) = \emptyset$. 
	Now, consider $Q \ltsred q Q'|Q$ with ${\sf bn}(q) \cap {\sf fn}(Q|R) = \emptyset$ (ensured by
    $\alpha$-conversion); thus, $Q|R \ltsred q Q'|Q|R = \hat Q$. Let $\rho$ be such that
	$p,{\sf id}_{{\sf bn}(p)}\compat q,\rho$; by Lemma~\ref{lem:compat-compose-subs},
	$p,\sigma \compat q,\sigma[\rho]$, where $\sigma[{\sf id}_{{\sf bn}(p)}] = \sigma$
	because ${\sf dom}(\sigma) = {\sf bn}(p)$. Now it suffices to prove that
	$(\sigma \hat P, \sigma[\rho]\hat Q) \in \Re$. 
	This follows from the hypothesis that $P'\bisim\rho Q'$: indeed, by closure of $\bisim$ under
	substitutions, $\sigma P'\bisim \sigma (\rho Q') = \sigma[\rho]Q'$; by Lemma~\ref{lem:bisim-par},
	$\sigma P'|Q|R \bisim \sigma[\rho]Q'|Q|R$. Now conclude: 
	since ${\sf dom}(\sigma) = {\sf bn}(p)$ and ${\sf bn}(p) \cap {\sf fn}(Q|R) = \emptyset$, it holds that
	$\sigma \hat P = \sigma P'|Q|R$; 
	since ${\sf dom}(\sigma[\rho]) = {\sf dom}(\rho) = {\sf bn}(q)$ and ${\sf bn}(q) \cap {\sf fn}(Q|R) = \emptyset$, 
	it holds that $\sigma[\rho] \hat Q = \sigma[\rho] Q'|Q|R$;
	finally, by definition, $\bisim\ \subseteq \Re$.

	\item $\mu = q$ and $\hat P = p \to P'|Q'|Q|R$: in this case, since the action comes from $Q$,
	by the side condition of rule {\sf parext}, it must be that ${\sf bn}(q) \cap {\sf fn}(p \to P'|R) = \emptyset$. 
	Now, consider $Q|R \ltsred q Q'|Q|R = \hat Q$. By Lemma~\ref{prop:compat-reflexive},
	$q,\sigma \compat q,\sigma$. It suffices to prove that $(\sigma \hat P, \sigma \hat Q) \in \Re$. 
	This follows from the definition of $\Re$: 
	since ${\sf dom}(\sigma) = {\sf bn}(q)$ and ${\sf bn}(q) \cap {\sf fn}(p \to P'|R) = \emptyset$, it holds that
	$\sigma \hat P = p \to P'|\sigma Q'|Q|R$ and $\sigma \hat Q = \sigma Q'|Q|R$.

	\item $\mu = \rest n r$, $R \ltsred\mu R'$ and $\hat P = p \to P'|Q|R'$: in this case, 
	by the side condition of rule {\sf parext}, it must be that ${\sf bn}(r) \cap {\sf fn}(p \to P'|Q) = \emptyset$. 
	Now, consider $Q|R \ltsred \mu Q|R' = \hat Q$ and reason like in the previous case, obtaining that
	$\sigma \hat P = p \to P'|Q|\sigma R'\ \Re\ \ Q|\sigma R' = \sigma \hat Q$.
	\end{enumerate}

\item $\mu = \tau$: in this case, there are five possible ways for producing $\mu$:
	\begin{enumerate}[label=\({\alph*}]
	\item $R \ltsred\tau R'$ and $\hat P = p \to P'|Q|R'$: this case is straightforward.

	\item $\hat P = \vartheta P'|\theta (Q'|Q)|R$, where $\{p\pmatch q\} = (\vartheta,\theta)$:
	Let $\rho$ be such that $p,{\sf id}_{{\sf bn}(p)}\compat q,\rho$; by Proposition~\ref{lem:compat-match},
	$\{q\pmatch q\} = (\vartheta[\rho],\theta)$. But a pattern can unify with itself only if it
    contains no binding names;
    this entails that $\theta = \{\}$ and $\vartheta = \{\}$.
	Hence, $\hat P = P'| Q'|Q |R$ and conclude by
	taking $Q|R \ltsred\tau Q'|Q'|Q|R = \hat Q$, since by hypothesis $P' \bisim Q'$.

	\item $\hat P = \rest n(\vartheta P'|Q|\theta R')$, where $R \ltsred{\rest n r} R'$ and 
	$\{p\pmatch r\} = (\vartheta,\theta)$: by $\alpha$-conversion, now let ${\sf bn}(p) \cap {\sf fn}(Q) = \emptyset$
	and $\wt n \cap {\sf fn}(p \to P'|Q) = \emptyset$. Now consider $Q \ltsred q Q'|Q$ with 
	${\sf bn}(q) \cap {\sf fn}(Q) = \emptyset$; by Proposition~\ref{lem:compat-match}, the hypothesis
	$p,{\sf id}_{{\sf bn}(p)}\compat q,\rho$ entails $\{q\pmatch r\} = (\vartheta[\rho],\theta)$. 
	Thus, $Q|R \ltsred\tau \rest n(\vartheta[\rho](Q'|Q)|\theta R') = \rest n(\vartheta[\rho]Q'|Q|\theta R') = \hat Q$,
	where the first equality holds because ${\sf dom}(\vartheta[\rho]) = {\sf dom}(\rho) = {\sf bn}(q)$ 
	and ${\sf bn}(q) \cap {\sf fn}(Q) = \emptyset$. Conclude by using the hypothesis $P' \bisim \rho Q'$,
	thanks to closure of $\bisim$ under substitutions, parallel and restriction.

	\item $\hat P = p \to P'\ |\ \rest n(\vartheta (Q'|Q)|\theta R')$, where $R \ltsred{\rest n r} R'$ and 
	$\{q\pmatch r\} = \linebreak (\vartheta,\theta)$: this case is simple, by considering 
	$Q|R \ltsred\tau \linebreak \rest n(\vartheta (Q'|Q)|\theta R') = \hat Q$ and by observing that 
	${\sf dom}(\vartheta) = {\sf bn}(q)$, with ${\sf bn}(q) \cap {\sf fn}(Q) = \emptyset$.

	\item $\hat P = p \to P'\ |\ \vartheta Q'|\theta Q'|Q|R$, where $\{q\pmatch q\} = (\vartheta,\theta)$:
	this case is straightforward, by observing that $\vartheta = \theta = \{\}$.
	\end{enumerate}
\end{enumerate}
Closure under substitution is straightforward by Proposition~\ref{prop:compat-sub-closed}.
\end{proof}

To conclude, notice that the more general claim
\begin{quote}
Let $P=p\to P'\bnf !q\to Q'$ and $Q= \,\,!q\to Q'$; if there are $\sigma$
and $\rho$ such that $p,\sigma\compat q,\rho$ and
$\sigma P'\bisim\rho Q'$, then $P\bisim Q$
\end{quote}
does {\em not} hold. To see this, consider the following two processes:
$$
\begin{array}{lll}
P = \l x \to P'\ |\ Q & \mbox{with} & P' =  x\bnf m \to \pro w
\ \ \mbox{ for } x \neq m
\\
Q =\,\, !\l x \to Q' & \mbox{with} & Q' = m\bnf m \to \pro w
\end{array}
$$
Trivially $\l x, \{m/x\} \compat \l x, \{m/x\}$ and $\{m/x\}P' \bisim \{m/x\}Q' = Q'$;
however, $P$ is {\em not} bisimilar to $Q$. Indeed, in the context $\context C \cdot = \,\cdot\bnf k \to \zero$,
for $k \neq m$, they behave differently: $\context C P$ can reduce in one step to a process that is stuck and cannot
exhibit any barb on $w$; by contrast, every reduct of $\context C Q$ reduces in another step to a process that exhibits
a barb on $w$. (As usual, for proving equivalences it is easier to rely on bisimulation, while for proving
inequivalences it is easier to rely on barbed congruence, thanks to Theorems~\ref{thm: sound} and~\ref{thm: complete}.)
Proposition~\ref{eq-grrr} is more demanding: it does not leave us free to choose whatever $\sigma$ we want, but
it forces us working with ${\sf id}_{{\sf bn}(p)}$. Now, the only $\rho$ such that $\l x,\{x/x\} \compat \l x,\rho$
is $\{x/x\}$; with such a substitution, the second hypothesis of the theorem, in this case $P' \bisim Q'$, does not hold
and so we cannot conclude that $P \bisim Q$.

\section{Comparison with Other Process Calculi}
\label{sec:compare}

This section exploits the techniques developed in \cite{G:IC08,G:CONCUR08} to formally
assess the expressive power of CPC with respect to $\pi$-calculus, Linda, Spi calculus, Fusion and Psi calculus.
After briefly recalling these models and some basic material from \cite{G:CONCUR08}, the 
relation to CPC is formalised. First, let each model, including CPC, be augmented with
a reserved process `$\ok$', used to signal successful termination.
This feature is needed to formulate what a {\em valid} encoding is in Definition~\ref{def:ve}.

\subsection{Some Process Calculi}
\label{subsec:calculi}

\paragraph{$\pi$-calculus \cite{milner.parrow.ea:calculus-mobile,sangiorgi.walker:theory-mobile}.}
The $\pi$-calculus processes are given by the following grammar:
$$
P \ ::= \ \zero \ \bnf \ \ok \ \bnf \ \oap a b.P\ \bnf \ \iap a x.P \ \bnf \ 
(\nu n) P \ \bnf \ P | Q \ \bnf \ !P
$$
and the only reduction axiom is
$$
\oap a b .P\bnf \iap a x.Q \quad\redar\quad P\bnf \{b/x\}Q
$$
The reduction relation is obtained by closing this interaction rule by parallel, restriction and the
same structural congruence relation defined for CPC.

\paragraph{Linda \cite{Gel85}.} Consider the following variant of Linda formulated to follow 
CPC's syntax. Processes are defined as:
$$
P \ ::= \ \zero \ \bnf \ \ok \ \bnf \ \oap {} {b_1,\ldots,b_k} \ \bnf \ \iap {} {t_1,\ldots,t_k}.P \ \bnf \ 
(\nu n) P \ \bnf \ P | Q \ \bnf \ !P
$$
where $b$ ranges over names and $t$ denotes a template field, defined by:
$$
t \ ::=\ \lambda x\ \quad |\quad \ex {\,b\,}
$$
Assume that input variables occurring in templates are all distinct.
This assumption rules out template $(\l x,\l x)$, but accepts $(\l x,
\ex b,\ex b)$.  Templates are used to implement Linda's pattern
matching, defined as follows:
$$
\begin{array}{c}
\vspace*{.4cm}
\pmtch(\ ; \, ) = \{\}
\qquad\qquad
\pmtch(\ex {\,b\,} ; b) = \{\}
\qquad\qquad
\pmtch(\lambda x ; b) = \{b/x\}
\\
\Rule{}
{\pmtch(t ; b) = \sigma_1 \qquad \pmtch(\wt t ; \wt b) = \sigma_2}
{\pmtch(t,\wt t\, ;\, b,\wt b) = \sigma_1 \uplus \sigma_2}
{}
\end{array}
$$
where
`$\uplus$' denotes the union of partial functions with disjoint domains. The interaction axiom is:
$$
\oap {} {\wt b}\bnf \iap{}{\wt t}.P \redar \sigma P
\qquad \mbox{ if }\ \pmtch(\wt t ; \wt b) = \sigma
$$
The reduction relation is obtained by closing this interaction rule by parallel, restriction and the
same structural congruence relation defined for CPC.

\paragraph{Spi calculus \cite{gordon1997ccp}.} This language is unusual as names 
are now generalised to {\em terms} of the form
\begin{eqnarray*}
M,N \ &::=&\  n \ \bnf\  x \ \bnf\  (M,N)
\ \bnf\  0 \ \bnf\  i \ \bnf\  {\it suc}(M) \ \bnf\  \encr M N
\end{eqnarray*}
\change{Thomas}
{They are rather similar to the patterns of CPC in that they may have
internal structure. Of
particular interest is the pair, that combines terms and allows the
construction of arbitrary structured data,
and the encryption construct.
Pairing is distinct from the polyadic data exchange
discussed previously, as compound messages may be bound to a single name
and then
decomposed later by some intensional reduction. Similarly, the encryption 
requires they key to be known to gain access to the encrypted term.}
{They are rather similar to the patterns of CPC in that they may have internal
structure. Of particular interest are the pair, successor and encryption that
may be bound to a name and then decomposed later by an intensional reduction.
Note that $i$ denotes a natural number greater than zero, and is considered
equal to the $i$th successor of zero.}

The processes of the Spi calculus are:
\begin{eqnarray*}
P,Q &::=& \ \ 0\ \bnf \ \ok \ \bnf\ P|Q\ \bnf\  !P\ \bnf\  \res m P\ \bnf\  M(x).P\
\bnf\  \overline{M}\langle N\rangle .P\\
& &                 \ \bnf\ [M\ {\it is}\ N]P
										\ \bnf\ {\it let}\ (x,y)=M\ {\it in}\ P\\
& &                 \ \bnf\ {\it case}\ M\ {\it of}\ \encr x N:P
                    \ \bnf\ {\it case}\ M\ {\it of}\ 0:P\ {\it suc}(x):Q
\end{eqnarray*}
The null process, parallel composition, replication and restriction
are all familiar.
The input $M(x).P$ and output $\overline{M}\langle N\rangle .P$ are generalised 
from $\pi$-calculus to 
allow arbitrary terms in the place of channel names and output arguments.
The match $[M\ {\it is}\ N]P$ determines equality of $M$ and $N$.
The splitting ${\it let}\ (x,y)=M\ {\it in}\ P$ decomposes pairs.
The decryption ${\it case}\ M\ {\it of}\ \encr x N:P$ decrypts $M$ 
and binds the encrypted message to $x$.
The integer test ${\it case}\ M\ {\it of}\ 0:P\ {\it suc}(x):Q$ branches
according to the number.
\change{Thomas}
{The last four can all get stuck if $M$ is an incompatible term.}
{Note that the last four processes can all get stuck if $M$ is an incompatible term.
Furthermore, the last three are intensional, i.e.\ they depend on the internal
structure of $M$.}

Concerning the operational semantics, we
consider a slightly modified version of Spi calculus where
interaction is generalised to
$$
\overline{M}\langle N\rangle .P\bnf M(x) .Q \quad\rew\quad P\bnf \{N/x\}Q
$$
where $M$ is any term of the Spi calculus. The remaining axioms are:
$$
\begin{array}{rcl}
[M\ {\it is}\ M]P &\redar& P
\vspace*{.1cm}
\\
{\it let}\ (x,y)=(M,N)\ {\it in}\ P &\redar& \{M/x,N/y\}P
\vspace*{.1cm}
\\
{\it case}\ \{M\}_N\ {\it of}\ \encr x N:P &\redar& \{M/x\}P
\vspace*{.1cm}
\\
{\it case}\ 0\ {\it of}\ 0:P\ {\it suc}(x):Q &\redar& P
\vspace*{.1cm}
\\
{\it case}\ suc(N)\ {\it of}\ 0:P\ {\it suc}(x):Q &\redar& \{N/x\}Q
\end{array}
$$
Again, the reduction relation is obtained by closing the interaction axiom
under parallel, restriction and the structural congruence of CPC.

\paragraph{Fusion \cite{parrow.victor:fusion-calculus}.}
Processes are defined as:
$$
P\ ::=\ \zero\ \bnf\ \ok \ \bnf \ P|P\ \bnf \ (\nu x)P\ \bnf \ !P\ \bnf \ 
\oap u {\wt x}.P\ \bnf \ \iap u {\wt x}.P 
$$
The interaction rule for Fusion is taken from \cite{WG:explicit-fusions}:
$$
\begin{array}{ll}
(\nu \wt u)(\oap u {\wt x}.P\bnf \iap u {\wt y}.Q\ | \ R)
\redar \sigma P\ | \ \sigma Q\ | \ \sigma R\ 
& \mbox{with } {\sf dom}(\sigma) \cup {\sf ran}(\sigma) \subseteq \{\wt x,\wt y\}
\\
& \mbox{and } \wt u = {\sf dom}(\sigma) \setminus {\sf ran}(\sigma)
\\
& \mbox{and } \sigma(v) = \sigma(w)
\\
& \mbox{iff } (v,w) \in E(\wt x = \wt y)
\end{array}
$$
where $E(\wt x = \wt y)$ is the least equivalence relation on names
generated by the equalities $\wt x = \wt y$ (that is defined 
whenever $|\wt x| = |\wt y|$).
Fusion's reduction relation is obtained by closing the interaction axiom
under parallel, restriction and the structural congruence of CPC.

\newcommand{\cheq}{\stackrel\cdot\leftrightarrow}
\newcommand{\terms}{{\bf T}}
\newcommand{\assertion}{{\bf A}}
\newcommand{\one}{{\bf 1}}
\newcommand{\compose}{\otimes}
\newcommand{\assert}[1]{\llparenthesis \, #1 \, \rrparenthesis}
\newcommand{\fram}[1]{{\cal F}(#1)}

\paragraph{Psi \cite{BJPV11}.}
For our purposes, Psi-calculi are parametrized w.r.t. two sets: terms $\terms$,
ranged over by $M,N,\ldots$, and assertions $\assertion$, ranged over by $\Psi$.
The empty assertion is written $\one$.
We also assume two operators: channel equivalence, $\cheq \subseteq \terms \times \terms$,
and assertion composition, $\compose: \assertion \times \assertion \rightarrow \assertion$. 
It is also required that $\cheq$ is transitive and
symmetric, and that $(\compose,\one)$ is a commutative monoid.

Processes in Psi are defined as:
$$
P\ ::=\ \zero\ \bnf\ \ok \ \bnf \ P|P\ \bnf \ (\nu x)P\ \bnf \ !P\ \bnf \ 
\oap M N.P\ \bnf \ \iap M {\lambda\wt x} N.P \ \bnf\ \assert\Psi
$$
We now give a reduction semantics, by isolating the $\tau$ actions of the
LTS given in \cite{BJPV11}. To this aim, we recall the definiton of frame
of a process $P$, written $\fram P$, as the set of unguarded assertions 
occurring in $P$. Formally:
$$
\fram{\assert \Psi} = \Psi
\qquad
\fram{(\nu x)P} = (\nu x)\fram P
\qquad
\fram{P|Q} = \fram P \compose \fram Q
$$
and is $\one$ in all other cases. We denote as $(\nu \wt b_P)\Psi_P$ the frame of $P$.
The structural laws are the same as in $\pi$-calculus.
The reduction relation is inferred by the following laws:
$$
\begin{array}{c}
\prooftree \Psi \vdash M \cheq N 
\justifies \Psi \triangleright \oap M K.P\ |\ \iap N {\lambda \wt x}H.Q \redar P\ |\ \{{\wt L}/{\wt x}\}Q 
\endprooftree\ K = H[\wt x := \wt L]
\vspace*{.4cm}
\\

\prooftree \Psi \compose \Psi_Q \triangleright P \redar P'
\justifies \Psi \triangleright P\ |\ Q \redar P' \ |\ Q
\endprooftree\ \fram Q = (\nu \wt b_Q)\Psi_Q, \wt b_Q \mbox{ fresh for } \Psi \mbox{ and } P
\vspace*{.4cm}
\\

\prooftree \Psi \triangleright P \redar P'
\justifies \Psi \triangleright (\nu x) P \redar (\nu x) P'
\endprooftree\ x \not\in \mbox{names}(\Psi)
\qquad
\prooftree P \equiv Q \quad \Psi \triangleright Q \redar Q' \quad Q' \equiv P'
\justifies \Psi \triangleright P \redar P'
\endprooftree
\end{array}
$$
We write $P \redar P'$ whenever $\one \triangleright P \redar P'$.

\subsection{Valid Encodings and their Properties}

This section recalls and adapts the definition of valid encodings as well as some
useful theorems (details in \cite{G:CONCUR08}) for formally relating process calculi.
The validity of such criteria in developing expressiveness studies emerges from the
various works 
\cite{G:IC08,G:DC10,G:CONCUR08}, that have also recently inspired similar works 
\cite{LPSS10,LVF10,gla12}. 

An {\em encoding} of a language $\Lang_1$ into another language $\Lang_2$ is a pair $(\encode\cdot,\renpol)$
where $\encode\cdot$ translates every $\Lang_1$-process into an $\Lang_2$-process
and $\renpol$ maps every name (of the source language) into a tuple of $k$ names (of the target language), for $k > 0$.
The translation $\encode\cdot$ turns every term of the source language into a term of the target; in doing this,
the translation may fix some names to play a precise r\^ole 
or may translate a single name into a tuple of names. This can be obtained
by exploiting $\renpol$. 

Now consider only encodings that satisfy the following properties.
Let a {\em $k$-ary
context} $\context C {\_\,_1; \ldots; \_\,_k}$ be a term where $k$
occurrences of $\zero$ are linearly replaced by the holes $\{\_\,_1;
\ldots; \_\,_k\}$ (every one of the $k$ holes must occur once and only once).
Moreover, denote with $\redar_i$ and $\Redar_i$ 
the relations $\redar$ and $\Redar$ in language $\Lang_i$;
denote with $\redar^\omega_i$ an infinite sequence of reductions in $\Lang_i$.
Moreover, we let $\beq_i$ denote the reference behavioural equivalence for language $\Lang_i$.
Also, let $P \suc_i$ mean that there exists $P'$ such that $P \Redar_i P'$ and $P' \equiv P''\bnf \ok$,
for some $P''$.
Finally, to simplify reading, let $S$ range
over processes of the source language (viz., $\Lang_1$) and $T$ range
over processes of the target language (viz., $\Lang_2$).

\begin{defi}[Valid Encoding]
\label{def:ve}
An encoding $(\encode\cdot,\renpol)$ of $\Lang_1$ into $\Lang_2$
is {\em valid} if it satisfies the following five properties:
\begin{enumerate}
\item {\em Compositionality:} for every $k$-ary operator $\op$ of $\Lang_1$
and for every subset of names $N$,
there exists a $k$-ary context $\CopN C \op N {\_\,_1; \ldots; \_\,_k}$ of $\Lang_2$
such that, for all $S_1,\ldots,S_k$ with ${\sf fn}(S_1,\ldots,S_k) = N$, it holds
that $\encode{\op(S_1,\ldots,S_k)} = \CopN C \op N {\encode{S_1};\ldots;\encode{S_k}}$.

\item {\em Name invariance:}
for every $S$ and name substitution $\sigma$, it holds that
$$
\encode{\sigma S}\ \left\{ 
\begin{array}{ll}
\ =\ \sigma'\encode S& \mbox{ if $\sigma $ is injective}\\
\ \beq_2\ \sigma'\encode S  & \mbox{ otherwise}
\end{array}
\right.
$$
where $\sigma'$ is such that 
$\renpol(\sigma(a)) = \sigma'(\renpol(a))$
for every name $a$.

\item {\em Operational correspondence:}
\begin{itemize}
\item for all $S \Redar_1 S'$, it holds that $\encode S \Redar_2 \beq_2 \encode {S'}$;
\item for all $\encode S \Redar_2 T$, there exists $S'$ such that $S \Redar_1\!\! S'$ 
and $T \Redar_2 \beq_2\!\! \encode {S'}$.
\end{itemize}

\item {\em Divergence reflection:}
for every $S$ such that 
$\encode S \redar\!\!_2^\omega$, it holds that 
\linebreak $S$ \mbox{$\redar\!\!_1^\omega$}.

\item {\em Success sensitiveness:}
for every $S$, it holds that $S \suc_1$ if and only if $\encode S \suc_2$.
\end{enumerate}
\end{defi}

The criteria we have just presented may seem quite demanding; for example,
they do not allow schemes including parameters that are changed along the 
way of the encoding (the most notable of such examples is Milner's encoding
of $\lambda$-calculus into $\pi$-calculus \cite{Milner92}). Of course, this
makes our separation results slightly weaker and leaves room for further
improvement. Moreover, as fully explained in \cite{GN:MSCS14}, we do not 
consider full abstraction as a validity criterion for expressiveness.

Now recall some results concerning valid encodings, in particular for showing
separation results, i.e.\ for proving that no valid encoding can exist between
a pair of languages ${\mathcal L}_1$ and ${\mathcal L}_2$ satisfying certain
conditions.
Here, these languages will be limited to CPC and those introduced in Section~\ref{subsec:calculi}.
Originally valid encodings considered were assumed to be {\em semi-homomorphic}, 
i.e.\ where the interpretation of parallel
  composition is via a context of the form $\res {\wt n}(\_\,_1\bnf
  \_\,_2\bnf R)$, for some $\wt n$ and $R$ that only depend on the
  free names of the translated processes.
This assumption simplified the proofs of the following results in general, i.e.\ without
relying on any specific process calculus; in our setting, since the languages are fixed,
we can prove the same results without assuming semi-homomorphism.

\change{barry}{It is worth remarking that, since the calculi considered here have been already defined
(see Section~\ref{subsec:calculi}), we do not need any assumption on the context used
to translate the parallel operator. In \cite{G:CONCUR08} we needed to explicitly
assume that such an operator is translated ``semi"-homomorphically, i.e.\ via a context
of the form $\res {\wt n}(\_\,_1\bnf \_\,_2\bnf R)$, for some $\wt n$ and $R$ that only 
depend on the free names of the translated processes. ``Semi"-homomorphism comes for free
by Proposition~\ref{deadlock} and by the kind of interactions that can occur in our sample 
process calculi.}{}

\comments{barry}{The limitation to semi-homomorphisms that was
  implicit in [9] is here made explicit. It is worth remarking that
  the limitation to the given set of languages does not remove the
  need for this assumption in the given proofs, even though it kills
  of one particular counter-example. }

\begin{thm}
\label{autoriduz}
Assume that there exists $S$ such that $S \noredar\!\!_1$ and
$S \not\suc_1\ $ and $S\bnf S \suc_1$; 
moreover, assume that every $T$
that does not reduce is such that $T\bnf T \noredar\!\!_2$. Then, 
there cannot exist any 
valid encoding of $\Lang_1$ into $\Lang_2$.
\end{thm}

To state the following proof technique, define the {\em matching degree} 
of a language $\Lang$, written $\md\Lang$, as the least upper bound on the number 
of names that must be matched to yield a reduction in $\Lang$. For example,
$\md{\pi\mbox{-calculus}} = 1$, since the only name matched for performing
a reduction is the name of the channel where the communication happens,
whereas $\md{\mbox{Linda}} = \md{\mbox{CPC}} = \infty$, since there is no upper bound
on the number of names that can be matched in a reduction.

\begin{thm}
\label{match}
If $\md{\Lang_1} > \! \md{\Lang_2}$, then there exists no 
valid encoding of $\Lang_1$ into $\Lang_2$.
\end{thm}

The previous proof techniques can be directly used in some cases: for example,
we shall prove that CPC cannot be encoded in any other sample calculus by exploiting
Theorem~\ref{autoriduz}. However, not all separation results can be obtained as 
corollaries of such results. The following technique provides a very useful tool when
the theorems above cannot be applied.

\begin{prop}
\label{deadlock}
Let $\encode\cdot$ be a valid encoding; then, $S \noredar\!\!_1$
implies that $\encode S \noredar\!\!_2$.
\end{prop}

The way in which we shall use this technique is the following. To prove that 
$\Lang_1$ cannot be encoded into $\Lang_2$ we reason by contradiction and assume a
valid encoding $\encode\cdot$. Then, we pick an $\Lang_1$-process $P$ that reduces;
we first show that this implies that $\encode P$ also reduces. We then analyze how
the latter reduction may have happened and, for every possible case, show a process
$P'$ obtained from $P$ (usually, by just swapping two names) such that $\encode{P'}$
reduces whereas $P'$ does not. This contradicts Proposition~\ref{deadlock}
and allows us to conclude that no valid encoding exists.

\subsection{CPC vs $\pi$-calculus and Linda}
\label{subsec:pi-linda}

A hierarchy of sets of process calculi with different communication primitives is
obtained in 
\cite{G:IC08} via combining four features: 
synchronism (synchronous vs asynchronous), 
arity (monadic vs polyadic data exchange), 
communication medium (channels vs shared dataspaces), 
and the presence of a form of pattern matching (that checks the arity of the 
tuple of names and equality of some specific names).
This hierarchy is built upon a very similar notion of encoding to that presented
in Definition~\ref{def:ve} and, in particular, it is proved that Linda \cite{Gel85} (called $\lzuzu$ in \cite{G:IC08}) is more expressive than
monadic/polyadic $\pi$-calculus \cite{milner.parrow.ea:calculus-mobile,milner:polyadic-tutorial} (called $\luzuz$ and $\luuuz$, respectively, in \cite{G:IC08}).

As Linda is more expressive than $\pi$-calculus, it is sufficient to show that CPC is more expressive than Linda.
However, apart from being a corollary of such a result, the lack of a valid encoding of CPC into $\pi$-calculus can also be shown by exploiting the matching degree, i.e.\ Theorem~\ref{match}: the matching degree of $\pi$-calculus is one, while the matching degree of CPC is infinite.

\begin{thm}
\label{noCPCinLinda-1}
There is no valid encoding of CPC into Linda.
\end{thm}
\begin{proof}
The self-matching CPC process $S=x \pre \ok$ is such that $S\not\redar$ and $S\not\suc$, however $S\bnf S\redar$ and $S\bnf S\suc$. Every Linda process $T$ such that $T\bnf T\redar$ can reduce in isolation, i.e.\ $T\redar$: this fact can be proved by induction on the structure of $T$. Conclude by Theorem~\ref{autoriduz}.
\end{proof}

The next step is to show a valid encoding of Linda into CPC.
The encoding $\encode\cdot$ is homomorphic with
respect to all operators except for input and output which are encoded as follows:
\begin{eqnarray*}
\encode{\iap{}{\wt t}.P} &\define& \pref {\patt{\wt t}} {\encode P}\\
\encode{\oap{}{\wt b}} &\define&  \pref {\patb{\wt b}} \zero
\end{eqnarray*}
The functions $\patt\cdot$ and $\patb\cdot$ are used to translate templates and data, respectively, 
into CPC patterns. The functions are defined as follows:
$$
\begin{array}{rcll}
\patt{\ } & \define & \l x \bullet \n
& \mbox{for $x$ a fresh name} 
\\
\patt{t,\wt t} & \define & t \bullet \n \bullet \patt{\wt t}
\\
\patb{\ } & \define & \n \bullet \l x
\\
\patb{b,\wt b} & \define & b \bullet \l x \bullet \patb{\wt b}
\qquad & \mbox{for $x$ a fresh name} 
\end{array}
$$
where $\n$ is any name (a symbolic name is used for clarity but no result relies upon this choice).
Moreover, the function $\patb\cdot$ associates a bound variable to 
every name in the sequence; this fact ensures that 
a pattern that translates a datum and a pattern that translates 
a template match only if they have the same length (this is a feature
of Linda's pattern matching but not of CPC's unification).
It is worth noting that the simpler translation
$
\encode{\oap{}{b_1,\ldots,b_n}} \ \define\ 
\pref {b_1 \bullet \ldots \bullet b_n} \zero
$
would not work: the Linda process $\oap{}{b}\bnf\oap{}{b}$
does not reduce, whereas its encoding would, in contradiction with Proposition~\ref{deadlock}.

Next is  to prove that this encoding is valid. 
This is a corollary of the following lemma, stating a strict 
correspondence between Linda's pattern matching and CPC's unification
(on patterns arising from the translation).

\begin{lem}
\label{lem:twomatchs}
$\pmtch(\wt t;\wt b) = \sigma$ if and only if 
\[\aaa {\patt{\wt t}}
{\patb{\wt b}} = (\sigma \cup \{\n/x\}, \{\n/x_0,\ldots,\n/x_n\}),\]
where $\{x_0,\ldots,x_n\} = {\sf bn}(\patb{\wt b})$ and
${\sf dom}(\sigma) \uplus \{x\} = {\sf bn}(\patt{\wt t})$ 
and $\sigma$ maps names to names. 
\end{lem}
\proof
In both directions the proof is by induction on the length of $\wt t$.
The forward direction is as follows.
\begin{itemize}
	\item The base case is when $\wt t$ is the empty sequence of template fields;
				thus, $\patt{\wt t} = \l x \bullet \n$.
				By definition of $\pmtch$, it must be that $\wt b$ is the empty sequence
				and that $\sigma$ is the empty substitution.
				Thus, $\patb{\wt b} = \n \bullet \l x$ and the thesis follows.
	\item For the inductive step $\wt t= t,\wt{t'}$ and $\patt{\wt t} = t \bullet \n \bullet \patt{\wt{t'}}$.
				By definition of $\pmtch$, it must be that $\wt b = b, \wt{b'}$ and $\pmtch(t,b) = \sigma_1$ and
				$\pmtch(\wt{t'},\wt{b'}) = \sigma_2$ and $\sigma = \sigma_1 \uplus \sigma_2$.
				By the induction hypothesis, $\aaa {\patt{\wt {t'}}} {\patb{\wt {b'}}} = 
				(\sigma_2 \cup \{\n/x\} ; \{\n/x_1,\ldots,\n/x_n\})$,
				where $\{x_1,\ldots,x_n\} = {\sf bn}(\patb{\wt {b'}})$
				and ${\sf dom}(\sigma_2) \uplus \{x\} = {\sf bn}(\patt{\wt {t'}})$.
				There are now two sub-cases to consider according to the kind of template field $t$.
				\begin{itemize}
					\item If $t = \pro b$ then $\sigma_1 = \{\}$; therefore,
								$\sigma
                                          = \sigma_2$ as well as
								$\aaa {\patt{\wt t}} {\patb{\wt b}} = 
								(\sigma\cup \{\n/x\}, \{\n/x_0,\ldots,\n/x_n\})$.
					\item If $t = \l y$ then $\sigma_1 = \{b/y\}$ and $y \not\in {\sf dom}(\sigma_2)$.
								Thus, $\patt{\wt t}$ is a pattern in CPC and it follows that
								$\aaa {\patt{\wt t}} {\patb{\wt b}} =
								(\sigma_1\cup\sigma_2\cup \{\n/x\}, \{\n/x_0,\ldots,\n/x_n\}) =
								(\sigma \cup \{\n/x\}, \{\n/x_0,\ldots,\n/x_n\})$.
				\end{itemize}
\end{itemize}
The reverse direction is as follows.
\begin{itemize}
	\item The base case is when $\wt t$ is the empty sequence of template fields;
				thus, $\patt{\wt t} = \l x \bullet \n$.
				Now proceed by contradiction.
				Assume that $\wt b$ is not the empty sequence.
				In this case, $\patb{\wt b} = b_0 \bullet\l x_0\bullet(b_1 \bullet\l x_1\bullet(
				\ldots(b_n \bullet\l x_n \bullet (\n \bullet \l x_{n+1}))\ldots)$,
				for some $n > 0$.
				By definition of pattern unification in CPC, $\patb{\wt b}$ and $\patt{\wt t}$ cannot unify,
				and this would contradict the hypothesis. Thus, it must be that $\wt b$ is the empty sequence
				and we conclude.
	\item The inductive case is when $\wt t = t, \wt{t'}$ and thus,
				$\patt{\wt t} = t \bullet \n \bullet \patt{\wt{t'}}$.
				If $\wt b$ was the empty sequence, then
				$\patb{\wt b} = \n \bullet \l x$ and it would not unify with $\patt{\wt t}$.
 				Hence, $\wt b = b, \wt{b'}$ and so
 				$\patb{\wt b} = b \bullet \l x \bullet \patb{\wt{b'}}$.
 				By definition of pattern-unification in CPC it follows that
 				$\aaa t b = (\sigma_1, \{\})$ and
				$\aaa {\patt{\wt {t'}}} {\patb{\wt {b'}}} =(\sigma_2 \cup \{\n/x\},
				\{\n/x_1,\ldots, \n/x_n\})$	and $\sigma=\sigma_1\cup\sigma_2$.
				Now consider the two sub-cases according to the kind of the template field $t$.
				\begin{itemize}
					\item If $t = \ex b$ then $\sigma_1 = \{\}$ and so $\sigma_2 = \sigma$.
								By induction hypothesis, $\pmtch(\wt{t'};\wt{b'}) = \sigma$,
								and so $\pmtch(\wt{t};\wt{b}) = \sigma$.
					\item If $t = \l y$ then $\sigma_1 = \{b/y\}$ and
								$\sigma_2=\{n_i/y_i\}$ for $y_i\in{\sf dom}(\sigma)\backslash \{y\}$
								and $n_i=\sigma y_i$.
								Thus, $y \not\in {\sf dom}(\sigma_2)$ and so $\sigma = \sigma_1 \uplus \sigma_2$.
								By the induction hypothesis, $\pmtch(\wt{t'};\wt{b'}) = \sigma_2$;
								moreover, $\pmtch(t;b) = \sigma_1$.
								Thus, $\pmtch(\wt{t};\wt{b}) = \sigma$.\qed
				\end{itemize}
\end{itemize}

\begin{lem}
\label{lem:structenc-linda}
If $P \equiv Q$ then $\encode P \equiv \encode Q$.
Conversely, if $\encode P \equiv Q$ then $Q = \encode{P'}$, for some $P' \equiv P$.
\end{lem}
\begin{proof}
Straightforward, from the fact that $\equiv$ acts only on operators that $\encode\cdot$ translates homomorphically.
\end{proof}

\begin{thm}
\label{thm:opcorr}
The translation $\encode\cdot$ from Linda into CPC preserves and reflects reductions.
That is:
\begin{itemize}
	\item If $P\redar P'$ then $\xtrans P \redar \xtrans {P'}$;
	\item if $\xtrans P \redar Q$ then $Q=\xtrans {P'}$ for some $P'$ such that $P\redar P'$.
\end{itemize}
\end{thm}
\begin{proof}
Both parts can be proved by a straightforward induction
on judgements $P \redar P'$ and $\encode P \redar Q$, respectively.
In both cases, the base step is the most interesting one and
follows from Lemma~\ref{lem:twomatchs}; the inductive
cases where the last rule used is the structural one rely on
Lemma~\ref{lem:structenc-linda}.
\end{proof}

\begin{cor}
\label{cor:valid-linda}
The encoding of Linda into CPC is valid.
\end{cor}
\begin{proof}
Compositionality and name invariance hold by construction. 
Operational correspondence and divergence reflection
follow from Theorem~\ref{thm:opcorr}. Success sensitiveness
can be proved as follows: $P \suc$ means that there exist $P'$
and $k \geq 0$ such that $P \redar^k P' \equiv P''\bnf \ok$; 
by exploiting Theorem~\ref{thm:opcorr} $k$ times 
and Lemma~\ref{lem:structenc-linda}, we obtain that $\encode P \redar^k
\encode{P'} \equiv \encode{P''}\bnf \ok$, i.e.\ that $\encode P \suc$.
The converse implication can be proved similarly. 
\end{proof}

\subsection{CPC vs Spi}
\label{subsec:spi-new}
CPC cannot be encoded into Spi calculus, as a corollary of Theorem~\ref{autoriduz}.
This can be proved as in Theorem~\ref{noCPCinLinda-1}: the self-unifyinging CPC
process $x \pre \ok$ cannot be properly rendered in Spi.

The remainder of this section develops an encoding of Spi calculus into CPC.
The terms can be encoded as patterns using the reserved
names {\sf pair}, {\sf encr}, $0$, {\sf suc}, and the natural numbers $>0$ ranged over by $i$ with
\[
\begin{array}{rclrcl}
\encode n &\define& n
\qquad\qquad
&
\encode {{\it suc}(M)} &\define& {\sf suc} \bullet \encode M
\\
\encode x &\define& x
&
\encode {(M,N)} &\define& {\sf pair}\bullet \encode M \bullet \encode N
\\
\encode 0 &\define& 0
&
\encode {\encr M N} &\define& {\sf encr}\bullet \encode M \bullet \encode N
\\
\encode i &\define& {\sf suc}^i 0
\end{array}
\]
where ${\sf suc}^i 0$ denotes {\sf suc} compounded $i$ times with $0$.
The tagging is used for safety, as otherwise there are potential
pathologies in the translation: for example,
without tags, the representation of an encrypted term could be confused
with a pair.

The encoding of the familiar process forms are homomorphic as expected. The input and output
both encode as cases:
\begin{eqnarray*}
\begin{array}{rcll}
\encode {M(x).P} &\define& \encode M \bullet\l x\bullet \n \pre \encode P\\
\encode{\overline{M}\langle N\rangle .P} &\define& \encode M\bullet(\encode N)\bullet\l x\pre\encode P
		&\quad \mbox{$x$ is a fresh name}
\end{array}
\end{eqnarray*}
The symbolic name $\n$ (input) and fresh name $x$ (output) are used to
ensure that encoded inputs will only unify with encoded outputs as for Linda.

The four remaining process forms all require pattern unification and so
translate to cases in parallel. In each encoding a fresh name $n$ is
used to prevent interaction with other processes, see
Proposition~\ref{prop:free_n_match_proc}. As in the Spi calculus,
the encodings will reduce only after a successful unification and will be
stuck otherwise. The encodings are
\begin{eqnarray*}
\encode {[M\ {\it is}\ N]P} &\define&
            \res n (\pro n\bullet \encode M\pre \encode P \bnf \pro
n\bullet \encode N)\\
\encode {{\it let}\ (x,y)= M\ {\it in}\ P} &\define&
            \res n (\pro n \bullet (\pro {{\sf pair}}\bullet\l x\bullet
\l y)\pre \encode P\\
    & & \quad\quad    |    \pro n\bullet \encode M)\\
\encode {{\it case}\ M\ {\it of}\ \encr x N:P} &\define&
            \res n (\pro n \bullet (\pro {{\sf encr}}\bullet\l x\bullet
\encode N)\pre \encode P\\
    & & \quad\quad    |    \pro n\bullet \encode M)\\
\encode {{\it case}\ M\ {\it of}\ 0:P\ {\it suc}(x):Q} &\define&
            \res n (\pro n\bullet \pro 0\pre \encode P \\
    & & \quad\quad    |    \pro n\bullet (\pro {{\sf suc}}\bullet \l
x)\pre \encode Q \\
    & & \quad\quad    |    \pro n\bullet \encode M )
\end{eqnarray*}
The unification $[M\ {\it is}\ N]P$ only reduces to $P$ if $M=N$, thus the
encoding creates two patterns using $\encode M$ and $\encode
N$ with one reducing to $\encode P$.  
The encoding of pair splitting ${\it let}\ (x,y)= M\ {\it in}\ P$ 
creates a case with a pattern that unifies with a tagged pair and binds the
components to $x$ and $y$ in $\encode P$. This is put in parallel with
another case that has $\encode M$ in the pattern. 
The encoding of a decryption ${\it case}\ M\ {\it of}\ \encr x N:P$ checks
whether $\encode M$ is encoded with key $\encode N$ and
retrieves the value encrypted by binding it to $x$ in the continuation.
Lastly the encoding of an integer test ${\it case}\ M\ {\it of}\ 0:P\ {\it
suc}(x):Q$ creates a case for each of the zero and the
successor possibilities. These cases unify the tag and the reserved
names $0$, reducing to $\encode P$, or ${\sf suc}$ and binding $x$ in
$\encode Q$. The term to be compared $\encode M$ is as in the other cases.

Let us now prove validity of this encoding.

\begin{lem}
\label{lem:structenc-spi}
If $P \equiv Q$ then $\encode P \equiv \encode Q$.
Conversely, if $\encode P \equiv Q$ then $Q = \encode{P'}$, for some $P' \equiv P$.
\end{lem}
\begin{proof}
Straightforward, from the fact that $\equiv$ acts only on operators that $\encode\cdot$ translates homomorphically.
\end{proof}

\begin{thm}
\label{spi2cpc-red}
The translation $\encode\cdot$ from Spi calculus into CPC preserves and reflects reductions, up-to CPC's barbed
congruence. That is:
\begin{itemize}
	\item If $P\redar P'$ then $\xtrans P \redar \beq \xtrans {P'}$;
	\item if $\xtrans P \redar Q$ then $Q \beq \xtrans {P'}$ for some $P'$ such that $P\redar P'$.
\end{itemize}
\end{thm}
\proof
The first claim can be proved by a straightforward induction
on judgement $P \redar P'$. The base case is proved by reasoning on the Spi axiom used to infer the reduction.
Although all the cases are straightforward, a reduction rule for integers is shown for illustration.
Consider the reduction for a successor as the reduction for zero is simpler.
In this case, $P = {\it case}\ {\it suc}(M)\ {\it of}\ 0:P_1\ {\it suc}(x):P_2$
and $P' = \{M/x\}P_2$. Then, 
\begin{eqnarray*}
\encode P &\define& 
			\res n (\pro n\bullet \pro 0\pre \encode {P_1} \\
	& & \quad\quad	\bnf	\pro n\bullet (\pro {{\sf suc}}\bullet \l x)\pre \encode {P_2}\\
	& & \quad\quad	\bnf	\pro n\bullet ({\sf suc} \bullet \encode M) \pre \zero) \; .
\end{eqnarray*}
and it can only reduce to
$$
\{\encode M / x\}\encode{P_2} \bnf \res n \pro n\bullet \pro 0\pre \encode {P_1}
$$
A straightforward induction on the structure of $P_2$ proves
$\{\encode M / x\}\encode{P_2} = \encode{\{M / x\}P_2}$. Thus,
$\encode P \redar \encode{\{M / x\}P_2}\bnf
\res n \pro n\bullet \pro 0\pre \encode {P_1} \beq \encode{P'}$,
where the last equivalence follows from Proposition~\ref{prop:free_n_match_proc}.
The inductive case is straightforward, with the structural case relying on Lemma~\ref{lem:structenc-spi}.

The second part can be proved by induction on judgement $\encode P \redar Q$.
There is just one base case, i.e.\ when $\encode P = p \pre Q_1\bnf q \pre Q_2$ and
$Q = \sigma Q_1\bnf \rho Q_2$ and $\{p\pmatch q\} = (\sigma, \rho)$.
By definition of the encoding, it can only be that
$p= \encode M\bullet\l x\bullet \n$ and $Q_1 = \encode {P_1}$ and
$q = \encode M\bullet (\encode N)\bullet \l x$ and $Q_2 = \encode {P_2}$
for some $P_1$, $P_2$, $M$ and $N$. This means that $P = \iap M x.P_1\bnf \oap M N . P_2$
and that $Q = \{\encode N/x\}\encode{P_1}\bnf\encode{P_2} = \encode{\{N/x\}P_1\bnf P_2}$.
To conclude, it suffices to take $P' = \{N/x\}P_1\bnf P_2$.
For the inductive case there are two possibilities.
\begin{itemize}
	\item The inference of $\encode P \redar Q$ ends with an application of the rule for
				parallel composition or for structural congruence: this case can be proved by a
				straightforward induction.
	\item The inference of $\encode P \redar Q$ ends with an application of the rule for restriction;
				thus, $\encode P = \res n Q'$, with $Q' \redar Q''$ and $Q = \res n Q''$.
				If $Q' = \encode {P''}$, for some $P''$, apply a straightforward induction.
				Otherwise, there are the following four possibilities.
				\begin{itemize}
					\item $Q' = \pro n\bullet\pro{\encode M}\pre\encode{P_1}\bnf\pro n\bullet\pro{\encode N}$
								and, hence, $Q'' = \encode{P_1}$. By definition of the encoding, 
								$P = [M\ {\it is}\ N] P_1$. Notice that the reduction $Q' \redar Q''$ can happen only if
								$\encode M$ and $\encode N$ unify; by construction of the encoding of Spi-terms,
								this can happen only if $M = N$ and, hence, $P \redar P_1$.
								The thesis follows by letting $P' = P_1$, since $n$ is a fresh name and so
								$Q = \res n \encode {P_1} \equiv \encode {P_1}$.
					\item $Q' =\pro n\bullet(\pro{{\sf pair}}\bullet(\l x\bullet\l y))\pre\encode{P_1}
								\bnf\pro n\bullet({\sf pair}\bullet(\encode M\bullet\encode N))$
								and, hence, $Q'' = \{\encode M/x,\encode N/y\}\encode{P_1}$. 
								This case is similar to the previous one, by letting $P$ be
								${\it let}\ (x,y)= (M,N)\ {\it in}\ P_1$.
					\item $Q' = \pro n\bullet (\pro {{\sf encr}}\bullet(\l x\bullet \encode N))\pre \encode {P_1}
								\bnf \pro n\bullet ({\sf encr}\bullet (\encode M \bullet \encode N))$
								and, hence, $Q'' = \{\encode M/x\}\encode{P_1}$.
								This case is similar to the previous one, by letting $P$ be
								${\it case}\ \encr M N\ {\it of}\ \encr x N:P_1$.
					\item $Q' = \pro n\bullet \pro 0\pre \encode {P_1}
								\bnf\pro n\bullet (\pro {{\sf suc}}\bullet \l x)\pre \encode {P_2}
								\bnf\pro n\bullet \encode M$.
								Hence, $P = {\it case}\ M\ {\it of}\ 0:P_1\ {\it suc}(x):P_2$.
								According to the kind of $\encode M$, there are two sub-cases
								(notice that, since $Q' \redar Q''$, no other possibility is allowed for $\encode M$).
								\begin{itemize}
									\item $\encode M = 0$: in this case, $Q'' = \encode {P_1}
												\bnf \pro n\bullet (\pro {{\sf suc}}\bullet \l x)\pre\encode {P_2}$
												and so $Q = \res n Q'' \equiv \encode {P_1}
												\bnf\res n\pro n\bullet(\pro {{\sf suc}}\bullet\l x)\pre\encode {P_2}
												\beq \encode {P_1}$.
												In this case, $M = 0$ and so $P \redar P_1$; to conclude,
												it suffices to let $P'$ be $P_1$.
									\item $\encode M = {\sf suc} \bullet \encode {M'}$, for some $M'$:
												in this case, $Q'' = \{\encode{M'}/x\}\encode {P_2} \bnf
												\pro n\bullet \pro 0\pre \encode {P_1}$
												and so $Q = \res n Q'' \equiv \encode {\{M'/x\}P_2}\bnf\res n
												\pro n\bullet 0\pre \encode {P_1}
												\beq \encode {\{M'/x\}P_2}$.
												In this case, $M = {\it suc}(M')$ and so $P \redar \{M'/x\}P_2$; to conclude, 
												it suffices to let $P'$ be $\{M'/x\}P_2$.\qed
								\end{itemize}
				\end{itemize}
\end{itemize}

\begin{cor}
\label{cor:valid-spi}
The encoding of Spi calculus into CPC is valid.
\end{cor}
\begin{proof}
See the proof for Corollary~\ref{cor:valid-linda}. 
\end{proof}

Notice that the criteria for a valid encoding do not imply full
abstraction of the encoding (actually, they were defined as an alternative to full abstraction
\cite{G:IC08,G:CONCUR08}). This means that the encoding of equivalent
Spi calculus processes can be distinguished by contexts in CPC that do not result from the
encoding of any Spi calculus context. Indeed, while this encoding allows Spi calculus to be modelled in
CPC, it does {\em not} entail that cryptography can be properly rendered.
Consider the pattern ${\sf encr}\bullet \l x \bullet \l y$ that could unify with the encoding of an encrypted term
to bind the message and key, so that CPC can break any encryption!
Indeed this is an artefact of the straightforward approach to encoding taken here.
Some discussion of alternative approaches to encryption in CPC are detailed in 
\cite{GivenWilsonPHD}.

\subsection{CPC vs Fusion}

The separation results for CPC and the other process calculi presented so far have all been proved via symmetry;
thus, the relationship between Fusion and CPC is of particular interest.
Such calculi are {\em unrelated}, in the sense that there exists no valid
encoding from one into the other. The impossibility for a valid
encoding of CPC into Fusion can be proved in two ways, by exploiting either the matching degree or the symmetry of CPC.

\begin{thm}
\label{thm:nocpc2fusion-1}
There is no valid encoding of CPC into Fusion.
\end{thm}
\begin{proof}
The matching degree of Fusion is 1 while the matching degree of CPC is infinite;
conclude by Theorem~\ref{match}.
Alternatively, reuse the proof for Theorem~\ref{noCPCinLinda-1} as every Fusion process 
$T$ is such that $T\bnf T\redar$ implies $T\redar$.
\end{proof}

The converse separation result is ensured by the following theorem.

\begin{thm}
\label{thm:fusionNoinCPC}
There exists no valid encoding of Fusion into CPC.
\end{thm}
\proof
By contradiction, assume that there exists a valid encoding $\encode\cdot$
of Fusion into CPC.
Consider the Fusion process $P \define \res x(\oap u x\bnf \iap u y.\ok)$,
for $x$, $y$ and $u$ pairwise distinct. By success sensitiveness, $P \suc$
entails that $\encode P \suc$. 

We first show that $\encode P$ must reduce before reporting success,
i.e.\ every occurrence of $\ok$ in $\encode P$ falls underneath
some prefix.  By compositionality, $\encode P \define \CopN C {\res x}
{\{u,x,y\}} {\CopN C | {\{u,x,y\}} {\encode{\oap u x};\\ \encode{\iap u
      y.\ok}}}$.  If $\encode P$ had a top-level unguarded occurrence
of $\ok$, then such an occurrence could be in $\CopN C {\res x}
{\{u,x,y\}} {\_\,}$, in $\CopN C | {\{u,x,y\}} {\_\,_1;\_\,_2}$, in
$\encode{\oap u x}$ or in $\encode{\iap u y.\ok}$; in any case, it
would also follow that at least one of:
%
\begin{equation}
\label{eq:fus:ex1}
\encode{\res x(\oap u x\bnf \iap y u.\ok)}
\end{equation}
that has $\iap u y$ replaced with $\iap y u$;
or
\begin{equation}
\label{eq:fus:ex2}
\encode{\res x(\oap x u\bnf \iap u y.\ok)}
\end{equation}
that has $\oap u x$ replaced with $\oap x u$,
would report success, 
whereas both
Equation~\ref{eq:fus:ex1}$\not\suc$
and
Equation~\ref{eq:fus:ex2}$\not\suc$,
against success sensitiveness of $\encode\cdot$.
Thus, the only possibility for $\encode P$ to report success is
to perform some reduction steps (at least one) and then exhibit
a top-level unguarded occurrence of $\ok$.

We now prove that every possible reduction leads to contradiction of the validity of $\encode\cdot$;
this suffices to conclude. There are five possibilities for $\encode P \redar$.
\begin{enumerate}
	\item Either $\CopNnoarg C {\res x} {\{u,x,y\}} \redar$, or
				$\CopNnoarg C | {\{u,x,y\}} \redar$, or 
				$\encode{\oap u x}\redar$ or $\encode{\iap u y.\ok}\redar$. In any
				of these cases, at least one out of 
				$\encode{\mbox{Equation~\ref{eq:fus:ex1}}}$ or
				$\encode{\mbox{Equation~\ref{eq:fus:ex2}}}$ would reduce;
				however, $\mbox{Equation~\ref{eq:fus:ex1}} \not\redar$ and
				$\mbox{Equation~\ref{eq:fus:ex2}} \not\redar$, 
				against Proposition~\ref{deadlock} (that must hold whenever $\encode\cdot$
				is valid).

	\item Reduction is generated by interaction between
				$\CopNnoarg C {\res x} {\{u,x,y\}}$ and
				$\CopNnoarg C | {\{u,x,y\}}$.  Then, as before, 
				$\encode{\mbox{Equation~\ref{eq:fus:ex1}}} \redar$
				whereas $\mbox{Equation~\ref{eq:fus:ex1}} \not\redar$, 
				against Proposition~\ref{deadlock}.

	\item Reduction is generated by interaction between
				$\CopNnoarg C \op {\{u,x,y\}}$ and $\encode{\oap u x}$, for
				$\op \in \{\res x,\,|\,\}$. Like case 2.

	\item Reduction is generated by interaction between
				$\CopNnoarg C \op {\{u,x,y\}}$ and 
				$\encode{\iap u y.\ok}$, for
				$\op \in \{\res x,\,|\,\}$. As before it follows that
				$\encode{\mbox{Equation~\ref{eq:fus:ex2}}} \redar$
				whereas $\mbox{Equation~\ref{eq:fus:ex2}} \not\redar$, 
				against Proposition~\ref{deadlock}.

	\item The reduction is generated by an interaction between the processes
				$\encode{\oap u x}$ and 
				$\encode{\iap u y.\ok}$. In this
				case, it follows that $\encode{\oap u x\bnf \iap u y.\ok}\redar$
				whereas $\oap u x\bnf \iap u y.\ok\not\redar$: indeed, the interaction
				rule of Fusion imposes that at least one between $x$ and $y$ must
				be restricted to yield the interaction.\qed
\end{enumerate}

\subsection{CPC vs Psi}

CPC and Psi are {\em unrelated}, in the sense that there exists no valid
encoding from one into the other. As in Theorem~\ref{noCPCinLinda-1}, 
the impossibility for a valid
encoding of CPC into Psi can be proved by exploiting 
the symmetry of CPC.
The converse separation result is ensured by the following theorem.

\begin{thm}
\label{thm:psiNoinCPC}
There exists no valid encoding of Psi into CPC.
\end{thm}
\proof
Assume that there exists a valid encoding $\encode\cdot$
of Psi into CPC. Consider the Psi process 
$P \define (\bar a.c \ |\ b.(\ok\ |\ c))\ | \ \assert{a \cheq b}$,
where we have omitted the argument of the actions to simplify the proofs,
and chosen $a$, $b$ and $c$ pairwise distinct; also consider the reduction
$$
\prooftree
	\prooftree \{a \cheq b\} \vdash a \cheq b
	\justifies \{a \cheq b\} \triangleright \bar a.c \ |\ b.(\ok\ |\ c) \redar c \ |\ \ok\ |\ c
	\endprooftree
\justifies \one \triangleright P \redar (c\ |\ \ok\ |\ c)\ | \ \assert{a \cheq b}
\endprooftree
$$
Therefore, $P \suc$ and, by success sensitiveness, $\encode P \suc$. 
Hence, by compositionality, $\encode P \define \CopN C {|} {\{a,b,c\}} 
{\CopN C | {\{a,b,c\}} {\encode{\bar a.c}; \encode{b.(\ok\ |\ c)}}; \encode{\assert{a \cheq b}}}$.
Like in the proof of Theorem~\ref{thm:fusionNoinCPC}, it can be proven that the only possibility 
for $\encode P$ to report success is to perform some reduction steps (at least one) and then exhibit
a top-level unguarded occurrence of $\ok$.

We now prove that every possible reduction leads to contradiction of the validity of $\encode\cdot$;
this suffices to conclude. Of course, none of $\encode{\bar a.c}$, $\encode{b.(\ok\ |\ c)}$
and $\encode{\assert{a \cheq b}}$ can reduce, because $\bar a.c$, $b.(\ok\ |\ c)$ and $\assert{a \cheq b}$
do not reduce. Thus, there are seven possibilities for $\encode P \redar$.
\begin{enumerate}
	\item Either $\CopNnoarg C {|} {\{a,b,c\}} \redar$ or the reduction is obtained by synchronizing
	the two copies of $\CopNnoarg C {|} {\{a,b,c\}}$. In both cases, $\encode{(\bar c.a \ |\ b.(\ok\ |\ c))\ | \ \assert{a \cheq b}}$
    (with $\bar a.c$ replaced by $\bar c.a$)
	would also reduce, whereas $(\bar c.a \ |\ b.(\ok\ |\ c))\ | \ \assert{a \cheq b} \not\redar$, 
	against Proposition~\ref{deadlock} (that must hold whenever $\encode\cdot$ is valid).

	\item The reduction is obtained by synchronizing $\encode{\bar a.c}$
	with (one of the two copies of) $\CopNnoarg C {|} {\{a,b,c\}}$. In this case, 
	also $\encode{(\bar a.c \ |\ c.(\ok\ |\ b))\ | \ \assert{a \cheq b}}$
    (with $b.(\ok\ |\ c)$ replaced by $c.(\ok\ |\ b)$)
	would reduce, whereas $(\bar a.c \ |\ c.(\ok\ |\ b))\ | \ \assert{a \cheq b} \not\redar$.

	\item The reduction is obtained by synchronizing $\encode{b.(\ok\ |\ c)}$
	with (one of the two copies of) $\CopNnoarg C {|} {\{a,b,c\}}$. This case is
	proved impossible like case 1 above.

	\item The reduction is obtained by synchronizing $\encode{\assert{a \cheq b}}$
	with (one of the two copies of) $\CopNnoarg C {|} {\{a,b,c\}}$. This case is
	proved impossible like cases 1 and 2 above.

	\item The reduction is obtained by synchronizing $\encode{\bar a.c}$
	with $\encode{b.(\ok\ |\ c)}$. In this case, 
	also $\encode{\bar a.c \ |\ b.(\ok\ |\ c)}$
	would reduce, whereas $\bar a.c \ |\ b.(\ok\ |\ c) \not\redar$.

	\item The reduction is obtained by synchronizing $\encode{\bar a.c}$
	with $\encode{\assert{a \cheq b}}$. This case is
	proved impossible like case 2 above.

	\item The reduction is obtained by synchronizing $\encode{b.(\ok\ |\ c)}$
	with $\encode{\assert{a \cheq b}}$. This case is
	proved impossible like case 1 above.\qed
\end{enumerate}

\section{Conclusions} 
\label{sec:conclusions}

Concurrent pattern calculus uses patterns to represent input,
output and tests for equality, whose interaction is driven by
unification that allows a two-way flow of information.  This symmetric
information exchange provides a concise model of trade in the
information age.  This is illustrated by the example of traders who
can discover each other in the open and then close the deal
in private.

As patterns drive interaction in CPC, 
their properties heavily influence CPC's behaviour theory. As pattern
unification may match any number of names these must all be accounted
for in the definition of barbs. More delicately, some patterns are
compatible with others, in that their unifications yield similar results.
The resulting  bisimulation requires that the transitions be compatible
patterns rather than exact.
Further, the pattern-matching bisimulation developed for CPC can easily account
for other kinds of pattern-matching, such as in polyadic $\pi$-calculus and
Linda \cite{GivenWilsonGorla13}.

CPC supports valid encodings of many popular concurrent calculi such
as $\pi$-calculus, Spi calculus and Linda as its patterns describe
more structures. However, these three calculi do not support valid
encodings of CPC because, among other things, they are insufficiently
symmetric. On the other hand, while fusion calculus is completely
symmetric, it has an incompatible approach to interaction.
Similarly, Psi calculus is unrelated to CPC due to supporting implicit
computations, while also being less symmetric.

Another path of development for a process calculus is implementation in a
programming language \cite{Pierce97pict:a,Klava,cpplinda,20110201:jocaml}.
The \bondi\ programming language is based upon pattern matching as the core
of reduction and the theory of pattern calculus \cite{pcb,bondi}. A
\cbondi\ has also been developed that extends \bondi\ with concurrency
and interaction based on the pattern unification and theory of CPC
\cite{GivenWilsonPHD,cbondi}.

\bigskip
\noindent{\bf Acknowledgments } We would like to thank the anonymous reviewers for their
fruitful comments and for their constructive attitude towards our paper.

\section*{Appendix A: Proofs of Section~\ref{sec:LTS}}

\paragraph{Proof of Proposition~\ref{prop:imfin}}
First of all, let us define an alternative (but equivalent, up-to $\equiv$) LTS
for CPC, written $\!\!\llts\mu$: it is obtained by replacing {\sf rep} with the 
following two rules (all the other rules are the same, with $\llts{}$ in place
of $\ltsred{}$ everywhere):
$$
\prooftree P \llts\mu P'
\justifies !P \llts\mu P'\bnf !P
\endprooftree
\qquad\qquad
\prooftree P \llts{\rest m p} P' \quad P \llts{\rest n q} P''
\justifies !P \llts\tau \res{\withsetnot{\wt m,\wt n}{\wt m\cup\wt n}}(\sigma P'\bnf \rho P'')\bnf !P
\endprooftree\ 
\begin{array}{l}
\{p \pmatch q\} = (\sigma, \rho)\\
\wt m \cap \wt n = \emptyset
\end{array}
$$
We can prove that: (1) if $P \llts\mu P'$ then $P \ltsred\mu P'$; and (2)
if $P \ltsred\mu P'$ then $P \llts\mu P''$, for some $P'' \equiv P'$
(both proofs are done by a straightforward induction on the derivation of
the premise, whose only interesting case is when $P =~!Q$, for some $Q$).

Now define the following measure associated to a process:
$$
\begin{array}{lll}
\meas \zero\ =\ 0
\qquad\qquad
&
\meas {p \pre P}\ =\ 1
&
\meas {\res n P}\ =\ \meas P
\vspace*{.2cm}
\\
\multicolumn{2}{l}{
\meas {P_1 \bnf P_2}\ =\ \meas {P_1}+\meas {P_2}+\meas {P_1}\cdot\meas {P_2}
}
\quad
&
\meas {\,!P}\ =\ \meas P+\meas P\cdot\meas P
\end{array}
$$
By induction on the structure of $P$, we can prove that $|\{P' : P \llts\mu P'\}| \leq \meas P$.
By exploiting this fact and (2) above, it follows that there are finitely many (up-to $\equiv$)
$P'$ such that $P \ltsred\mu P'$.
\qed

\paragraph{Proof of Lemma~\ref{lem:lts-exhibit-p}}
The proof is by induction on the inference for $P\ltsred{\rest m p} P'$.
The base case is when the last rule is {\sf case}, with $P = (p\to P_1)\ltsred p P_1 = P'$;
conclude by taking $\wt n = \emptyset$ and $Q_1 = P_1$ and $Q_2 = \zero$.
For the inductive step, consider the last rule in the inference.
\begin{itemize}
	\item If the last rule is {\sf resnon} then $P = \res o P_1 \ltsred{\rest m p}\res o P_1' = P'$, where
				$P_1 \ltsred{\rest m p} P_1'$ and $o\notin{\sf names}(\rest m p)$.
				By induction, there exist $\wt n'$ and $Q_1'$ and $Q_2'$ such that 
				$P_1 \equiv (\nu \wt m)(\nu \wt n')(p\to Q_1'\bnf Q_2')$ and
				$P_1' \equiv (\nu \wt n')(Q_1'\bnf Q_2')$ and 
				$\wt n'\cap{\sf names}(\rest m p)=\emptyset$ and ${\sf bn}(p)\cap{\sf fn}(Q_2')=\emptyset$.
				As $o\notin{\sf names}(\rest m p)$ and by $\alpha$-conversion $o\notin\wt n'$,
				conclude with $Q_1 = Q_1'$ and $Q_2 = Q_2'$ and $\wt n = \wt n',o$.
	\item If the last rule is {\sf open} then $P = \res o P_1\ltsred{\res {\wt m' ,o} p} P_1' = P'$, where
				$P_1\ltsred{\res {\wt m'} p}P_1'$ and 
				$o\in {\sf vn}(p)\backslash(\wt m'\cup{\sf pn}(p)\cup{\sf bn}(p))$
				and $\wt m=\wt m',o$.
				By induction, there exist $\wt n'$ and $Q_1'$ and $Q_2'$ such that 
				$P_1 \equiv (\nu \wt m')(\nu \wt n')(p\to Q_1'\bnf Q_2')$ and
				$P_1' \equiv (\nu \wt n')(Q_1'\bnf Q_2')$ and
				$\wt n'\cap{\sf names}(\res {\wt m'} p)=\emptyset$ and
				${\sf bn}(p)\cap{\sf fn}(Q_2')=\emptyset$.
				Conclude with $\wt n = \wt n'$ and $Q_1 = Q_1'$ and $Q_2 = Q_2'$.
	\item If the last rule is {\sf parext} then $P = P_1\bnf P_2 \ltsred{\rest m p} P_1' \bnf P_2$, 
				where $P_1\ltsred{\rest m p} P_1'$
				and ${\sf fn}(P_2)\cap(\wt m\cup{\sf bn}(p))=\emptyset$.
				By induction, there exist $\wt n'$ and $Q_1'$ and $Q_2'$ such that 
				$P_1 \equiv (\nu \wt m)(\nu \wt n')(p\to Q_1'\bnf Q_2')$ and 
				$P_1' \equiv (\nu \wt n')(Q_1'\bnf Q_2')$ and
				$\wt n'\cap{\sf names}(\rest m p)=\emptyset$ and ${\sf bn}(p)\cap{\sf fn}(Q_2')=\emptyset$.
				As ${\sf bn}(p)\cap{\sf fn}(P_2) = \emptyset$, we can
				conclude with $\wt n = \wt n'$ and $Q_1 = Q_1'$ and $Q_2 = Q_2'\bnf P_2$.
	\item If the last rule is {\sf rep} then $P = \,\,!Q \ltsred{\rest m p} P'$, 
				where $Q\bnf !Q \ltsred{\rest m p} P'$. We conclude by induction
				and by the fact that $P \equiv Q\bnf!Q$.
\qed
\end{itemize}

\paragraph{Proof of Proposition~\ref{prop:tau-red}}
The first claim is proved by induction on the inference for $P\ltsred\tau P'$.
The base case is with rule {\sf unify}: $P = P_1\bnf Q_1$, where
$P_1 \ltsred{\rest m p}P'_1$ and $Q_1 \ltsred{\rest n q} Q'_1$ and 
$P' = \res{\withsetnot{\wt m, \wt n}{\wt m\cup \wt n}}(\sigma P'_1\bnf\rho Q'_1)$ and
$\{p\pmatch q\}=(\sigma,\rho)$ and 
$\wt m \cap {\sf fn}(Q_1) = \wt n \cap {\sf fn}(P_1) = \emptyset$
and $\wt m\cap\wt n=\emptyset$. By Lemma~\ref{lem:lts-exhibit-p}, it follows that 
$P_1 \equiv \rest m\rest o(p\to P''_1\bnf P''_2)$ and
$P_1' \equiv \rest o(P''_1\bnf P''_2)$, with $\wt o\cap{\sf names}(\rest m p)=\emptyset$
and ${\sf bn}(p)\cap{\sf fn}(P''_2)=\emptyset$; similarly, 
$Q_1 \equiv \rest n\rest r(q\to Q''_1\bnf Q''_2)$ and $Q_1' \equiv \rest r(Q''_1\bnf Q''_2)$, 
with $\wt r\cap{\sf names}(\rest n q)=\emptyset$
and ${\sf bn}(q)\cap{\sf fn}(Q''_2)=\emptyset$.
By exploiting $\alpha$-conversion on the names in $\wt o$ and $\wt r$, we have 
$\withsetnot{\wt o,\wt r}{(\wt o\cup \wt r)}\cap({\sf names}(\rest m p)\cup{\sf names}(\rest n q))=\emptyset$;
thus, $P_1\bnf Q_1 \equiv \res{\withsetnot{\wt m, \wt n}{\wt m\cup \wt n}}
\res{\withsetnot{\wt o, \wt r}{\wt o\cup \wt r}} (p\to P''_1\bnf P''_2\bnf q\to Q''_1\bnf Q''_2)
\redar \res{\withsetnot{\wt m, \wt n}{\wt m\cup\wt n}}
\res{\withsetnot{\wt o, \wt r}{\wt o\cup \wt r}} (\sigma P''_1\bnf P''_2\bnf \rho Q''_1\bnf Q''_2)$.
Since $\sigma$ avoids  $\wt o$, ${\sf dom}(\sigma)\cap{\sf fn}(P''_2)=\emptyset$ and
$\rho$ avoids $\wt r$, ${\sf dom}(\rho)\cap{\sf fn}(Q''_2)=\emptyset$ and
$\wt o\cap{\sf fn}(Q''_1\bnf Q''_2)=\wt r\cap{\sf fn}(P''_1\bnf P''_2)=\emptyset$, conclude 
$P \redar \res{\withsetnot{\wt m, \wt n}{\wt m\cup \wt n}}
\res{\withsetnot{\wt o, \wt r}{\wt o\cup\wt r}} (\sigma P''_1\bnf P''_2\bnf \rho Q''_1\bnf Q''_2)
\equiv \res{\withsetnot{\wt m, \wt n}{\wt m\cup \wt n}}
(\sigma(\rest o (P''_1\bnf P''_2))\bnf \rho(\rest r(Q''_1\bnf Q''_2)))
\equiv \res{\withsetnot{\wt m, \wt n}{\wt m\cup\wt n}}(\sigma P'_1\bnf \rho Q'_1) = P'$.

For the inductive step, reason on the last rule used in the inference.
\begin{itemize}
	\item If the last rule is {\sf parint} then $P = P_1\bnf P_2$,
				for $P_1 \ltsred\tau P_1'$ and $P'=P_1'\bnf P_2$.
				Apply induction to the transition $P_1 \ltsred\tau P_1'$ to obtain that $P_1\redar P_1'$;
				thus, $P \redar P'$.
	\item If the last rule is {\sf resnon} then $P = \res n P_1$, for $P_1 \ltsred\tau P_1'$ and $P'=\res n P_1'$.
				Again, conclude by induction.
	\item If the last rule is {\sf rep} then $P = \,\,!P_1$, for $P_1\bnf!P_1 \ltsred\tau P'$.
				By induction, $P_1\bnf!P_1 \redar P'$ and conclude, since $P \equiv P_1\bnf!P_1$.
\end{itemize}

\medskip\noindent
The second claim is by induction on the inference for $P \redar P'$.
The base case is when $P = p\to P'_1\bnf q\to Q'_1$ and $P' =\sigma P'_1\bnf\rho Q'_1$,
for $\{p\pmatch q\}=(\sigma,\rho)$. By the {\sf unify} rule in the LTS
        \begin{equation*}
        \prooftree (p\to P'_1) \ltsred{p}P'_1 \qquad (q\to Q'_1) \ltsred{q}Q'_1
        \justifies p\to P'_1\bnf q\to Q'_1\ \ltsred\tau\ \sigma P'_1\bnf\rho Q'_1
        \endprooftree\ \ \{p\pmatch q\}=(\sigma,\rho)
        \end{equation*}
and the result is immediate.
For the inductive step, reason on the last rule used in the inference.
\begin{itemize}
	\item If $P=P_1\bnf P_2$, where $P_1\redar P'_1$ and $P' = P_1'\bnf P_2$,
		then use the induction and exploit the {\sf parint} rule.
	\item If $P=\res n P_1$, where $P_1\redar P'_1$ and $P' = \res n P'_1$,
		then use the induction and exploit the {\sf resnon} rule.
	\item Otherwise, it must be that $P \equiv Q \redar Q' \equiv P'$. 
		By induction, $Q \ltsred\tau Q'$ for some $Q'' \equiv Q'$.
		We now have to prove that structurally equivalent processes have
		the same $\tau$-transitions, up-to $\equiv$; this is done via
		a second induction, on the inference of the judgement $P \equiv Q$.
		The following are two representative base cases; the other base cases are easier,
		as is the inductive case.
		\begin{itemize}
		\item $P =\, !R \equiv R\bnf!R = Q$: since $Q = R\bnf!R \ltsred\tau Q''$, for $Q'' \equiv Q'$,
			we can use rule {\sf rep} of the LTS and obtain $P \ltsred\tau Q''$; 
			we can conclude, since $Q'' \equiv Q' \equiv P'$.
		\item $P = \res n P_1\, \bnf\ P_2 \equiv \res n (P_1 \bnf\ P_2) = Q$, that holds since
			$n \not\in {\sf fn}(P_2)$: by the first inductive hypothesis, 
			$\res n (P_1 \bnf\ P_2) \ltsred\tau Q''$, for $Q'' \equiv Q'$. Moreover, by definition
			of the LTS, the last rule used in this inference must be {\sf resnon}; thus, 
			$P_1 \bnf\ P_2 \ltsred\tau Q'''$ and $Q'' = \res n Q'''$. There are three possible
			ways to generate the latter $\tau$-transition:
			\begin{itemize}
			\item $P_1 \ltsred\tau P_1'$ and $Q''' = P_1'\bnf P_2$: in this case
			$$
			\prooftree 
				\prooftree P_1 \ltsred\tau P_1'
				\justifies \res n P_1 \ltsred\tau \res n P_1'
				\endprooftree 
			\justifies P = \res n P_1 \bnf P_2 \ltsred\tau \res n P_1'\bnf P_2
			\endprooftree 
			$$
			and conclude by noticing that $\res n P_1'\bnf P_2 \equiv \res n (P_1'\bnf P_2) = Q'' \equiv Q' \equiv P'$.
		\item $P_2 \ltsred\tau P_2'$ and $Q''' = P_1\bnf P_2'$: this case is similar to the previous one, but simpler.
		\item $P_1 \ltsred{\rest m p} P_1'$ and $P_2 \ltsred{\rest n q} P_2'$, and $Q''' = 
			(\nu \withsetnot{\wt m,\wt n}{\wt m\cup \wt n})(\sigma P_1'\bnf \rho P_2')$, where $\{p \pmatch q\} = (\sigma, \rho)$,
			$\wt m \cap {\sf fn}(P_2) = \wt n \cap {\sf fn}(P_1) = \emptyset$ and $\wt m \cap \wt n = \emptyset$:
			this case is similar to the base case of the first claim of this Proposition and, essentially, relies
			on Lemma~\ref{lem:lts-exhibit-p}. The details are left to the interested reader.
\qed
			\end{itemize}
		\end{itemize}
\end{itemize}

\section*{Appendix B: Proofs of Section~\ref{sec:sound}}

\paragraph{Proof of Lemma~\ref{lem:bisim-case}}
It is necessary to prove that the relation
$$
\Re = \{(p \to P , p \to Q)\ :\ P \bisim Q\}\ \cup\ \bisim
$$
is a bisimulation. The only possible challenge of $p \to P$ is
$p \to P \ltsred p P$ such that ${\sf bn}(p) \cap {\sf fn}(Q) = \emptyset$;
moreover, fix any $\sigma$ such that ${\sf dom}(\sigma) = {\sf bn}(p)$.
The only possible reply from $p \to Q$ is $p \to Q \ltsred p Q$, that is
a valid reply (in the sense of Definition~\ref {def:bisim}). Indeed,
$p,\sigma \compat p,\sigma$, by Proposition~\ref{prop:compat-reflexive},
and $(\sigma P , \sigma Q) \in \Re$, because $P \bisim Q$ and 
$\bisim$ is closed under substitutions by definition.
Closure under substitution holds by definition of $\Re$.
\qed

\paragraph{Proof of Lemmata~\ref{lem:bisim-nu} and~\ref{lem:bisim-par}}
The two lemmata have to be proved together; as in $\pi$-calculus, this is necessary because of name extrusion.
We can conclude if we show that the relation
$$
\Re = \{(\rest n (P\bnf R),\rest n (Q\bnf R))\ :\ P \bisim Q\}
$$
is a bisimulation.
Fix any transition $\rest n (P\bnf R) \ltsred{\mu} \hat P$ that,
by definition of the LTS, has been inferred as follows:
$$
\hspace*{3cm}
\prooftree 
	\prooftree P\bnf R \ltsred{\bar\mu} \bar P
	\justifies \vdots
	\endprooftree
\justifies \rest n (P\bnf R) \ltsred{\mu} \hat P
\endprooftree
\hspace*{3cm}
(\star)
$$
where $\mu = \rest m \bar\mu$ and $\hat P = \res{\ \wt n\! \setminus\!\! \wt m} \bar P$ and
the dots denote repeated applications of {\sf resnon} (one for every name in $\wt n \setminus \wt m$) 
and {\sf open} (one for every name in $\wt m$).

If $\bar\mu = \tau$, then $\wt m = \emptyset$; moreover, $P\bnf R \ltsred{\bar\mu} \bar P$ can
be generated in three ways:
\begin{itemize}
\item If the transition is
	\begin{equation*}
	\prooftree P\ltsred\tau P'
	\justifies P\bnf R\ltsred\tau P'\bnf R
	\endprooftree
	\end{equation*}
	then because of $P\bisim Q$ there exists $Q\ltsred\tau Q'$
        such that $P' \bisim Q'$; hence conclude with
	$\rest n (Q\bnf R) \ltsred\tau \rest n(Q'\bnf R)$.
\item If the transition is
	\begin{equation*}
	\prooftree R\ltsred\tau R'
	\justifies P\bnf R\ltsred\tau P\bnf R'
	\endprooftree
	\end{equation*}
	consider $\rest n (Q\bnf R) \ltsred\tau \rest n(Q\bnf R')$ and conclude.
\item If the transition is
	\begin{equation*}
	\prooftree P \ltsred{\rest l p} P' \quad R \ltsred{\rest o r} R'
	\justifies P\bnf R \ltsred\tau \res{\withsetnot{\wt l,\wt o}{\wt l\cup\wt o}}(\sigma P'\bnf\theta R')
	\endprooftree
	\end{equation*}
	with $\{p \pmatch r\} = (\sigma, \theta)$ and
	$\wt l \cap {\sf fn}(R) = \wt o \cap {\sf fn}(P) = \wt l\cap\wt o = \emptyset$.
	Now, there exist $(q,\rho)$ and $Q'$ 
	such that $Q\ltsred{\rest l q}Q'$ and
	$p,\sigma\compat q,\rho$ and $\sigma P'\bisim\rho Q'$.
	By Proposition~\ref{lem:pat-lessthan}, $\{q\pmatch r\}=(\rho,\theta)$
	and so
	\begin{equation*}
	\prooftree Q \ltsred{\rest l q} Q' \quad R \ltsred{\rest o r} R'
	\justifies Q\bnf R \ltsred\tau \res{\withsetnot{\wt l,\wt o}{\wt l\cup\wt o}}(\rho Q'\bnf\theta R')
	\endprooftree
	\end{equation*}
	where, by $\alpha$-conversion, we can always let $\wt o\cap{\sf fn}(Q) = \emptyset$
	(the other side conditions for applying rule {\sf unify} already hold).
	By repeated applications of rule {\sf resnon}, infer 
	$\rest n (Q\bnf R) \ltsred\tau \rest n\res{\withsetnot{\wt l,\wt o}{\wt l\cup\wt o}}(\rho Q'\bnf\theta R')$
	and conclude.
\end{itemize}

\noindent If $\bar\mu = \rest l p$, it must be that $({\sf bn}(p) \cup \wt l) \cap {\sf fn}(\rest n(Q\bnf R)) = \emptyset$.
Then, fix any $\sigma$ such that ${\sf dom}(\sigma) = {\sf bn}(p)$ and ${\sf fn}(\sigma) \cap \wt l = \emptyset$.
The transition $P\bnf R \ltsred{\bar\mu} \bar P$ can be now generated in two ways:
\begin{itemize}
\item The transition is
	\begin{equation*}
	\prooftree P\ltsred{\rest l p} P'
	\justifies P\bnf R\ltsred{\rest l p} P'\bnf R
	\endprooftree
	\ \ (\wt l \cup {\sf bn}(p)) \cap {\sf fn}(R) = \emptyset
	\end{equation*}
	By $P\bisim Q$ there exist $(q,\rho)$ and $Q'$ 
	such that $Q\ltsred{\rest l q} Q'$ and
	$p,\sigma \compat q,\rho$ and $\sigma P' \bisim \rho Q'$. By $\alpha$-equivalence,
	let ${\sf bn}(q) \cap {\sf fn}(R) = \emptyset$; thus, $Q \bnf R \ltsred{\rest l q} Q' \bnf R$.
	By applying the same sequence of rules {\sf resnon} and {\sf open} used for $(\star)$
	(this is possible since ${\sf fn}(p) = {\sf fn}(q)$, see Lemma~\ref{prop:compat-fn}), 
	conclude with $\rest n (Q\bnf R) \ltsred{\res{\wt l,\wt m} q} \res{\ \wt n\! \setminus\!\! \wt m}(Q'\bnf R)
	= \hat Q$.
	Since ${\sf dom}(\sigma) \cap {\sf fn}(R) = {\sf bn}(p) \cap {\sf fn}(R) = \emptyset$ and
	substitution application is capture-avoiding by definition, obtain that $\sigma\hat P = 
	\sigma(\res{\ \wt n\! \setminus\! \wt m}(P'\bnf R)) = 
	\res{\ \wt n\! \setminus\!\! \wt m}(\sigma P'\bnf R)$.
	Similarly, $\rho\hat Q = \res{\ \wt n\! \setminus\!\! \wt m}(\rho Q' \bnf R)$.
	This suffices to conclude $(\sigma\hat P,\rho\hat Q) \in \Re$, as desired.
\item The transition is
	\begin{equation*}
	\prooftree R\ltsred{\rest l p} R'
	\justifies P\bnf R\ltsred{\rest l p} P\bnf R'
	\endprooftree
	\ \ (\wt l \cup {\sf bn}(p)) \cap {\sf fn}(P) = \emptyset
	\end{equation*}
	By $\alpha$-equivalence, let $(\wt l \cup {\sf bn}(p)) \cap {\sf fn}(Q) = \emptyset$;
	this allows us to infer $Q\bnf R\ltsred{\rest l p} Q\bnf R'$. The same sequence of
	rules {\sf resnon} and {\sf open} used for $(\star)$, yields
	$\rest n (Q\bnf R) \ltsred{\res{\withsetnot{\wt l,\wt m}{\wt l\cup\wt m}} p}
    \res{\ \wt n\! \setminus\! \wt m}(Q\bnf R') = \hat Q$.
	By Proposition~\ref{prop:compat-reflexive}, $p,\sigma \compat p,\sigma$. Moreover, 
	since ${\sf dom}(\sigma) \cap {\sf fn}(P,Q) = \emptyset$ and
	substitution application is capture-avoiding, obtain that $\sigma\hat P = 
	\res{\ \wt n\! \setminus\!\! \wt m}(P\bnf \sigma R')$ and
	$\sigma\hat Q = \res{\ \wt n\! \setminus\!\! \wt m}(Q \bnf \sigma R')$.
	This suffices to conclude $(\sigma\hat P,\sigma\hat Q) \in \Re$, as desired.
\end{itemize}
Closure under substitution holds by definition of $\Re$.
\qed

\paragraph{Proof of Lemma~\ref{lem:bisim-rep}}
This proof rephrases the similar one in \cite{sangiorgi.walker:theory-mobile}. 
First, define the $n$-th approximation of the bisimulation:
$$
\begin{array}{rcl}
\bisim_0 & = & Proc \times Proc
\\
\stackrel\bullet\bisim_{n+1} & = & \{(P,Q)\ : 
\\
&& \quad \forall\ P\ltsred{\mu}P'\\
&& \quad\qquad \mu=\tau\ \Rightarrow\ \exists\ Q\ltsred\tau Q'.\ (P',Q')\in\ \bisim_n\\
&& \quad\qquad \mu=\rest n p \ \Rightarrow\ \forall \sigma\ s.t.\ \,{\sf dom}(\sigma)={\sf bn}(p)\ \wedge \\
&& \qquad\qquad\qquad\qquad\qquad\qquad {\sf fn}(\sigma)\cap\wt n=\emptyset\ \wedge \\
&& \qquad\qquad\qquad\qquad\qquad\qquad\! ({\sf bn}(p)\cup\wt n)\cap{\sf fn}(Q) =\emptyset\\ 
&& \quad\qquad\qquad\qquad\qquad\ \ \exists\ (q,\rho)\mbox{ and } Q' s.t.\ Q\ltsred{\rest n q}Q' \wedge \\
&& \qquad\qquad\qquad\qquad\qquad\qquad p,\sigma\compat q,\rho
	\wedge (\sigma P',\rho Q')\in\ \bisim_n\\
&& \quad \mbox{Symmetrically for transitions of } Q \}
\\
\bisim_{n+1} & = & \mbox{the largest subrelation of $\stackrel\bullet\bisim_{n+1}$ closed under substitutions}
\end{array}
$$
Trivially, $\bisim_0\ \supseteq\ \bisim_1\ \supseteq\ \bisim_2\ \supseteq\ \cdots\,$.

We now prove that, since the LTS is structurally image finite (see Proposition~\ref{prop:imfin}),
it follows that
\begin{equation}
\label{eq:appr}
\bisim\ =\ \bigcap_{n \geq 0} \bisim_n
\end{equation}
One inclusion is trivial: by induction on $n$, it can be proved that $\bisim\ \subseteq\ \bisim_n$ for every $n$,
and so $\bisim\ \subseteq\ \bigcap_{n \geq 0} \bisim_n$.
For the converse, fix $P \ltsred\mu P'$ and consider the case for $\mu = \rest m p$, since
the case for $\mu = \tau$ can be proved like in $\pi$-calculus. For every $n \geq 0$, since $P \bisim_{n+1} Q$, 
there exist $(q_n,\rho_n)$ and $Q_n$ 
such that $Q \ltsred{\rest m q_n} Q_n$ and $p,\sigma \compat q_n,\rho_n$
and $\sigma P' \bisim_n \rho_n Q_n$. However, by Proposition~\ref{prop:maximal}, 
there are finitely many (up-to $\alpha$-equivalence) such $q_n$'s; thus, there must exist
(at least) one $q_k$ that leads to infinitely many $Q_n$'s that, because of Proposition~\ref{prop:imfin},
cannot be all different (up-to $\equiv$). Fix one of such $q_k$'s; there must exist (at least) one $Q_h$ 
such that $Q\ltsred{\rest m q_k} Q_h$ and there are infinitely many $Q_n$'s such that 
$Q\ltsred{\rest m q_k} Q_n$ and $Q_n \equiv Q_h$. Fix one of such $Q_h$'s. It suffices to prove that
$\sigma P' \bisim_n \rho_h Q_h$, for every $n$. This fact trivially holds whenever $n \leq h$:
in this case, we have that $\bisim_n\ \supseteq\ \bisim_h$. So, let $n > h$. If $Q_n \equiv Q_h$,
conclude, since $\equiv$ is closed under substitutions (notice that $\rho_n = \rho_h$ since
$q_n = q_h = q_k$) and $\equiv\ \subseteq\ \bisim_n$, for every $n$. 
Otherwise, there must exist $m > n$ such that $Q_m \equiv Q_h$ 
(otherwise there would not be infinitely many $Q_n$'s structurally equivalent to $Q_h$): 
thus, $\sigma P' \bisim_m \rho_h Q_h$ that implies $\sigma P' \bisim_n \rho_h Q_h$, since $m > n$.

\medskip

Thus, $!P \bisim\,\, !Q$ if and only if $!P \bisim_n \,!Q$, for all $n$.
Let $P^n$ denote the parallel composition of $n$ copies of the process $P$ (and similarly for $Q$).
Now, it can be proved that
\begin{equation}
\label{eq:n}
!P \bisim_n P^{2n} \quad\mbox{and}\quad !Q \bisim_n Q^{2n}
\end{equation}
The proof is by induction on $n$ and exploits a Lemma similar to Lemma~\ref{lem:bisim-par} (with $\bisim_n$ in
place of $\bisim$); the details are left to the interested reader.
By repeatedly exploiting Lemma~\ref{lem:bisim-par}, it follows that $P^{2n} \bisim Q^{2n}$
and so by \eqref{eq:appr}
\begin{equation}
\label{eq:PQ}
P^{2n} \bisim_n Q^{2n}
\end{equation}
Now by \eqref{eq:PQ} it follows that $P\bisim Q$ implies that $P^{2n}\bisim_n Q^{2n}$, for all $n$.
By \eqref{eq:n} and Lemma~\ref{lem:trans-bisim} (that also holds with $\bisim_n$ in place of $\bisim$),
it follows that $!P\bisim_n \,!Q$, for all $n$.
By \eqref{eq:appr}, conclude that $!P\bisim\,\, !Q$.
\qed

\bibliographystyle{abbrv}
\bibliography{main}

\end{document}